\documentclass[10pt]{article} 

\usepackage[utf8]{inputenc} 

\usepackage{amsthm, amsmath, amssymb, mathtools}
\usepackage[dvipsnames]{xcolor}
\usepackage[colorlinks]{hyperref}
\usepackage{comment}
\usepackage[affil-it]{authblk}

\usepackage{caption}
\usepackage{subcaption}

\usepackage{enumitem}

\usepackage{bm}

\usepackage[backend=biber, style=alphabetic, sorting=nyt, maxnames=20, maxalphanames=5]{biblatex}
\usepackage{tensor}

\usepackage{empheq}

\usepackage{faktor}

\usepackage{float}

\usepackage{cleveref}

\newtheoremstyle{break}
  {\topsep}{\topsep}%
  {\itshape}{}%
  {\bfseries}{}%
  {\newline}{}%
\theoremstyle{break}

\newtheorem*{theorem*}{Theorem}
\newtheorem{theorem}{Theorem}[section]
\newtheorem{proposition}[theorem]{Proposition}
\newtheorem{corollary}[theorem]{Corollary}
\newtheorem{remark}[theorem]{Remark}
\newtheorem{lemma}[theorem]{Lemma}
\newtheorem{definition}[theorem]{Definition}
\newtheorem{conjecture}[theorem]{Conjecture}

\numberwithin{equation}{section}

\DeclarePairedDelimiter\autobracket{(}{)}

\DeclarePairedDelimiter\abs{\lvert}{\rvert}
\newcommand{\br}[1]{\autobracket*{#1}}

\usepackage[normalem]{ulem}

\usepackage[top=53pt, bottom=53pt, left=90pt, right=90pt]{geometry} 
\usepackage{graphicx} 

\usepackage{fancyhdr} 
\hypersetup {
pdfpagemode = {UseNone},
linkcolor = NavyBlue,
citecolor = Green,
}

\title{Strong cosmic censorship for the spherically symmetric Einstein-Maxwell-charged-Klein-Gordon system with positive $\Lambda$: stability of the Cauchy horizon and $H^1$ extensions}

\author{Flavio Rossetti\thanks{\textit{E-mail address}: \texttt{flavio.rossetti@tecnico.ulisboa.pt}}}

\affil{CAMGSD, Departamento de Matem{\'a}tica, Instituto Superior T{\'e}cnico IST,
Universidade de Lisboa UL, Avenida Rovisco Pais 1, 1049-001 Lisboa, Portugal}

\newcommand{\bbR}{\mathbb{R}}

\newcommand{\bbN}{\mathbb{N}}
\newcommand{\bbS}{\mathbb{S}}
\newcommand{\Ric}{\text{Ric}}
\newcommand{\R}{\text{R}}
\newcommand{\Tem}{T^{\text{EM}}}
\newcommand{\Tphi}{T^{\phi}}

\newcommand{\uapp}{u_{\mathcal{A}}}

\newcommand{\dee}[2]{\partial_{#1} {#2}}

\renewcommand{\Re}{\operatorname{Re}}
\renewcommand{\Im}{\operatorname{Im}}

\date{} 

\begin{document}
\maketitle

\begin{abstract}
We investigate the interior of a dynamical black hole as  described by the Einstein-Maxwell-charged-Klein-Gordon system of equations with a cosmological constant, under spherical symmetry. In particular, we consider a characteristic initial value problem where, on the outgoing initial  hypersurface, interpreted as the event horizon $\mathcal{H}^+$ of a dynamical black hole, we prescribe: a) initial data asymptotically approaching a fixed sub-extremal Reissner-Nordstr{\"o}m-de Sitter solution; b) an exponential Price law upper bound for the charged scalar field.

After showing local well-posedness for the corresponding first-order system of partial differential equations, we establish the existence of a Cauchy horizon $\mathcal{CH}^+$ for the evolved spacetime, extending the  bootstrap methods used in the case $\Lambda = 0$ by Van de Moortel \cite{VdM1}.  In this context, we show the existence of $C^0$ spacetime extensions beyond $\mathcal{CH}^+$. Moreover, if the scalar field decays at a  sufficiently fast rate along $\mathcal{H}^+$, we show that the renormalized Hawking mass remains bounded for a large set of initial data. With respect to the analogous model concerning an uncharged and massless scalar field, we are able to  extend the known range of parameters for which  mass inflation is prevented, up to the optimal threshold  suggested by the linear analyses by Costa-Franzen \cite{CostaFranzen} and Hintz-Vasy \cite{HintzVasy1}.

 In this no-mass-inflation scenario, which includes near-extremal solutions, we further prove that the spacetime can be extended across the Cauchy horizon with continuous metric,  Christoffel symbols in $L^2_{\text{loc}}$ and  scalar field in $H^1_{\text{loc}}$.

By generalizing the work by Costa-Girão-Natário-Silva \cite{CGNS4} to the case of a charged and massive scalar field, our results reveal a potential failure of the Christodoulou-Chru{\'s}ciel version of the strong cosmic censorship under spherical symmetry.
\end{abstract}

\tableofcontents

\section{Introduction} 

\subsection{An overview of the model}

The present work concerns the system of equations describing a spherically symmetric Einstein-Maxwell-charged-Klein-Gordon model with positive cosmological constant $\Lambda$, namely:
\begin{equation} \label{main_system}
\begin{cases}
\Ric(g)-\frac{1}{2}\R(g) \,g + \Lambda g = 2 \br{\Tem + \Tphi }, &\textit{\normalfont (Einstein's equations)}\\[0.3em]
dF=0, \quad d \star F = \star J, &\textit{\normalfont (Maxwell's equations)}\\[0.3em]
\br{D_{\mu}D^{\mu}-m^2} \phi = 0, &\textit{\normalfont (Klein-Gordon equation)}
\end{cases}
\end{equation}
where 
\[
\Tem_{\mu \nu} = g^{\alpha \beta}F_{\alpha \mu} F_{\beta \nu} - \frac{1}{4} g_{\mu \nu} F^{\alpha \beta}F_{\alpha \beta}
\]
is the energy-momentum tensor of the electromagnetic field,
\[
\Tphi_{\mu \nu} = \Re \br{D_{\mu} \phi \overline{D_{\nu} \phi }} - \frac{1}{2}  g_{\mu \nu} \br{D_{\alpha} \phi \overline{D^{\alpha}\phi} + m^2 |\phi|^2 }
\]
is the energy-momentum tensor of the charged scalar field and
\[
J_{\mu}= -\frac{i q}{2} \br{\phi \overline{D_{\mu} \phi} - \overline{\phi} D_{\mu} \phi}
\]
is the current appearing in Maxwell's equations.
The above system will be presented in full detail in section \ref{section:so3}.

The main objective of our work is to study the evolution  problem associated to the partial differential equations (\textbf{PDE}s) \eqref{main_system}, with respect to initial data prescribed on two transversal  characteristic hypersurfaces. Along one of these hypersurfaces, the  initial data are specified to mimic the expected behaviour of a charged, spherically symmetric black hole,  approaching\footnote{In an appropriate sense, which will be specified rigorously during our study.} a sub-extremal Reissner-Nordstr{\"o}m-de Sitter solution and interacting with a charged and (possibly) massive scalar field. This scenario is essentially encoded in the requirement of an exponential decay for such a scalar field along the event horizon in terms of an Eddington-Finkelstein type of coordinate, i.e.\ an \textit{exponential Price law upper bound}. The evolved data then describe the interior of a dynamical black hole. A key point of our study regards the existence of a Cauchy horizon and of metric extensions beyond it, in different classes of regularity and most  notably in $H^1_{\text{loc}}$. An overview of our results is given in section \ref{subsec:main_results}.

This problem arises naturally as a consequence of a decades-long debate on the mathematical treatment of black hole spacetimes. Already in the case $\Lambda = 0$, it was observed that the celebrated Kerr and Reissner-Nordstr{\"o}m solutions to the Einstein equations 
share a daunting property: the future domain of dependence\footnote{We define the future domain of dependence of a spacelike hypersurface $S$ as the set of points $p$ such that every past-directed, past-inextendible causal curve starting at $p$ intersects $S$, see e.g. \cite{NatarioRel}.} of complete, regular, asymptotically flat, spacelike hypersurfaces is future-extendible in a smooth way. This translates into a failure of classical determinism\footnote{The lack of global uniqueness can occur without any loss of regularity of the solution metric, even though the Einstein equations are hyperbolic, up to their diffeomorphism invariance. See \cite{Dafermos_2003} for a discussion of this phenomenon.} which is in sharp contrast with the Schwarzschild case, where analogous extensions are forbidden even in the continuous class \cite{Sbierski}. On the other hand, observers crossing 
a Cauchy horizon\footnote{That is the boundary of the future domain of dependence of a complete Cauchy hypersurface in the extended spacetime manifold. We refer the reader to the introduction of \cite{CGNS4} for a discussion on the topic and to \cite{ChruscIsen93} for constructions of specific non-isometric extensions.} are generally expected to experience a \textit{blueshift instability} \cite{PenroseBattelle, SimpsonPenrose}, leading to the blow-up of dynamical quantities  and thus rescuing the deterministic principle. The instability is related to an unbounded amount of energy associated to test waves propagating in such a universe, as measured along the Cauchy horizon.\footnote{At the level of geometric optics,  \cite{PenroseBattelle} asserts: ``\textit{As [an observer crossing the Cauchy horizon $\mathcal{CH}^+$] looks out at the universe that he is `leaving behind', he sees, in one final flash, as he crosses $\mathcal{CH}^+$, the entire later history of the rest of his `old universe}' ''. See also Fig.\ \ref{fig:blueshiftinst}.} 

This led to the first formulation of the strong cosmic censorship conjecture (\textbf{SCCC}), which, roughly speaking, forbids the existence of metric extensions beyond the Cauchy horizon of black hole spacetimes, at least in some regularity class.

In fact, the regularity of such extensions measures to which extent the SCCC fails. Although the Einstein equations are PDEs in the second derivative of the metric tensor, the $H^1_{\text{loc}}$ (Christodoulou-Chru{\'s}ciel) regularity is actually the minimal requirement to make sense of the extensions as weak solutions to the Einstein equations (see also \cite{CGNS4, DafYak18, ChrFormation, ChruscielSCCC}), even though local well--posedness is not known to hold in this case \cite{BoundedL2}.  Weaker formulations can also be considered  \cite{Reintjes1}. 

Being related to an exponential blow-up of the \textbf{energy}, the linear instability associated to blueshift concerns the first derivatives of scalar perturbations. Such effect becomes weaker in the near-extremal limit, for which the surface gravity of the Cauchy horizon approaches zero. This instability is absent in the extremal case \cite{GajicLinear1, GajicLinear2, GajicLuk}.

Moreover, dispersive effects in the black hole exterior also impact the picture described by the blueshift instability. Scalar perturbations are expected to decay along the event horizon of black holes at a rate which, essentially, is exponential in the case $\Lambda > 0$ \cite{DyatlovLinear, DafRodSchwdS, fang2022linear, fang2022nonlinear,  HintzVasyStability, Mavro}, inverse polynomial in the case $\Lambda = 0$ (see \cite{Price72, GundPricePull} for heuristics and numerical arguments, see \cite{DafRodPriceLaw, AngArGajPrice, AngArGajAsymp, DonningerSchlagSoffer1, DonningerSchlagSoffer2, Tataru, MetcalfeTataruTohaneanu, HintzPrice} for rigorous results) and logarithmic for $\Lambda < 0$ \cite{HolzSmulAdS, HolzSmulAdS2}. This behaviour is known as  \textbf{Price's law}.

The case $\Lambda > 0$ suggests  a scenario where the exponential blow-up of the energy due to blueshift (in the sub-extremal case) can be counterbalanced by an exponentially fast decay in the black hole exterior (at least near extremality), leading to extensions of $H^1$ regularity at the Cauchy horizon. While this effect has been described rigorously in \cite{CostaFranzen, HintzVasy1} for the linear wave equation on a Reissner-Nordstr{\"o}m-de Sitter background, the present work proves the analogous result in a non-linear setting with charged scalar field (see already section \ref{subsec:main_results} and see \cite{CGNS4} for the non-linear and uncharged case). 
\begin{figure}[h]
\centering
\includegraphics[width=0.8\textwidth]{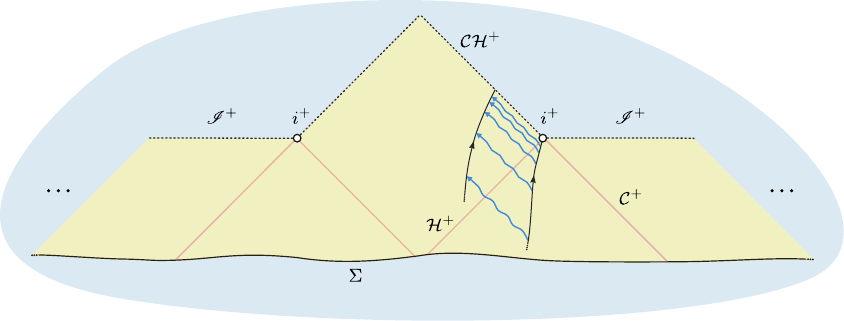}
\caption{A schematic depiction of the blueshift instability in the case, for instance, of a sub-extremal Reissner-Nordstr{\"o}m-de Sitter black hole (although we stress that this phenomenon is not specific to the case $\Lambda > 0$). Assume that radiation pulses at constant time intervals, as measured in proper time, from a timelike curve in the exterior of the black hole. The wordline in the exterior of the black hole has infinite affine length. On the other hand, observers along the timelike worldline in the black hole interior can reach the Cauchy horizon $\mathcal{CH}^+$ in finite proper time, thus receiving an infinite amount of pulses  at the moment of crossing $\mathcal{CH}^+$.}
\label{fig:blueshiftinst}
\end{figure} 

 The current understanding of non-linear effects is influenced by the celebrated work by Israel and Poisson on the Einstein-Maxwell-null dust system \cite{IsraelPoisson}. A subsequent turning point  was the adoption of the Einstein-Maxwell-(real) scalar field model to study black hole interiors \cite{Dafermos_2003}. In particular, the addition of a  scalar field not only renders the problem non-trivial (already under spherical symmetry) but is also motivated by the hyperbolic properties expected from the Einstein equations. A review of this construction can be found in the introduction of \cite{DafermosLuk}. In the latter, the analysis of the non-linear problem in the absence of spherical symmetry was initiated. We will review relevant past work in section \ref{subsec:review_work}.

The present article is strictly related to the following formulations of the SCCC for the non-linear system \eqref{main_system}.
\begin{conjecture}[$C^0$ formulation of the SCCC for the Einstein-Maxwell-charged-Klein-Gordon system with $\Lambda > 0$, under spherical symmetry] \label{conjcontinuous}
Let $\Sigma$ be a regular, compact (without boundary),\footnote{The case of $S^3$ is often considered in the literature. Non-compact, ``asymptotically de Sitter'' hypersurfaces can also be examined. See also Fig.\ \ref{fig:global_structure} for an overview on potential global structures.} spacelike hypersurface. Let us prescribe generic (in the spherically symmetric class) initial data on $\Sigma$, for the Einstein-Maxwell-charged-Klein-Gordon system with $\Lambda > 0$. Then, the future maximal globally hyperbolic development\footnote{Which is, roughly speaking, the largest Lorentzian manifold determined by the evolution of Cauchy data via the Einstein equations. See \cite{RingstromCauchy} for a precise definition.} of such initial data cannot be locally extended as a Lorentzian manifold endowed with a continuous metric.
\end{conjecture}

\begin{conjecture}[$H^1_{\text{loc}}$ formulation of the SCCC for the Einstein-Maxwell-charged-Klein-Gordon system with $\Lambda > 0$, under spherical symmetry] \label{conjH1}
Let $\Sigma$ be a  regular, compact (without boundary), spacelike hypersurface. Let us prescribe generic (in the spherically symmetric class) initial data on $\Sigma$, for the Einstein-Maxwell-charged-Klein-Gordon system with $\Lambda > 0$. Then, the future maximal globally hyperbolic development of such initial data cannot be locally extended as a Lorentzian manifold endowed with a continuous metric and Christoffel symbols in $H^1_{\text{loc}}$.
\end{conjecture}

In particular, this work contributes to fill a gap in the literature by suggesting a negative resolution to  conjecture \ref{conjcontinuous} and, more significantly, a negative resolution, for a large subset of the parameter space, to conjecture \ref{conjH1}. Such resolutions become definitive once the initial conditions of our system are shown to arise from suitable Cauchy data (although we do not deal with this problem in the current work,   we explain in section \ref{sec:SCCCsubtle} why we expect such results  to hold in the black hole exterior).

Being a natural generalization of \cite{CGNS4} (but notice that the coupled scalar field is now \textit{charged} and possibly \textit{massive}) and \cite{VdM1} (in our case, however, $\Lambda > 0$, thus the Price law is exponential rather than inverse polynomial), we take these works as starting points for our analysis. In section \ref{subsec:review_work} we review previous relevant work and compare the  techniques we employed with those of \cite{CGNS4} and \cite{VdM1}.

\subsection{Summary of the main results} \label{subsec:main_results}
The main conclusions of our work can be outlined as follows. We first discuss the characteristic initial value problem (\textbf{IVP}) for the first-order system related to \eqref{main_system} (this corresponds to equations \eqref{dur}--\eqref{algconstr}) and discuss its local well-posedness, modulo diffeomorphism invariance (see theorem \ref{thm:localexistence} and proposition \ref{prop: contdependence}). 

Subsequently, we show the equivalence between the first-order and the second-order Einstein-Maxwell-charged-Klein-Gordon systems (see remark \ref{rmk:equivalence}) under additional regularity of the initial data, which we assume for the remaining part of our work. 
Finally, we assess the global uniqueness of the dynamical black hole under investigation, and its consequences for the SCCC, in the following classes of regularity.

\begin{theorem*}[stability of the Cauchy horizon and existence of continuous extensions. See corollary \ref{coroll: cauchyhor}, theorem \ref{thm:general_C0_extension}]
Let $Q_+$, $\varpi_+$ and $\Lambda$ be, respectively, the charge, mass and cosmological constant associated to a fixed sub-extremal Reissner-Nordstr{\"o}m-de Sitter black hole. 
Let us consider, with respect to the first-order system  \eqref{dur}--\eqref{algconstr}, the future maximal globally hyperbolic development $(\mathcal{M}, g, F, \phi)$ of spherically symmetric initial data prescribed on two transversal null hypersurfaces, where $q \in \mathbb{R} \setminus \{0\}$ is the charge of the scalar field $\phi$ and $m \ge 0$ its mass.

In particular, given $\mathcal{Q} = \mathcal{M}/SO(3)$, let $(u, v)$ be a system of null coordinates on $\mathcal{Q}$,  determined\footnote{When the initial data of the characteristic IVP coincide with Reissner-Nordstr{\"o}m-de Sitter initial data, the coordinate $v$ is equal to the Eddington-Finkelstein coordinate $\frac{r^* + t}{2}$ used near the black hole event horizon. See appendix  \ref{app:eddfinkel} for more details.} by  \eqref{assumption:nu} and \eqref{assumption:kappa} and such that
\begin{equation} \label{prelim_metric}
g = -\Omega^2(u, v)du dv + r^2(u, v) \sigma_{S^2},
\end{equation}
where $\sigma_{S^2}$ is the standard metric on the unit round sphere.  For some $U > 0$ and $v_0 > 0$, we express the two initial hypersurfaces in the above null coordinate system as $[0, U] \times \{v_0\} \cup \{0\} \times [v_0, +\infty)$.

Assume that the initial data to the characteristic IVP satisfy \hyperref[section:assumptions]{assumptions} \textbf{(A)}--\textbf{(G)}, which require, in particular, that the initial data asymptotically approach those of the afore-mentioned sub-extremal black-hole, and that 
 an exponential Price law upper bound holds: there exists $s > 0$ and $C > 0$ such that
\begin{equation} \label{prelim_Pricelaw}
|\phi|(0, v) + |\partial_v \phi|(0, v) \le C e^{-sv}.
\end{equation}
Then, if $U$ is sufficiently small,\footnote{Compared to the $L^{\infty}$ norm of the initial data.} there exists a unique solution to this characteristic IVP, defined on $[0, U] \times [v_0, +\infty)$, given by the metric in \eqref{prelim_metric} and such that
\[
u \mapsto \lim_{v \to +\infty} r(u, v) \ge \mathfrak{r} > 0
\]
is a continuous function from $[0, U]$ to $[\mathfrak{r}, r_{-}]$, for some $\mathfrak{r} > 0$, and
\[
\lim_{u \to 0} \lim_{v \to +\infty} r(u, v) = r_{-},
\]
where $r_{-}$ is the radius of the Cauchy horizon of the reference sub-extremal Reissner-Nordstr{\"o}m-de Sitter black hole.

Moreover, $(\mathcal{M}, g, F, \phi)$ can be   extended up to the Cauchy horizon $\mathcal{CH}^+ = \{ v = +\infty \}$ with continuous metric and continuous scalar field $\phi$. 
\end{theorem*}
This \textbf{suggests a negative resolution to  conjecture \ref{conjcontinuous}} (see also section \ref{sec:SCCCsubtle}), for every non-zero value of the scalar field charge and for every non-negative value of its mass. This resolution becomes definitive if  \eqref{prelim_Pricelaw} is established for solutions evolved from regular Cauchy data.

\begin{theorem*}[no-mass-inflation scenario and $H^1$ extensions. See theorems \ref{thm:no_mass_inflation} and \ref{thm:H1}]
Under the assumptions of the previous theorem, let $s$ be the constant in \eqref{prelim_Pricelaw} and consider
\[
\rho \coloneqq \frac{K_{-}}{K_+} > 1,
\]
where $K_{-}$ and $K_+$ are the absolute values of the surface gravities of the Cauchy horizon and of the event horizon, respectively, of the reference sub-extremal Reissner-Nordstr{\"o}m-de Sitter black hole. We define the renormalized Hawking mass
\[
\varpi  = \frac{r}{2}\br{1- g(d r^{\sharp}, d r^{\sharp})} + \frac{Q^2}{2r} - \frac{\Lambda}{6}r^3,
\]
where $Q$ is the charge function of the dynamical black hole and $\Lambda$ is the cosmological constant.
If:
\begin{equation} \label{condition_srho}
s > K_{-} \quad \text{and} \quad \rho < 2,
\end{equation}
then the geometric quantity $\varpi$ is bounded, namely there exists $C > 0$, depending uniquely on the initial data, such that:
\[
\|\varpi \|_{L^{\infty}([0, U] \times [v_0, +\infty))} = C < +\infty.
\]
Moreover, for such values of $s$ and $\rho$, there exists a coordinate system for which $g$ can be extended continuously up to $\mathcal{CH}^+$, with Christoffel symbols in $L^2_{\text{loc}}$ and $\phi$ in $ H^1_{\text{loc}}$.
\end{theorem*}
When the decay given by the exponential Price law upper bound is sufficiently fast and the parameters of the reference sub-extremal black hole lie in a certain range,\footnote{This includes the case in which the reference black hole is close to being extremal.} \textbf{the above result  puts into question the validity of conjecture \ref{conjH1}} (see also section \ref{sec:SCCCsubtle}).

As shown in Fig.\ \ref{fig:no_mass_inflation1}, this theorem significantly extends the range of parameters for which $H^1$ extensions can be obtained, when comparing to the case of a massless and uncharged scalar field \cite{CGNS4}, in a way which is expected to be sharp in the $(s, \rho)$ parameter space \cite{CostaFranzen, HintzVasy1}. We stress that the range of parameters allowing for $H^1$ extensions is not expected to depend on the charge of the scalar field. On the other hand, the techniques of the present paper are more general than those previously used in the case of an uncharged scalar field, when $\Lambda > 0$. 

Notice that the first condition in \eqref{condition_srho} depends on the choice of the outgoing null coordinate $v$, see also appendix \ref{app:eddfinkel}.

\begin{figure}[H]
\centering
\includegraphics[width=0.6\textwidth]{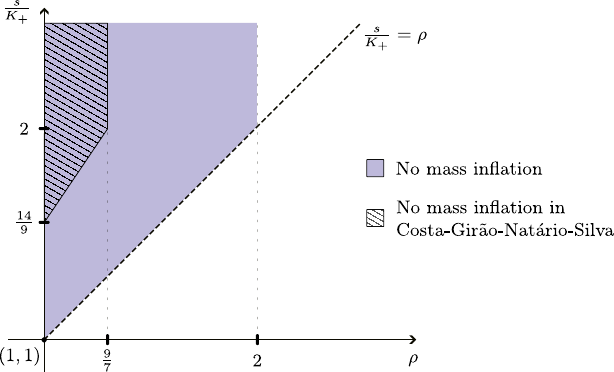}
\caption{Values of $s$ and $\rho = \frac{K_{-}}{K_+}$  that guarantee the absence of mass inflation for a generic set of initial data, in comparison to the results in \cite{CGNS4}. In the extended range that we obtained,  $H^1_{\text{loc}}$ extensions of the solutions to the IVP are constructed. This extended range includes near-extremal black holes, for which $\rho$ is close to 1. The value of $s$ is expected to depend in a non-trivial way on the parameters of the fixed sub-extremal black hole, see \cite{QNMandSCCC, SCCCstill_subtle, hintz2021mode, HintzXie, CasalsMarinho}.} 
\label{fig:no_mass_inflation1}
\end{figure}

\subsection{Strong cosmic censorship in full generality} \label{sec:SCCCsubtle}
The results of our analysis show that the main conclusion of \cite{CGNS4} is still valid when we add  a charge (and possibly mass) to the coupled scalar field: \textbf{according to the state of the art of SCCC, the potential scenarios that lead to a violation of  determinism in the case $\bm{\Lambda > 0}$ cannot be ruled out}.

To explain how to reach this conclusion, we first illustrate the connection between our results and conjecture \ref{conjH1}, and then the relation between the latter and other formulations of SCCC.

The results presented in section \ref{subsec:main_results}
stem from the evolution of initial data from two transversal \textit{null} hypersurfaces, one of which  modelling the event horizon of a dynamical black hole. This suggests that the black hole we take under  consideration has already formed at the beginning of the evolution.\footnote{Notice that the interior and exterior problems can be treated separately, due to domain of dependence arguments, and the exterior region in this context was proved to be stable in \cite{Hintz}, with respect to gravitational and electromagnetic perturbations.}  

Strictly speaking, however, the SCCC constrains the admissible future developments of \textit{spacelike} hypersurfaces, on which \textit{generic}\footnote{By ``generic'' we mean  open and dense  in the moduli space of initial
data, in a suitable topology (see, e.g. \cite{RingstromAnn} and \cite{LukOh1}). This can also be
interpreted in terms of a ``finite co-dimension condition'' (a related discussion can be found in \cite{LukOh1}
for the non-linear case and in \cite{DafYak18} in a linear setting). Genericity of the set of black
hole parameters plays a prominent role in \cite{KehleDioph}, for a negative cosmological constant.} initial data are specified. Although \textbf{we are not committed to any specific global structure}, it is useful to think of the characteristic IVP of section \ref{section:ivp} as arising from suitable (compact or non-compact) initial data prescribed in the black hole exterior (see also Fig.\ \ref{fig:global_structure}). The reason why we believe that our results are relevant to the SCCC is that the initial data we prescribe for our characteristic IVP are compatible with numerical solutions obtained in the exterior of Reissner-Nordstr{\"o}m-de Sitter black holes (see the paragraph below).  Such results are also expected to hold for perturbations of such black holes, due to the non-linear stability results \cite{Hintz, HintzVasyStability}.

\begin{figure}[h]
     \centering
     \includegraphics[width=\linewidth]{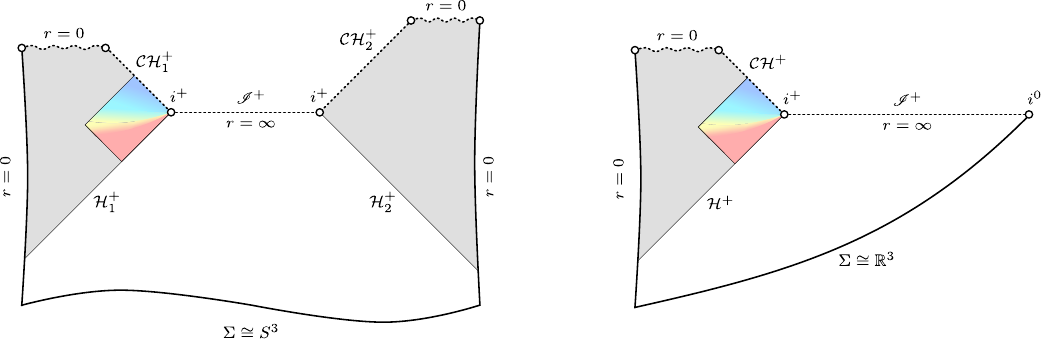}
     \caption{Two possible global spacetime structures  in which the characteristic IVP of section \ref{section:ivp} fits in. The coloured region is the one depicted in the forthcoming Fig.\ \ref{fig:bh_interior}, namely the black hole interior region under inspection in the present work. These global structures require a dynamical formation of trapped surfaces, starting from non-trapped initial data (see also \cite{CostaTrapped, AnTrapped}).}
     \label{fig:global_structure}
\end{figure}

In particular, our $v$ coordinate can be compared to half of the outgoing Eddington-Finkelstein coordinate of \cite{SCCCstill_subtle} (see appendix \ref{app:eddfinkel}). This can be used to see that, according to the same article,  condition \eqref{condition_srho} can be (numerically) attained at the linear level: $s > K_{-}$ (i.e.\ $\beta = \frac{s}{2 K_{-}} > \frac12$, in the language of \cite{SCCCstill_subtle})  for linear waves propagating in the exterior of near-extremal (i.e.\ $\rho$ close to 1) Reissner-Nordstr{\"o}m-de Sitter black holes. This result is also expected to hold on dynamical black holes sufficiently close to Reissner-Nordstr{\"o}m-de Sitter, as the one we analyse. Indeed, notice that on a fixed background $s$ is equal, up to a multiplying factor, to the spectral gap of the Laplace-Beltrami operator $\square_g$. If we denote by $\text{QNM}$ the set of the quasinormal modes of the reference black hole, the spectral gap can also be expressed by 
\[
\inf_{\substack{\omega \in \text{QNM}\\ \Im(\omega) < 0 }} \br{- \Im (\omega)}.
\]
In particular, $s$ depends on the black hole parameters. At the same time, \cite{Hintz} shows that, under small perturbations, the parameters of the perturbed solution are sufficiently close to the initial ones.

Therefore, once the numerical results of \cite{SCCCstill_subtle} are verified rigorously and extended to the non-linear case, a mathematical proof for the validity of \eqref{condition_srho} would in principle be possible, leading to a negative resolution of conjecture \ref{conjH1}.
 
 On the other hand, conjectures \ref{conjcontinuous} and \ref{conjH1} regard  initial data which are generic in the spherically symmetric class, but \textit{non-generic} in the much larger \textit{moduli space} of all possible initial configurations. In particular, when we talk about the \textit{SCCC under spherical symmetry}, the validity or failure of such a conjecture does not automatically settle the non-spherically symmetric problem. Nonetheless, due to analogies with the conformal structure of a Kerr-Newman spacetime, it is widely believed \cite{Dafermos_2003, Dafermos_2005_uniqueness, CGNS4, DafermosLuk} (and there is substantial numerical evidence in the case $\Lambda > 0$ \cite{CasalsMarinho, DaveyDiasGil})  that results obtained in the charged spherically symmetric context can provide vital evidence to either uphold or refute formulations of the SCCC in absence of spherical symmetry.\footnote{The case of uncharged black holes, however, needs to be inspected separately. See the numerical study \cite{DiasReallSantos} for Kerr-de Sitter.  In particular, the $H^1$ formulation of the SCCC is expected to be violated in Kerr-Newman-de Sitter black holes for parameters sufficiently close to those of a near-extremal Reissner-Nordstr{\"o}m-de Sitter solution \cite{DaveyDiasGil}.} 

A larger moduli space can also be achieved by working in a  rougher class of initial data. In \cite{DafYak18}, where global uniqueness is investigated at the linear level, the following is proved: given initial data $(\psi_0, \psi_0') \in  H^1_{\text{loc}} \times L^2_{\text{loc}}$ to the linear wave equation $\square_g \psi = 0$ on a sub-extremal Reissner-Nordstr{\"o}m-de Sitter black hole, with data prescribed on a complete spacelike hypersurface, $\psi$ cannot be extended across the Cauchy horizon in $H^1_{\text{loc}}$ (see also the numerical work \cite{Reall+Rough}). 
This means that, if the numerical results \cite{QNMandSCCC} obtained in the exterior of Reissner-Nordstr{\"o}m-de Sitter black holes were to be proved rigorously, the initial data leading to a potential failure of the SCCC could be ignored, because they are non-generic in this larger set of rough initial configurations.
Whether this Sobolev level of regularity for the initial data is ``suitable'' may well depend on the problem to be studied and is still a topic of discussion in the literature. We also refer  the reader to  \cite{DafShlapScattering, McNamara78} to compare with instability results in the $\Lambda = 0$ case.

\subsection{Previous works and outline of the bootstrap method} \label{subsec:review_work}

\textbf{The case $\bm{\Lambda = 0}$}: after foundational work by Christodoulou \cite{ChristodoulouIVP}, Dafermos  \cite{Dafermos_2003} first laid the framework for the  analysis of the SCCC in spherical symmetry via modern PDE methods. 
In \cite{Dafermos_2005_uniqueness},  the $C^0$ version of the SCCC was proved to be false for the spherically symmetric Einstein-Maxwell-(real) scalar field system, conditionally on the validity of a Price law upper bound, which was later seen to hold in \cite{DafRodPriceLaw}.  
The work \cite{LukOh1} proved the validity of the $C^2$ formulation of SCCC for the same spherically symmetric system.
 
 The first steps towards the analysis of a complex-valued  (i.e.\ \textit{charged}) and possibly massive scalar field were taken in \cite{Kommemi}: this is the model that we denoted as the Einstein-Maxwell-charged-Klein-Gordon system. The work \cite{Kommemi}, in particular, presented a soft analysis of the future boundary of spacetimes undergoing spherical collapse.
 More recent work on the SCCC for the same system can be found in \cite{VdM1, VdM2, VdM_mass, KehleVdM1}, where  continuous spacetime extensions were constructed and several instability results obtained.  
 
 The analysis of the non-linear problem without symmetry assumptions and in vacuum, namely the validity of the SCCC for a Kerr spacetime, was initiated in the seminal work \cite{DafermosLuk}. Here, the $C^0$ formulation of the SCCC was shown to be false, under the assumption of the quantitative stability of the Kerr exterior. The problem is indeed intertwined with stability results of black hole exteriors, object of an intense activity that built on a sequence of remarkable results and notably led, in recent times, to   \cite{GiorgiKlainermanSzeftelKerr, klainkerr, klainbrief} for slowly rotating Kerr solutions. 
Further celebrated results are the non-linear stability of the Schwarzschild family \cite{DafHolzRodTaylor21}. For an overview on the problem and recent contributions both in the linear and non-linear realm, we refer the reader to  \cite{DHR3, AnderssonBlue, SRTdC1, SRTdC2, GiorgiLinear, Benomio1} and references therein.
See also \cite{Aretakis1, AngArGajAsymp} for results in the extremal case.

\textbf{The case $\bm{\Lambda > 0}$}: Differently from the asymptotically flat case, Reissner-Nordstr{\"o}m-de Sitter black holes present a spacelike future null infinity and an additional Killing horizon $\mathcal{C}^+$ (the cosmological horizon). As a result, scalar perturbations of such  black holes  decay  \textit{exponentially} fast along $\mathcal{H}^+$ \cite{DyatlovLinear, Mavro, fang2022linear}. This behaviour competes with the blueshift effect expected near the Cauchy horizon and determining a growth, also exponential, of the main perturbed quantities. While the $H^1$ version of the SCCC is expected to prevail in the asymptotically flat case (the exponential contribution from blueshift wins over the inverse polynomial tails propagating from the event horizon), the situation drastically changes when $\Lambda > 0$. In the latter case, indeed, knowing the exact asymptotic rates of scalar perturbations along the event horizon, and thus determining which of the two exponential contributions is the leading one, is of vital importance to determine the admissibility of spacetime extensions.

So far, a quantitative description of the exponents for such a decay is available at the linear level via numerical methods.
In particular, it has been shown numerically that the $H^1$ norm (i.e. the energy) of linear waves propagating in sub-extremal Reissner-Nordstr{\"o}m-de Sitter black holes is bounded up to the Cauchy horizon, provided that such black hole backgrounds are close to being extremal \cite{QNMandSCCC, SCCCstill_subtle, MoTian} (see also the rigorous results \cite{hintz2021mode} and \cite{HintzXie} in the small mass limit). We refer to  \cite{CasalsMarinho, DaveyDiasGil} for the analogous stability result for scalar perturbations of Kerr-Newman-de Sitter. A different outcome was found in the Kerr-de Sitter case \cite{DiasReallSantos}, where numerical results showed that the  $H^1_{\text{loc}}$ version of the SCCC is respected in the linear setting.

On the other hand, the SCCC with $\Lambda > 0$ has been rigorously studied via the Einstein-Maxwell-(real) scalar field model in \cite{CGNS1, CGNS2, CGNS3} (compare with \cite{Dafermos_2003}). One of the main novelties of this series of papers is a partition of the black hole interior in terms of level sets of the radius function, rather than curves of constant shifts as in \cite{Dafermos_2003, Dafermos_2005_uniqueness} (the latter curves fail to be spacelike when $\Lambda \ne 0$). These three works deal with a real scalar field, taken identically zero along the  event horizon. Conversely, an exponential Price law (both a lower bound and an upper bound) was prescribed in \cite{CGNS4} along the event horizon (compare with \cite{Dafermos_2005_uniqueness}), and generic initial data leading to $C^0$ extensions were constructed. Conditions to have either $H^1_{\text{loc}}$ extensions or mass inflation were given in terms of the decay of the scalar field along the event horizon.

The work that we present in the next sections is a natural prosecution of \cite{CGNS4} as it replaces the neutral scalar field with a charged one, whose behaviour along the event horizon is dictated by an exponential Price law upper bound. We extend, in a way which is expected to be optimal, the set of parameters $s$ and $\rho$ that allow for $H^1_{\text{loc}}$ extensions.

Again, the problem is related with stability results of black hole exteriors: see \cite{Hintz} for the non-linear stability of slowly  rotating Kerr-Newman-de Sitter black holes.

 \textbf{The case $\bm{\Lambda < 0}$}: recent progress with the investigation of the SCCC in the case $\Lambda < 0$ can be found, for instance, in  \cite{HolzSmulAdS, HolzSmulAdS2, KehleDioph, KehleLinear}.

\textbf{Comparison with \cite{VdM1} and \cite{CGNS4}}: the present work generalizes the results of \cite{CGNS4} to the case of a charged and (possibly) massive scalar field, without requiring any lower bound\footnote{An (integrated or pointwise) lower bound for the scalar field is generally expected in order to prove instability estimates, such as mass inflation. We will not pursue this direction in our study, even though the estimates we obtain in the black hole interior pave the way to a future analysis of instability results.} for it along the initial outgoing characteristic hypersurface. Differently from the neutral case, a charged scalar field is compatible with a more realistic description of gravitational collapse (see Fig.\ \ref{fig:global_structure} for potential global structures). At the same time, the charged setting requires gauge covariant derivatives, brings additional terms in the equations and generally undermines monotonicity properties due to the fact that the scalar field is now complex-valued. With respect to the real case, additional dynamical quantities, such as the electromagnetic potential $A_u$ and the charge function $Q$, need to be controlled during the evolution. Moreover, the presence of a mass term for the scalar field allows to describe compact objects formed by massive spin-0 bosons.

In particular, we replace many monotonicity arguments exploited in \cite{CGNS4} with bootstrap methods. In this sense, our work follows the spirit of \cite{VdM1}, where the stability of the Cauchy horizon was proved by bootstrapping the main estimates from the event horizon. In our case, however, we need to propagate exponential terms of the form $e^{c(s)v}$, with $c(s) > 0$ and $v$ being an Eddington-Finkelstein type of coordinate. 

\begin{figure}[h]
\centering
\includegraphics[width=0.5\textwidth]{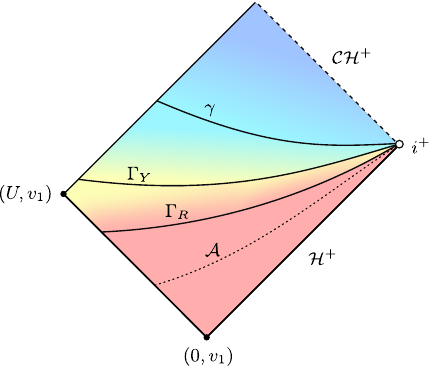}
\caption{A posteriori, these are the spacelike curves that partition the interior of the dynamical black hole for $v_1 \ge v_0$ large. Compare with Fig.\ 3 in \cite{VdM1} and with the figures in \cite[section 4]{CGNS4}. 
For the construction, we consider $r_{-}$ and $r_+$, with $0 < r_{-} < r_+$, respectively the radius of the Cauchy horizon and the radius of the event horizon of the reference sub-extremal black hole. We fix $R < r_+$, close to $r_+$, and $Y > r_{-}$ close to $r_{-}$.
The region to the past of $\Gamma_R = \{r = R\}$ is the redshift region (the curve $\mathcal{A}$ is the apparent horizon). Between $\Gamma_R$ and $\Gamma_Y = \{r = Y\}$ is the no-shift region. Bounded between $\Gamma_Y$ and $\gamma$ (which is not a level set of $r$) is the early blueshift region. Finally, to the future of $\gamma$ we have the late blueshift region.}
\label{fig:bh_interior}
\end{figure} 

We now summarize the main technical strategies that we adopt in the paper (see also Fig.\ \ref{fig:bh_interior} for the partition of the black hole interior that we adopt) and how they differ from the techniques of \cite{CGNS4} and \cite{VdM1}. Here, the quantities $r_{-}$ and $r_+$ represent the radii of the Cauchy horizon and of the event horizon, respectively, of the fixed sub-extremal Reissner-Nordstr{\"o}m-de Sitter black hole. We recall that in section \ref{subsec:main_results} we presented the coordinates $u = u_{\text{Kru}}$, $v = v_{\text{EF}}$ and the function $\Omega^2$ (see \eqref{prelim_metric}). The functions $\varpi$ and $\mu$ are defined in \eqref{def_varpi} and \eqref{def_mu}, respectively.
\begin{itemize}
\item \textbf{Redshift effect}, sections \ref{sec:EH}, \ref{sec:RR} and \ref{sec:NR}: Along the event horizon, we assume an exponential Price law upper bound $e^{-sv}$, for a fixed $s \in (0, +\infty)$, for the scalar field. Near the event horizon, the redshift effect competes with the exponential Price law. This is translated into the presence of the slower\footnote{Here, $c(s)$ is bounded from above by a value depending on the surface gravity of the event horizon of the final black hole, see definition \eqref{def_cs}.} decay rate $e^{-c(s)v}$, rather than $e^{-sv}$.  The monotonicity results and Gr{\"o}nwall inequalities adopted in the real-scalar-field case are replaced by a bootstrap argument. The latter is analogous to the procedure in \cite{VdM1}, where however in our case we need to take care of the different decay rates in function of $s$, due to the exponential nature of the estimates. 
\item \textbf{Blueshift effect}, sections \ref{sec:NR}, \ref{sec:EBR1}, \ref{sec:EBR2}, \ref{section:lateblueshift}: Near the Cauchy horizon (whose existence is proved at the end of section \ref{section:lateblueshift}) an exponentially growing contribution, arising from blueshift, affects the main estimates. However, this effect is relevant only in a relatively small region in which we can still propagate the previously obtained decay rates, up to a small error in the exponent of the main terms. In \cite{CGNS4}, this was dealt with by using the monotonicity properties of $\varpi$ (for instance to bound $1-\mu$ away from zero in the no-shift region and to show the integrability of $\partial_u r$ and $\partial_v r$ near the Cauchy horizon). In our case, however, these monotonicity properties are lacking. In their place, we exploit bootstrap arguments. 

The bootstrap procedure has to differ from that of \cite{VdM1}: there, the estimate $u_{\text{EF}} \sim |\log(2K_+ u_{\text{Kru}})| \sim v_{\text{EF}}$ (see section \ref{section:notations} for an overview of the different coordinate systems) is used multiple times. The constant in this estimate leads to a multiplying factor in front of terms of the form $v^{-s}$. These factors can be dealt with using the smallness parameters of the bootstrap procedure. On the other hand, in our case, such estimates on the null coordinates lead to exponentially large error terms that cannot be removed.
To solve this problem, we 
\[
\text{replace} \quad  e^{-c(s)v} \quad \text{with} \quad e^{-\mathcal{C}_s(u, v)} \sim e^{-(c(s) -\tau)v}
\]
in the main estimates, for a suitable positive\footnote{Strictly speaking, $\mathcal{C}_s$ is positive only in a specific region of the black hole interior that has the curve $\gamma$ as future boundary, see section \ref{sec:EBR1}. This does not constitute a problem since all estimates to the future of this region can be expressed in terms of ${\mathcal{C}_s}_{|\gamma}$.} function $\mathcal{C}_s$ and for $\tau > 0$. In particular, the bootstrap estimates need to contain a function that depends both on the Kruskal-type coordinate $u$ and on the Eddington-Finkelstein type of coordinate $v$.\footnote{We follow the convention of \cite{CGNS1, CGNS2, CGNS3, CGNS4}, where the null coordinates $(u, v) = (u_{\text{Kru}}, v_{EF})$ are used in the entire black hole interior. See also section \ref{section:notations} for a comparison with other conventions.}

A key role in the last region is played by the term $\partial_v \log \Omega^2$ (see also the importance of this term in \cite{VdM1}) to replace the monotonicity properties used in \cite{CGNS4} and to obtain a finer control on the $H^1$ norm of the scalar field. Indeed, the product $u \Omega^2$ measures, in some sense, the redshift and blueshift effects.

\item \textbf{$\bm{C^0}$ and $\bm{H^1}$ extensions}, sections \ref{section:C0extension}, \ref{subsec:nomassinflation}, \ref{subsec:constructionH1extension}: once we propagate the exponential decay rates from $\mathcal{H}^+$, the construction of a $C^0$ extension of the spacetime metric mainly follows from  the results in \cite{Dafermos_2005_uniqueness, CGNS4, LukOh1, VdM1}. 

The $H^1$ extension that we construct is related to the boundedness of the renormalized Hawking mass $\varpi$ (see also \cite{CGNS4}). In  the expression \eqref{varpi_explicit} for $\varpi$ (here $\theta = r \partial_v \phi$ and $\lambda = \partial_v r$), all quantities are bounded except possibly for the integral of 
\begin{equation} \label{integralterm_varpi}
\frac{|\partial_v \phi|^2}{\partial_v r}.
\end{equation}
We prove that, for a sufficiently fast decay rate of the scalar field along $\mathcal{H}^+$, i.e. for $s$ sufficiently large, and if the asymptotic black hole is sufficiently near-extremal, $|\partial_v r|$ admits a lower bound that makes \eqref{integralterm_varpi} integrable in $v$. Thus,  $\varpi$ is bounded. The integrability of \eqref{integralterm_varpi} is then used to prove that the Christoffel symbols are locally square integrable, therefore leading to a $H^1_{\text{loc}}$ extension of the metric along $\mathcal{CH}^+$.
\end{itemize}

The \textbf{outline of the paper} is the following: in section \ref{section:so3} we delve into the main aspects of the characteristic IVP for the Einstein-Maxwell-charged-Klein-Gordon system under spherical symmetry.
Well-posedness of the characteristic IVP is investigated in section \ref{section:wellposedness}. 

In section \ref{section:stability} we establish quantitative bounds along $\mathcal{H}^+$ for the main PDE variables of the characteristic IVP, and propagate these bounds to the black hole interior by bootstrap. This allows us to prove the stability of $\mathcal{CH}^+$ and to construct a continuous spacetime extension at the end of the same section.

In section \ref{sec:nomassinflation} we provide a sharp condition on the reference sub-extremal black hole to generically prevent the occurrence of mass inflation. Under such an assumption, we also construct an $H^1_{\text{loc}}$ spacetime extension up to and including the Cauchy horizon.

\subsection{Acknowledgements}
The author is grateful to his PhD supervisors João Costa and José Natário for suggesting the problem and for their invaluable guidance and advice while completing the manuscript. The author would also like to thank Ryan Unger for insightful comments on the manuscript.
This work was supported by FCT/Portugal through the PhD scholarship UI/BD/152068/2021; partially supported by FCT/Portugal through CAMGSD, IST-ID, projects
UIDB/04459/2020 and UIDP/04459/2020; and partially supported by the H2020-MSCA-2022-SE project EinsteinWaves, GA No. 101131233.

\section{The Einstein-Maxwell-charged-Klein-Gordon system under spherical symmetry} \label{section:so3}
The present work revolves around the system of equations describing a spherically symmetric Einstein-Maxwell-charged-Klein-Gordon model with positive cosmological constant $\Lambda$: 
\begin{align*}
\begin{cases}
\Ric_{\mu \nu}(g)-\frac{1}{2}\R(g) g_{\mu \nu} + \Lambda g_{\mu \nu} = 2 \br{\Tem_{\mu \nu} + \Tphi_{\mu \nu} }, \\[0.3em]
\Tem_{\mu \nu} = g^{\alpha \beta}F_{\alpha \mu} F_{\beta \nu} - \frac{1}{4} g_{\mu \nu} F^{\alpha \beta}F_{\alpha \beta}, \\[0.3em]
\Tphi_{\mu \nu} = \Re \br{D_{\mu} \phi \overline{D_{\nu} \phi }} - \frac{1}{2}  g_{\mu \nu} \br{D_{\alpha} \phi \overline{D^{\alpha}\phi} + m^2 |\phi|^2 }, \\[0.3em]
dF=0, \quad d \star F = \star J, \\[0.3em]
\br{D_{\mu}D^{\mu}-m^2} \phi = 0. 
\end{cases}
\end{align*}
The above system describes a spherically symmetric  spacetime $(\mathcal{M}, g)$ interacting with a scalar field $\phi$, the latter being endowed with charge $q \in \bbR \setminus \{0\}$ and mass $m \ge 0$.
The units are chosen so that $c=4\pi G = \epsilon_0 = 1$ and the symbol $\star$ denotes the Hodge-star operator. The charge generates an electromagnetic field which interacts with the spacetime itself and is modelled by the strength field tensor $F$.  The description of the charged scalar field through a complex-valued function and the use of the gauge covariant derivative 
\begin{equation} \label{cov_derivative}
D_{\mu} = \nabla_{\mu} + i q A_{\mu}
\end{equation}
 are standard techniques in the gauge-theoretical framework (see \cite{Kommemi} 
 for more details) and in Lagrangian formalism. The current 
\begin{equation} \label{noethercurr}
J_{\mu}= -\frac{i q}{2} \br{\phi \overline{D_{\mu} \phi} - \overline{\phi} D_{\mu} \phi}
\end{equation}
appearing in Maxwell's equations, in particular, can be seen as the Noether current of the (classical) Lagrangian density $\mathcal{L}=|D_{\mu} \phi|^2-m^2|\phi|^2$.

We express the isometric action of $SO(3)$ on our spacetime manifold by assuming
\begin{equation} \label{metric_eqn}
g = g_{\mathcal{Q}} +r^2(u, v)\sigma_{S^2}, 
\end{equation}
with
\[
g_{\mathcal{Q}} \coloneqq -\Omega^2(u, v)du dv,
\]
for some real-valued function $\Omega$, where $\sigma_{S^2}=d\theta^2 + \sin^2(\theta)d\varphi^2$ is the standard metric on the unit round 2-sphere, $r$ is the area-radius function and $u, v$ are null coordinates in the $(1+1)$--dimensional manifold $\mathcal{Q} \coloneqq \mathcal{M}/SO(3)$. In particular, we can conformally embed $(\mathcal{Q}, g_{\mathcal{Q}})$ in $(\bbR^2, \eta)$, where $\eta$ is the Minkowski metric, and we can restrict ourselves to working with this subset of the two-dimensional space  since conformal transformations preserve the causal structure.  Moreover, we will make use of the $U(1)$ gauge freedom to assume that $\phi$ and $A$ (see \eqref{cov_derivative}) do not depend on the angular coordinates. We will also require that $\partial_u = \frac{\partial}{\partial u}$ is a future-pointing, ingoing\footnote{In the language of \cite{Kommemi}, a vector field defined on the quotient spacetime manifold is \textbf{ingoing} if it points towards the centre of symmetry, i.e.\ the projection on $\mathcal{M}/SO(3)$ of the set of fixed points of the $SO(3)$ action. See also \cite[proposition 2.1]{Kommemi} for a related analysis of topologies for the initial data, in the $\Lambda = 0$ case.} vector field and $\partial_v$ is a future-pointing, outgoing one. 

We use the same symbol to denote a spherically symmetric function defined on $\mathcal{M}$, and its push-forward through the projection map $\pi \colon \mathcal{M} \to \mathcal{Q}$.

In order to analyse the behaviour of our dynamical system, we are first naturally led to investigate the maximal globally hyperbolic development (\textbf{MGHD}) generated by  (spherically symmetric) initial data prescribed on a suitable spacelike hypersurface (see also section \ref{sec:SCCCsubtle} for possible global structures). After this analysis, the null structure of the equations motivates us to treat the characteristic initial value problem on $\mathcal{Q}$, using the $(u, v)$ coordinate system.

With respect to the afore-mentioned $(u, v)$ null coordinates, we are going to work in the set $[0, U] \times [0, +\infty)$, for some $U > 0$. Therefore, there exists a 1-form $A$ defined on the MGHD of the initial data set, such that $F=dA=F_{uv}(u, v)du \wedge dv$ for some real-valued function $F_{uv}$. In particular, we can define the real-valued function $Q(u, v)$ such that 
\begin{equation} \label{fieldtensor}
F= \frac{Q \Omega^2}{2r^2}du \wedge dv.
\end{equation}
Since we can choose the electromagnetic potential $A$ up to an exact form by performing a gauge choice, we can set $A=A_u du$. Thus, the gauge covariant derivative $D_v$ is just $\nabla_v$ (see \eqref{cov_derivative}) and 
\begin{equation} \label{def_uptofunction}
\dee{v}{A_u} = -\frac{Q \Omega^2}{2 r^2}.
\end{equation}
The above expression defines $A_u$ up to a function of $u$. So, without loss of generality, we will assume that $A_u$ vanishes on the initial ingoing hypersurface. Namely, for $U > 0$ and $v_0 \ge 0$ fixed, we take $A_u(u, v_0) = 0$ for every $u \in [0, U]$.

\subsection{The second-order and first-order systems:  the characteristic initial value problem} \label{section:ivp}
Using the results of appendix \ref{appendixUsefulExpressions}, we can write down the Einstein equations in the following form:
\begin{empheq}[left=\empheqlbrace]{align}
&\dee{u}{\br{\frac{\partial_u r}{\Omega^2}}} = -r \frac{\abs{D_u \phi}^2}{\Omega^2}, \quad  \text{(Raychaudhuri equation along the }u\text{ direction),} \label{raych_u}\\
&\dee{v}{\br{\frac{\dee{v}{r}}{\Omega^2}}} = -r \frac{\abs{\dee{v} \phi}^2}{\Omega^2}, \quad  \text{  (Raychaudhuri equation along the }v\text{  direction),} \label{raych_v} \\
&\dee{u}{\dee{v}{r}} = \frac{\Omega^2}{4r^3}Q^2 + \frac{\Omega^2 r}{4}\Lambda - \frac{\Omega^2}{4r} - \frac{\dee{v}{r} \dee{u}{r}}{r}  + \frac{m^2 \Omega^2}{4}r |\phi|^2, \quad\,\,\,  \text{ (wave eqn.\ for }r\text{),} \label{waveeqn_r}\\
&\dee{u}{\dee{v}{\log \Omega^2}} = - 2\Re{\br{D_u \phi \overline{\dee{v} \phi}}} - \frac{\Omega^2 Q^2}{ r^4} + \frac{\Omega^2}{2 r^2} + \frac{2\dee{v}{r} \dee{u}{r}}{r^2}, \,\,\,\, \text{(wave eqn.\ for}\log \Omega^2\text{),}  \label{waveeqn_log}
\end{empheq}
Moreover, the equation of motion $\br{D_{\mu}D^{\mu} - m^2}\phi = 0$ becomes
\begin{equation} \label{klein_gordon}
\dee{u}{\dee{v}{\phi}} = -\frac{1}{r} \br{D_u{r} \dee{v}{\phi} + \dee{v}{r} D_u{\phi}} + \frac{iq Q \Omega^2}{4r^2} \phi -\frac{m^2 \Omega^2}{4}\phi, \quad \text{(Klein-Gordon equation)}
\end{equation}
where we used that 
\[
\nabla_{\mu}\nabla^{\mu} \phi = -\frac{4}{r\Omega^2}\br{\dee{u}{r}\dee{v}{\phi} + \dee{v}{r}\dee{u}{\phi} + r \dee{u}{\dee{v}{\phi}} },
\]
and
\[
\nabla_{\mu}A^{\mu} = -\frac{2}{\Omega^2} \dee{v}{A_u} - \frac{4}{\Omega^2 r}A_u \dee{v}{r} = \frac{Q}{r^2} - \frac{4}{\Omega^2 r}A_u \dee{v}{r}.
\]
Furthermore, the second Maxwell equation $d\star F = \star J$ implies
\begin{empheq}[left=\empheqlbrace]{align}
\partial_u Q &= -q r^2 \Im \br{\phi \overline{D_u \phi}}, \label{maxwelleqnu} \\
\partial_v Q &= q r^2 \Im \br{ \phi \overline{\dee{v} \phi}}. \label{maxwelleqnv}
\end{empheq}
In order to study the well-posedness of our system of PDEs, it is useful to rewrite the above equations as a first-order system. To do this, we define the following quantities:
\begin{align}
\nu &\coloneqq \dee{u}{r}, 
\\
\lambda &\coloneqq \dee{v}{r}, \label{def_lambda}
\end{align}
which were first introduced by Christodoulou, and
\begin{align}
\varpi &\coloneqq \frac{Q^2}{2r} + \frac{r}{2} - \frac{\Lambda}{6}r^3 + \frac{2r}{\Omega^2}\nu \lambda, \label{def_varpi}
\\
\mu &\coloneqq \frac{2\varpi}{r} - \frac{Q^2}{r^2} + \frac{\Lambda}{3}r^2, \label{def_mu}
\\
\theta &\coloneqq r \dee{v} \phi, \label{def_theta}
\\
\zeta &\coloneqq r D_u \phi, \label{def_zeta}
\\
\kappa &\coloneqq -\frac{\Omega^2}{4 \nu}, \label{def_kappa}
\end{align}
whenever $r$, $\Omega^2$ and $\nu$ are non-zero, an assumption that we will make on the initial hypersurfaces of our PDE system.
All the newly defined quantities are real-valued, except for $\theta$ and $\zeta$, since the scalar field $\phi$ is complex-valued. From the above we can obtain the useful relation\footnote{The variable $\kappa$ was originally introduced in order to avoid problematic terms showing up in the case $\lambda = 0$.}
\[
\lambda = \kappa(1-\mu).
\]
We also notice that $\varpi$ is a geometric quantity (often called the \textbf{renormalized Hawking mass}), due to the fact that $1-\mu = g(d r^{\sharp}, d r^{\sharp})$. Using the new definitions, some lines of computations show that the Einstein, Klein-Gordon and Maxwell equations imply the following first-order system:\footnote{Strictly speaking, the equality $\partial_v \nu = \partial_u \lambda$ and the equation for $\partial_v \zeta$ require an additional step to be obtained. Indeed, they can be seen to hold for continuous solutions such that all partial derivatives present in the PDE system are continuous, after proving that $\partial_u \partial_v r = \partial_v \partial_u r$ (which follows from the second equality in \eqref{dulambda}) and  that $\partial_u \partial_v \phi = \partial_v \partial_u \phi$ (which follows from \eqref{duphi}). See also the proof of local well-posedness in \cite{mythesis}.}
\allowdisplaybreaks
\begin{align}
\dee{u}{r} &= \nu, \label{dur} \\
\dee{v}{r} &= \lambda, \label{dvr}\\
D_u \phi = \br{\partial_u + i q A_u}\phi &= \frac{\zeta}{r}, \label{duphi}\\
D_v \phi =\dee{v}{\phi} &= \frac{\theta}{r}, \label{dvphi}\\ 
 \dee{v}{\nu} = \dee{u}{\lambda} &=-2\nu \kappa \frac{1}{r^2}\br{\frac{Q^2}{r} +\frac{m^2 r^3|\phi|^2}{2} + \frac{\Lambda}{3}r^3 - \varpi}, \label{dulambda}
\\
\dee{u}{\varpi} &=\frac{\nu}{2}m^2 r^2 |\phi|^2 - Q q \Im{\br{\phi \overline{\zeta}}} + \frac{1}{2} \br{\frac{|\zeta|}{\nu}}^2 \nu \br{1-\frac{2\varpi}{r} + \frac{Q^2}{r^2} -\frac{\Lambda}{3}r^2}, \label{duvarpi}
\\
\dee{v}{\varpi} &= \frac{\lambda}{2}m^2 r^2 |\phi|^2 + Qq \Im{\br{\phi \overline{\theta}}} + \frac{|\theta|^2}{2\kappa}, \label{dvvarpi}
\\
D_u \theta = \br{\partial_u + i q A_u} \theta &= -\frac{\zeta}{r} \lambda +\frac{\nu \kappa \phi}{r}\br{m^2 r^2 - i q Q}, \label{dutheta}
\\
D_v \zeta = \dee{v}{\zeta} &= - \frac{\theta}{r}\nu +\frac{\nu \kappa \phi}{r}\br{m^2 r^2 + i q Q}, \label{dvzeta}
\\
\dee{u}{\kappa} &= \frac{\kappa \nu }{r}\br{\frac{|\zeta|}{\nu}}^2, \label{dukappa}\\
\dee{u}{Q} &= -q r \Im(\phi \overline{\zeta}), \label{duQ}  \\
\dee{v}{Q} &= q r \Im(\phi \overline{\theta}), \label{dvQ} \\
\dee{v}{A_u} &= 2\frac{Q \nu \kappa}{r^2}, \label{dvA}
\end{align}
with the algebraic constraint
\begin{equation}
\lambda = \kappa (1-\mu). \label{algconstr}
\end{equation}
It is also convenient to define the quantity
\begin{equation} \label{def_surfacegrav}
K \coloneqq \frac12 \partial_r (1-\mu)(r, \varpi, Q) = \frac{\varpi}{r^2} - \frac{Q^2}{r^3} - \frac{\Lambda}{3} r,
\end{equation}
which, as we will see, is a generalization of the surface gravity of the Killing horizons of a Reissner-Nordstr{\"o}m black hole with cosmological constant $\Lambda$. The above quantities then satisfy
\begin{equation} \label{dulambda_kappa}
 \partial_v \nu =\partial_u \lambda  = \nu \kappa \br{2 K - m^2 r |\phi|^2},
\end{equation}
by \eqref{dulambda}. 
By \eqref{def_varpi} and \eqref{def_kappa}, we also have: 
\begin{equation} \label{varpi_kappa}
\kappa \varpi = \kappa \br{\frac{Q^2}{2r} + \frac{r}{2} - \frac{\Lambda}{6}r^3} - \frac{r}{2}\lambda.
\end{equation}
Therefore, we can recast \eqref{dulambda} as:
\begin{equation} \label{rwaveeqn}
\partial_u(r \lambda) = \partial_v (r \nu) = -\nu \kappa \br{ \frac{Q^2}{r^2} + m^2 r^2 |\phi|^2 -1 + \Lambda r^2}.
\end{equation}
Using \eqref{def_kappa}, \eqref{duphi}, \eqref{dvphi} and \eqref{algconstr}, the equations of the first-order and second-order systems can be expressed in several forms. For our purposes, it will be useful to write the equation for $\partial_u \varpi$ as
\begin{equation} \label{duvarpi2}
\partial_u \varpi = \frac{\nu}{2}m^2 r^2 |\phi|^2 - Qq r \Im{(\phi \overline{D_u \phi})} - \frac{2\lambda}{\Omega^2}r^2 |D_u \phi|^2,
\end{equation}
and the Klein-Gordon equation \eqref{klein_gordon} as one of the following expressions:
\begin{align}
D_u(r\partial_v \phi) &=- \lambda D_u \phi + m^2 \nu \kappa r \phi - i \frac{q Q \nu \kappa}{r} \phi, \label{waveeqn_phi2_r} \\
\partial_v(r D_u \phi) &= -\nu \partial_v \phi + m^2 \nu \kappa r \phi + i \frac{q Q \nu \kappa}{r} \phi. \label{waveeqn_phi3_r}
\end{align}
As noticed in \cite{VdM1} (and the result still holds for non-zero values of $\Lambda$), the wave equation \eqref{waveeqn_log} for $\log \Omega^2$ can be cast as 
\begin{equation} \label{waveeqn_log_new}
\partial_u \partial_v \log \Omega^2 = \kappa \partial_u (2K) - 2 \Re \br{D_u \phi \overline{\partial_v \phi}} - \frac{2\kappa}{r^2} \br{\partial_u \varpi - \frac{2 Q \partial_u Q}{r}}.
\end{equation}
The proof of this result only requires some algebra and the definitions \eqref{def_surfacegrav} and \eqref{def_varpi}, in this exact order.

In the next sections, we are going to discuss the local well-posedness of a characteristic IVP for the first-order system \eqref{dur}--\eqref{algconstr}. Solutions to such a system are elements $(\lambda,$ $\varpi,$ $\theta,$ $\kappa,$ $\phi,$ $Q,$ $r,$ $\nu,$ $\zeta,$ $A_u)$ in an appropriate Cartesian product of function spaces. The system is overdetermined, since we are solving 13 partial differential equations and an algebraic constraint for 10 unknowns. This is why we will solve for equations \eqref{dvr}-\eqref{duphi}, \eqref{dulambda}-\eqref{duvarpi}, \eqref{dutheta}-\eqref{duQ}, \eqref{dvA} and we will consider equations \eqref{dur}, \eqref{dvphi}, \eqref{dvvarpi}, \eqref{dvQ}, \eqref{algconstr} as constraints.
In particular, for some $U > 0$, $v_0 \ge 0$, we will prescribe initial data on the set $\br{[0, U] \times \{v_0\}} \cup \br{\{0\} \times [v_0, +\infty)} \subset [0, U]\times [0, +\infty)$, seen as a subset of the conformal embedding of $(\mathcal{Q}, g_{\mathcal{Q}})$ in $(\bbR^2, \eta)$. In fact, we will specify:
\begin{equation} \label{initialdata_u}
\begin{cases}
r(u, v_0),\\
\zeta(u, v_0),\\
A_u(u, v_0),
\end{cases}
\quad \text{ for } u \in [0, U],
\end{equation}
and 
\begin{equation} \label{initialdata_v}
\begin{cases}
\lambda(0, v),\\
\varpi(0, v),\\
\theta(0, v),\\
\kappa(0, v),\\
\phi(0, v), \\
Q(0, v),
\end{cases}
\quad \text{ for } v \in [v_0, +\infty).
\end{equation} 
The constraint equation \eqref{dur} will be imposed on $[0, U] \times \{v_0\}$ so that the function $u \mapsto \nu(u, v_0)$ can be obtained from $u \mapsto r(u, v_0)$.
In the following, we will denote the prescribed initial data by the zero subscript, e.g. $r_0(u) \coloneqq r(u, v_0)$ and $A_{u, 0}(u)=A_u(u, v_0)$. The coordinates $(u, v)$ have not been fixed yet: we only required them to be null and to stem from the above conformal embedding.  We are going to fix a specific gauge in the following. 

\subsection{Assumptions} \label{section:assumptions}
Let us make some general assumptions so that our spherically symmetric model describes a dynamical black hole asymptotically approaching (in some sense) a sub-extremal Reissner-Nordstr{\"o}m solution with cosmological constant $\Lambda$ (we refer to \cite{CostaGirao, CostaFranzen, CGNS2} for an overview of this spacetime). In this section, we fix the values $v_0$ and $U$ so that $v_0 \ge 0$ and $0 < U <+\infty$. 
\begin{itemize}
\item[\textbf{A)}] \textbf{Gauge-fixing for the null coordinates}: exploiting the diffeomorphism invariance of the Einstein-Maxwell-charged-Klein-Gordon system, we fix the $u$ coordinate by setting
\begin{equation} \label{assumption:nu}
\nu_0(u) = -1, \quad \forall \, u \in [0, U]
\end{equation}
and we specify the $v$ coordinate by choosing
\begin{equation} \label{assumption:kappa}
\kappa_0(v) = 1, \quad \forall \, v \in [v_0, +\infty).
\end{equation}
The negativity of $\nu_0$ can be related to the absence of anti-trapped surfaces.  We also require the positivity of the area-radius:
\begin{equation} \label{pos_radius}
r_0(u) > 0 \quad \forall \, u \in [0, U] \quad \text{ and } \quad r(0, v) > 0 \quad \forall\, v \in [v_0, +\infty).
\end{equation}
\end{itemize}
Using the above gauge, we define 
\[
\mathcal{H}^+ \coloneqq \{(u, v)  \colon \, u=0, v \ge v_0\}
\]
as the \textbf{event horizon} of the dynamical black hole.
\begin{itemize}
\item[\textbf{B)}] \textbf{Regularity assumptions}:
we suppose that 
\[
\nu_0, \zeta_0, \lambda_0, \theta_0, \kappa_0 \text{ are } C^0 \text{ functions in their respective domains},
\]
and that 
\[
r_0, \varpi_0, \phi_0, A_{u, 0}, Q_0 \text{ are } C^1 \text{ functions in their respective domains} .
\]
We emphasize that $\theta_0$, $\zeta_0$ and $\phi_0$ are complex-valued functions. We will also make use of the function $v \mapsto r(0, v)$, whose regularity depends on the regularity of $\lambda_0$. These assumptions are enough to prove well-posedness of the first-order system.  In order to demonstrate that such a system implies the Einstein-Maxwell equations, it is sufficient to additionally require that $\nu_0$, $\lambda_0$ and $\kappa_0$ are $C^1$ functions in their respective domains (see also the forthcoming remark \ref{rmk:equivalence}).
\item[\textbf{C)}] \textbf{Compatibility conditions}:
we require the following constraints to be satisfied:
\begin{equation}
\begin{cases}
r'_0(u) = \nu_0(u), \\
\phi'_0(v) = \frac{\theta_0(v)}{r(0, v)}, \\
\varpi'_0(v) = \frac{\lambda_0(v)}{2}m^2 r(0, v)^2 |\phi_0(v)|^2 + Q_0(v) q \Im(\phi_0(v) \overline{\theta_0}(v)) + \frac{|\theta_0(v)|^2}{2\kappa_0(v)}, \\
Q'_0(v) = q \, r(0, v) \Im(\phi_0(v) \overline{\theta_0}(v)) \\
\lambda_0(v) = \kappa_0(v) \br{1- \frac{2\varpi_0(v)}{r(0, v)} + \frac{Q_0(v)^2}{r(0, v)^2} - \frac{\Lambda}{3}r(0, v)^2},
\end{cases}
\end{equation}
for every $u \in [0, U]$, every $v \in [v_0, +\infty)$.
\item[\textbf{D)}] \textbf{Absence of anti-trapped surfaces}: we require
\begin{equation} \label{absence_antitrapped}
\nu(0, v) < 0, \quad \forall \, v \in [v_0, +\infty).
\end{equation}
Once we construct a solution in a set $\mathcal{D} = [0, U') \times [v_0, +\infty)$ for $0 < U' < +\infty$, the Raychaudhuri equation \eqref{raych_u} implies that $\nu$ remains negative in $\mathcal{D}$. 
\end{itemize}
Starting from section \ref{section:stability}, we also assume the following.
\begin{itemize}
\item[\textbf{E)}] \textbf{Exponential Price law upper bound}: we require that $\Lambda > 0$ and that, on the event horizon of the dynamical black hole, $\phi$ decays as
\begin{equation} \label{assumption:expPricelaw}
|\phi|(0, v) + |\partial_v \phi|(0, v) \le C e^{-sv}, \quad \forall \, v \in [v_0, +\infty)
\end{equation}
for some $C>0$, $s > 0$. This assumption is one of the main consequences of the presence of a positive cosmological constant $\Lambda$. Here, the coordinate $v$ is related to an Eddington-Finkelstein coordinate adopted near the event horizon of a Reissner-Nordstr{\"o}m-de Sitter black hole (see appendix \ref{app:eddfinkel}).
\item[\textbf{F)}] \textbf{Asymptotically approaching a sub-extremal  black hole}: consider a fixed sub-extremal Reissner-Nordstr{\"o}m-de Sitter spacetime. Let $r_+$ and $K_+$ be, respectively, the radius and the surface gravity of its event horizon, and let $Q_+$ be its charge and $\varpi_+$ its mass. Notice that several constraints apply on these four values, due to the sub-extremality condition.\footnote{Given $K_+ \coloneqq \frac{1}{2}\partial_r(1-\mu)(r_+, \varpi_+, Q_+)$, the values of $\varpi_+, Q_+$ and $r_+$ are constrained by the sub-extremality condition $K_+ > 0$ (see also \cite[section 3]{Hintz}).
 Moreover, $(1-\mu)(r_+, \varpi_+, Q_+)=0$.  As in \cite{CGNS2}, in the course of our work we assume that $(1-\mu)(\cdot, \varpi_+, Q_+)$ admits the three distinct and positive roots $r_{-}, r_+$ and $r_C$, with $ r_{-} < r_+ < r_C$, corresponding to the radii of the Cauchy horizon, event horizon and cosmological horizon, respectively, of a Reissner-Nordstr{\"o}m-de Sitter black hole.}
Then, on the event horizon, we require:\footnote{More generally, we could assign a condition on the $\limsup$ of $r, Q, \varpi, K$. This requires minimal changes in the proof of proposition \ref{prop: boundsEH} and yields the same results.} 
\begin{align} 
\lim_{v \to +\infty} r(0, v) &=  r_+,  \label{assumption:r} \\
\lim_{v \to +\infty} Q(0, v) &= Q_+, \label{assumption:Q}
\\
\lim_{v \to +\infty} \varpi(0, v) &= \varpi_+, \label{assumption:varpi} \\
\lim_{v \to +\infty} K(0, v) &= K_+. \label{assumption:surface_grav}
\end{align}
\item[\textbf{G)}] \textbf{Hawking's area theorem}: we assume\footnote{The case $\lambda_{|\mathcal{H}^+} \equiv 0$ was studied in \cite{CGNS2, CGNS3}, since Reissner-Nordstr{\"o}m-de Sitter data were prescribed along $\mathcal{H}^+$. The case in which $\lambda(0, \cdot)$ is identically zero for large values of the $v$ coordinate, but not for  small ones, falls into that analysis after a coordinate shift. Moreover, due to our construction, the initial data along the event horizon are approaching the data of a Reissner-Nordstr{\"o}m-de Sitter black hole in a non-trivial way. So, the apparent horizon does not coincide with the event horizon. By Hawking's area theorem \cite{CGNS4} and lemma \ref{lemma:constantu}, then, we are left to examine the case $\lambda_{|\mathcal{H}^+} > 0$.}
\begin{equation} \label{assumption:lambda}
\lambda(0, v) > 0, \quad \forall \, v \ge v_0.
\end{equation}
 \end{itemize}

\subsection{Notations and conventions} \label{section:notations}
Given two non-negative functions $f$ and $g$, we use the notation $f \lesssim g$ to denote the existence of a positive constant $C$ such that $f \le C g$. The relation $f \gtrsim g$ is defined analogously, and we write $f \sim g$ whenever both $f \lesssim g$ and $f \gtrsim g$ hold. Different constants may be denoted by the same symbols if the value of such constant is not relevant in the computations. Whenever the symbols $\lesssim$ and $\gtrsim$ are used, we mean that the respective constants depend on the initial data only, except for the following cases:
\begin{itemize}
\item In the redshift region (section \ref{sec:RR}), constants depend on the initial data and possibly on the constant $\eta$ (see the statement of proposition \ref{prop: redshift}), 
\item In the no-shift region (section \ref{sec:NR}), early blueshift region (sections \ref{sec:EBR1} and \ref{sec:EBR2}) and late blueshift region (section \ref{section:lateblueshift}), constants depend on the initial data, and possibly on $\eta$, $R$ and $Y$ (see propositions \ref{prop: redshift} and \ref{prop: noshift}).
\end{itemize}

\textbf{Coordinate systems}:
We follow the conventions of \cite{CGNS1, CGNS2, CGNS3, CGNS4}, where the $(u, v) = (u_{\text{Kru}}, v_{\text{EF}})$ coordinate system  (that we specified in section \ref{section:assumptions}) is used in the entirety of the black hole interior. Notice that  additional coordinate systems have been used in the literature in  the case $\Lambda = 0$. For instance, in \cite{VdM1} and \cite{LukOh1}, our coordinate system is only used in the redshift region, whereas the  Eddington-Finkelstein coordinate $u_{EF}$,  defined as
\[
u_{\text{EF}} \coloneqq \frac{1}{2K_+}\log \br{2K_+ u_{\text{Kru}}}
\]
is used in the remaining region of the black hole interior.
Moreover, the quantity $\nu = \nu_{\text{Kru}}$ (see also $\Omega^2 = \Omega^2_{\text{Kru}}$) that we use throughout the paper corresponds to the quantity $\nu_{\text{H}}$ (see also $\Omega^2_H$) in \cite{VdM1}. It is therefore different from the quantity 
\[
\nu_{\text{EF}} = 2K_+ u_{\text{Kru}} \nu_{\text{Kru}}
\]
used in \cite{VdM1} to the future of the redshift region.

\section{Well-posedness of the initial value problem} \label{section:wellposedness}
Following \cite{CGNS1}, we are going to discuss local existence, uniqueness and continuous dependence with respect to the initial data for the solutions to the first-order system \eqref{dur}-\eqref{algconstr}. In this section, we consider the constants $U \ge 0$ and $V \ge 0$, while we take $v_0 = 0$ for the sake of convenience. The initial data  $\kappa_0$, $\nu_0$ and $A_{u, 0}$ are taken to be constant, according to \hyperref[section:assumptions]{assumption} \textbf{(A)} and \eqref{def_uptofunction}. The case of more general functions $\nu_0, \kappa_0$ and $A_{u, 0}$ follows in a straightforward way.
\begin{definition}[solution to the PDE system] \label{solution_continuity}
We define a solution to the PDE system \eqref{dur}--\eqref{dvA}, with initial conditions \eqref{initialdata_u}, \eqref{initialdata_v} and satisfying the regularity \hyperref[section:assumptions]{assumptions} \textbf{(B)} and constraint \eqref{algconstr}, to be a vector $(\lambda,$ $\varpi,$ $\theta,$ $\kappa,$ $\phi,$ $Q,$ $r,$ $\nu,$ $\zeta,$ $A_u)$ of continuous functions defined on $[0, U] \times [0, V]$, such that all partial derivatives appearing in such a  system are continuous.
\end{definition}

\begin{theorem}[local existence and uniqueness] \label{thm:localexistence}
Under the \hyperref[section:assumptions]{assumptions} \textbf{(A)}, \textbf{(B)}, \textbf{(C)} and \textbf{(D)}, let us prescribe initial data on the characteristic initial set $[0, U] \times \{0\} \cup \{0\} \times [0, V]$ for some $0< U <+\infty$ and $0 < V < + \infty$. We define the quantity 
\begin{align*}
N_{\text{i.d.}} \coloneqq  &\br{\|\lambda_0\|+ \|\varpi_0\|+ \|\theta_0\|+ \|\kappa_0\| + 
\|\phi_0\|+ \|Q_0\|}_{L^\infty([0, V])} \\
&+ \br{\|r_0\| + \|\nu_0 \|+  \|\zeta_0\|+ \|A_{u, 0}\|}_{L^{\infty}([0, U])}.
\end{align*}
Then, there exists a \textbf{time of existence} $0 < \epsilon=\epsilon(N_{\text{i.d.}}) \le U$ for which the characteristic IVP for the first-order system \eqref{dur}--\eqref{algconstr} admits a unique solution in $\mathcal{D}_{\epsilon} \coloneqq [0, \epsilon] \times [0, V]$, the solution being such that the functions $r, \nu$ and $\kappa$ are bounded away from zero in $\mathcal{D}_{\epsilon}$.

Similarly, there exists a \textbf{time of existence} $0 < \epsilon = \epsilon(N_{\text{i.d.}}) \le V$ such that the characteristic IVP admits a unique solution in $\mathcal{D}^{\epsilon} \coloneqq [0, U] \times [0, \epsilon]$ and such that the functions $r, \nu$ and $\kappa$ are bounded away from zero in $\mathcal{D}^{\epsilon}$.
\end{theorem}
\begin{proof}
The proof consists of an adaptation of the fixed-point argument exploited in the proof of theorem 4.2 in \cite{CGNS1}. 
In particular, the proof shows that the following formal expressions, 
derived from equations \eqref{dur}--\eqref{dvA}, are rigorously satisfied in a suitable metric space of solutions:
\allowdisplaybreaks
\begin{align}
\lambda(u, v) &= \lambda_0(v) - \int_0^u \frac{2\nu \kappa}{r^2} \br{\frac{Q^2}{r} + \frac{m^2 r^3}{2}|\phi|^2 + \frac{\Lambda}{3}r^3 - \varpi} (u', v) du', \label{firstorder_lambda} \\
\varpi(u, v) &= \varpi_0(v) e^{-\int_0^u \frac{|\zeta|^2}{r\nu}(u', v) du'} \label{firstorder_varpi} \\
&+ \int_0^u e^{-\int_s^u \frac{|\zeta|^2}{r\nu} (u', v)du'} \br{\frac{|\zeta|^2}{2\nu} \br{1 + \frac{Q^2}{r^2} - \frac{\Lambda}{3}r^2} - q Q \Im \br{\overline{\zeta}\phi} + \frac{\nu}{2}m^2 r^2 |\phi|^2}(s, v) ds, \nonumber  \\
\theta(u, v) &= \theta_0(v) e^{-i q \int_0^u A_u(u', v) du'} \label{firstorder_theta} \\
&+ \int_0^u e^{-i q \int_s^u A_u(u', v) du'} \br{-\frac{\zeta \lambda}{r} + \frac{\nu \kappa \phi}{r} \br{m^2 r^2 - i q Q}}(s, v) ds, \nonumber \\
\kappa(u, v) &= e^{\int_0^u \frac{|\zeta|^2}{r \nu} (u', v) du'}, \label{firstorder_kappa} \\
\phi(u, v) &= \phi_0(v)e^{-i q \int_0^u A_u (u', v)du'} + \int_0^u \frac{\zeta}{r}(s, v) e^{-i q \int_s^u A_u(u', v)du'} ds, \label{firstorder_phi} \\
Q(u, v) &= Q_0(v) - q \int_0^u r(u', v) \Im(\phi \overline{\zeta})(u', v) du', \label{firstorder_Q}\\
r(u, v) &= r_0(u) + \int_0^v \lambda(u, v')dv', \label{firstorder_r} \\
\nu(u, v) &= - \exp \br{-\int_0^v \frac{2\kappa}{r^2}\br{\frac{Q^2}{r} + \frac{m^2r^3}{2} |\phi|^2 + \frac{\Lambda}{3}r^3 - \varpi} (u, v')dv' } \label{firstorder_nu} \\
\zeta(u, v) &= \zeta_0(u)-\int_0^v \br{\frac{\theta}{r}\nu - \frac{\nu \kappa \phi}{r}\br{m^2 r^2 + i q Q}}(u, v')dv', \label{firstorder_zeta} \\
A_u(u, v) &=   2 \int_0^v \frac{Q \nu \kappa}{r^2}(u, v') dv'. \label{firstorder_A}
\end{align}
In the above formulas, we exploited the gauge choice $A_{u, 0} \equiv 0$ and the fact that $\nu_0 \equiv -1$ and $\kappa_0 \equiv 1$. The full details of the proof can be found in \cite{mythesis}.
\end{proof}

\begin{remark}[maximal past sets]
A set $\mathcal{P} \subset [0, U] \times [v_0, +\infty)$ is a \textbf{past set} if $[0, u] \times [v_0, v] \subset \mathcal{P}$ for every $(u, v) \in \mathcal{P}$.
As noticed in \cite[Theorem 4.4]{CGNS1}, every solution to the characteristic IVP with initial data prescribed on $[0, U] \times [v_0, +\infty)$, for some $v_0 \ge 0$, and satisfying \hyperref[section:assumptions]{assumptions} \textbf{(A)}, \textbf{(B)}, \textbf{(C)} and \textbf{(D)}, can be extended to a maximal past set $\mathcal{P} \supset [0, U] \times \{v_0\} \cup \{0\} \times [v_0, +\infty)$.
\end{remark}

\begin{proposition}[continuous dependence with respect to the initial data] \label{prop: contdependence}
Let us consider two solutions $f$ and $\tilde f$ of the first-order system, with initial data $(r_0,$ $\nu_0,$ $\lambda_0,$ $\varpi_0,$ $\theta_0,$ $\zeta_0,$ $\kappa_0,$ $\phi_0,$ $Q_0,$ $A_{u, 0})$ and $(\tilde r_0,$ $\tilde \nu_0,$ $\tilde \lambda_0,$ $\tilde \varpi_0,$ $\tilde \theta_0,$ $\tilde \zeta_0,$ $\tilde \kappa_0,$ $\tilde \phi_0,$ $\tilde Q_0,$ $\tilde A_{u, 0})$, respectively. Assume that the two solutions are defined in $\mathcal{D} \coloneqq [0, U] \times [0, V]$ for some $U, V \in \bbR_0^+$ and define the two quantities
\begin{align*}
d(U, V) &\coloneqq \|f - \tilde{f}\|_{L^{\infty}(\mathcal{D})},
\end{align*}
and 
\begin{align*}
d_0(U, V) &\coloneqq \left( \|r_0- \tilde r_0\| + \|\nu_0 -\tilde \nu_0\| + \|\zeta_0 - \tilde \zeta_0\| + \|A_{u,0} - \tilde A_{u, 0}\|\right)_{L^{\infty}([0, U])} \\ 
& + \left( \|\lambda_0 - \tilde \lambda_0 \| + \|\varpi_0 -\tilde \varpi_0 \| + \|\theta_0 - \tilde \theta_0 \| +
\|\kappa_0 - \tilde \kappa_0 \| + \| \phi_0 - \tilde \phi_0 \| + \|Q_0 - \tilde Q_0\|  \right)_{L^{\infty}([0, V])}.
\end{align*}
Then, if $d_0(U, V)$ is sufficiently small, we have:
\[
d(U, V) \le C d_0(U, V),
\]
where $C$ is a positive constant depending on $U, V$ and on
\[
N(\mathcal{D}) \coloneqq \|f \|_{L^{\infty}(\mathcal{D})} + \left \| \frac{1}{r} \right \|_{L^{\infty}(\mathcal{D})} + \left \|\frac{1}{\nu}  \right\|_{L^{\infty}(\mathcal{D})}.
\]
\end{proposition}
\begin{proof}
See \cite{mythesis}.
\end{proof}

\begin{remark}[On the equivalence between the first-order and second-order system] \label{rmk:equivalence}
Under \hyperref[section:assumptions]{assumptions} \textbf{(A)}, \textbf{(B)}, \textbf{(C)}, \textbf{(D)}, the second-order system \eqref{raych_u}--\eqref{klein_gordon}, \eqref{duQ}--\eqref{dvQ} implies the first-order system \eqref{dur}--\eqref{algconstr}.

On the other hand, if $(r, \Omega^2, \phi, Q)$ solves the second-order system (namely $r$, $\Omega^2$, $\phi$, $Q$ satisfy the system, they are continuous functions and all derivatives in the second-order system are continuous), then the Raychaudhuri equations and the wave equation for $r$ imply that $r \in C^2$. Such a property is guaranteed for a solution to the first-order system if we additionally require that
\[
\textbf{(B2) }
\nu_0, \lambda_0 \text{ and } \kappa_0 \text{ are } C^1 \text{ functions in their respective domains.}
\]
A straightforward adaptation of \cite[Section 6]{CGNS1}  reveals that assumption \textbf{(B2)} is in fact sufficient to show that a solution to the first-order system solves the second-order system, in the above-mentioned sense. 

Therefore, we say that the first-order system and the second-order system are equivalent if assumptions \textbf{(A)}, \textbf{(B)}, \textbf{(B2)}, \textbf{(C)} and \textbf{(D)} are satisfied.
\end{remark}

\begin{remark}[general properties of solutions] \label{rmk:signs}
As a consequence of \eqref{firstorder_lambda}-\eqref{firstorder_A}, we have that, for a solution $f = (\lambda,$ $\varpi,$ $\theta,$ $\kappa,$ $\phi,$ $Q,$ $r,$ $\nu,$ $\zeta,$ $A_u)$ of the first order system in a past set $\mathcal{P}$:
\begin{itemize}
\item $\kappa$ is positive,
\item $\nu$ is negative.
\end{itemize}
We notice that the renormalized Hawking mass $\varpi$ is generally not a monotonic function, differently from the case of the Einstein-Maxwell-(\textbf{real}) scalar field system. However, we have
\[
\partial_u \br{\varpi - \frac{Q^2}{2r}} < 0 \quad \text{ and } \quad \partial_v \br{\varpi - \frac{Q^2}{2r}} \ge 0
\]
in the region $\{(u, v) \colon \, \lambda(u, v) \ge 0\}$.
\end{remark}

\begin{remark}[An extension criterion] \label{rmk:extension_principle}
Related to the well-posedness of the characteristic IVP, one can investigate the conditions that allow to extend the domain of existence of a solution. 
Results of this sort were obtained in \cite{Dafermos_2005_trapped}, \cite{CGNS1} and \cite{Kommemi}. In \cite{Dafermos_2005_trapped, Kommemi}, in particular, the lack of extensions was used to characterize the \textit{first singularities} possibly present in the considered spacetimes, therefore addressing that problem of spacetime predictability which, in more general settings, is captured by the weak cosmic censorship conjecture. In particular, the following result was proved in \cite{Kommemi} for the Einstein-Maxwell-(charged) scalar field system with $\Lambda = 0$: given initial data prescribed on $[0, U] \times \{0\} \cup \{0 \} \times [0, V]$ for some $U, V \in \mathbb{R}_0^+$, where the initial data satisfy the assumptions of the local existence theorem \ref{thm:localexistence}, a solution $f$ defined in 
\[
\mathcal{D} \coloneqq [0, U') \times [0, V'),
\]
 with $0 < U' < U$ and $0 < V' < V$, can be extended along both the future ingoing and outgoing directions if we can control the $L^{\infty}$ norm of the solution and of the inverse of $r$ and $\Omega^2$. More precisely, the result is formulated in a global sense in terms of Cauchy data. In \cite{Kommemi}, it was also shown that for such extension to exist, two conditions are sufficient (provided that singularities emanating from spacetime endpoints are avoided): 1) that the area-radius function can be estimated in $\mathcal{D}$ from below and above by positive constants, and 2) that the spacetime volume of $\mathcal{D}$, i.e.
\[
\int_{\mathcal{D}} \Omega^2 du dv
\]
is bounded from above. 

The above result holds for the second-order system \eqref{raych_u}--\eqref{klein_gordon}, \eqref{duQ}--\eqref{dvQ} since the presence of $\Lambda$ in \eqref{waveeqn_r} does not require any substantial change in the proofs of the above statements. This is enough for the purposes of the present paper (see remark \ref{rmk:equivalence}). An additional proof that only requires the level of regularity of the first-order system \eqref{dur}--\eqref{algconstr} is presented in \cite{mythesis}.
\end{remark}

\section{Existence and stability of the Cauchy horizon} \label{section:stability}
We recall that a set $\mathcal{P} \subset [0, U] \times [v_0, +\infty)$ is called a \textbf{past set} if $J^{-}(u, v) \coloneqq [0, u] \times [v_0, v]$ is contained in $\mathcal{P}$ for every $(u, v)$ in $\mathcal{P}$. Given a set $S \subset [0, U] \times [v_0, +\infty)$, we define 
\[
J^{-}(S) \coloneqq \bigcup_{(u, v) \in S} J^{-}(u, v),
\]
and analogous definitions hold for $J^{+}(S)$, $I^{-}(S)$ and $I^{+}(S)$ (where $I^{-}(u, v)\coloneqq [0, u) \times [v_0, v)$).

 In the following, we work with a solution to the characteristic IVP of section \ref{section:ivp}, defined in the maximal past set $\mathcal{P}$ containing
\[
\mathcal{D}_0 \coloneqq [0, U] \times \{v_0\} \cup \{0\}  \times [v_0, +\infty), \quad \text{ for some } 0 < U < +\infty,\, v_0 \ge 0.
\]
The existence of $\mathcal{P}$ is guaranteed by the well-posedness results proved in section \ref{section:wellposedness}. In turn, its maximality is to be interpreted in the sense that the solution cannot be defined on any larger past set (see also remark \ref{rmk:extension_principle}).
From now on, we suppose that \hyperref[section:assumptions]{assumptions} \textbf{(A)}, \textbf{(B)}, \textbf{(B2)} (see remark \ref{rmk:equivalence}),  \textbf{(C)}, \textbf{(D)}, \textbf{(E)}, \textbf{(F)} and \textbf{(G)} hold. In particular, we are going to use the equations from both the first-order and second-order systems, depending on the most convenient choice.

In this context, $\mathcal{P}$ corresponds to a region containing the event horizon and extending inside the dynamical black hole under investigation. 
In the course of the following proofs, we actually work with solutions defined in 
\[
\mathcal{P} \cap \{v \ge v_1\},
\]
for some $v_1 \ge v_0$. We require (finitely many times) that the values of $U$ and $v_1$ are, respectively, sufficiently small and sufficiently large. A posteriori, this implies that our results hold in a region adjacent to both the event horizon and the Cauchy horizon of the dynamical black hole under investigation.  We will highlight the steps where we constrain the values of $U$ and $v_1$. This occurs finitely many times and depends only on the $L^{\infty}$ norm of the initial data and possibly on $\eta$, $R$ (see proposition \ref{prop: redshift}), $Y$ (see proposition \ref{prop: noshift}), and $\varepsilon$ (see proposition \ref{prop: blueshift}).  These quantities can ultimately be chosen in terms of the initial data.

In this section, we are also going to prove that the \textbf{apparent horizon}
\[
\mathcal{A} \coloneqq \{(u, v) \in [0, U] \times [v_0, +\infty) \colon\, \lambda(u, v) = 0\}
\]
is a non-empty $C^1$ curve for large values of the $v$ coordinate. It is also convenient to define the \textbf{regular region} 
\[
\mathcal{R} \coloneqq \{(u, v) \in [0, U] \times [v_0, +\infty) \colon \, \lambda(u, v) > 0 \},
\]
which, as we will see, extends from the event horizon to part of the redshift region. The remaining part of the  black hole interior is occupied by  the \textbf{trapped region}
\begin{equation} \label{def_trapped}
\mathcal{T} \coloneqq \{(u, v) \in [0, U] \times [v_0, +\infty) \colon \, \lambda(u, v) < 0 \}.
\end{equation}
For the system under analysis, the regular region is a priori non-empty, differently from the case of the Reissner-Nordstr{\"o}m-de Sitter solution. In the latter case, indeed, every 2-sphere in the interior of the black hole region is a trapped surface and, furthermore, the apparent horizon coincides with the event horizon.

\begin{remark}[on the main constants] \label{remark:mainconstants}
Throughout the bootstrap procedure, we use several auxiliary, positive quantities:
\begin{itemize}
\item Redshift region (proposition \ref{prop: redshift}): $\eta$, $\delta = \delta(\eta)$, $R =  R(\delta)$,
\item No-shift region (proposition \ref{prop: noshift}): $\varepsilon = \varepsilon(\beta)$, $\Delta$, $Y = Y(\Delta)$,
\item Early blueshift region (see \eqref{def_beta}): $\beta$.
\end{itemize}
We stress that, in proposition \ref{prop: blueshift}, we choose $\Delta$ in function of $\varepsilon$. However, as explained in remark \ref{rmk:not_necessary}, an analogous proof holds even if we choose $\Delta$ independently from $\varepsilon$.
All above constants can be ultimately defined in terms of the initial data.

In particular, we have:
\begin{itemize}
\item $\delta \to 0$ as $\eta \to 0$ and $R \to r_+$ as $\delta \to 0$ (see proposition \ref{prop: redshift}),
\item $Y \to r_{-}$ as $\Delta \to 0$ (see proposition \ref{prop: blueshift}),
\item $\varepsilon \to 0$ as $\beta \to 0$ (see lemma \ref{lemma:lambdanu_blueshift}).
\end{itemize}
\end{remark}

\subsection{Event horizon} \label{sec:EH}
Recall that $\varpi_+$, $Q_+$ and $\Lambda$ are the parameters of the final sub-extremal Reissner-Nordstr{\"o}m-de Sitter black hole (see \hyperref[section:assumptions]{assumption} \textbf{(F)}), while $r_+$ and $K_+$ denote the radius and the surface gravity, respectively, associated to its event horizon. 
\begin{proposition}[bounds along the event horizon] \label{prop: boundsEH}
For every $v \ge v_0$, we have:
\begin{align}
0 < r_+ - r(0, v) &\le C_{\mathcal{H}}e^{-2sv}, \label{bound_r_eventhorizon}\\
0 < \lambda(0, v) &\le C_{\mathcal{H}}e^{-2sv}, \label{bound_lambda_eventhorizon} \\
  \frac{e^{2K_+ v}}{C_{\mathcal{H}}} \le |\nu|(0, v) &\le C_{\mathcal{H}} e^{2K_+ v}, \label{bound_nu_eventhorizon}\\
 |\varpi(0, v) - \varpi_+| &\le C_{\mathcal{H}}e^{-2sv}, \label{bound_varpi_eventhorizon} \\
 |Q(0, v) - Q_+| & \le C_{\mathcal{H}} e^{-2sv}, \label{bound_Q_eventhorizon} \\
|K(0, v) - K_+| &\le C_{\mathcal{H}} e^{-2sv},  \label{bound_K_eventhorizon} \\
|\partial_v \log \Omega^2(0, v) - 2K(0, v) | &\le C_{\mathcal{H}} e^{-2sv}, \label{bound_dvlog_eventhorizon} \\
 |D_u \phi|(0, v) &\le \begin{cases}
C_{\mathcal{H}} |\nu|(0, v) e^{-sv}, &\mathrm{if }\,\, s < 2K_+, \\
C_{\mathcal{H}} (v - v_0), &\mathrm{if }\,\, s = 2K_+,\\
C_{\mathcal{H}}, & \mathrm{if }\,\,  s > 2K_+,
\end{cases} \label{bootstrap_firstbound}\\
|\phi|(0, v) + |\partial_v \phi|(0, v) &\le C_{\mathcal{H}} e^{-sv}. \label{final_Pricelaw}
\end{align}
where $C_{\mathcal{H}}$ is a positive constant depending only on the initial data.
\end{proposition}
\begin{proof}
The proof exploits the decay due to the redshift effect, similarly to the proof of \cite[proposition 4.4]{VdM1}. In the current case, however, the competition between the redshift effect and the exponential Price law is evident in the estimate for $|D_u \phi|$ (see already remark \ref{rmk:interplay}).

In the following, the letter $C$ will denote a positive constant depending uniquely on the initial data.  Moreover, we will exploit the assumptions \eqref{assumption:kappa} and \eqref{absence_antitrapped} on $\kappa$ and $\nu$, respectively, multiple times in the next computations. 

\textbf{Preliminary bounds and exponential growth of $\bm{|\nu|}$}: first, we notice that \eqref{assumption:r} and \eqref{assumption:lambda}  imply
\begin{equation} \label{preliminarybound_r_eventhorizon}
0 < r(0, v_0) \le r(0, v) < r_+ < + \infty, \quad \forall \, v \ge v_0.
\end{equation}
It will also be useful to write \eqref{dulambda_kappa}, when evaluated on $\mathcal{H}^+$, as
\begin{equation} \label{lognu}
\partial_v \log|\nu|(0, v) = 2K(0, v) - m^2 r(0, v) |\phi|^2(0, v).
\end{equation}

Now, given $\epsilon > 0$, expressions \eqref{assumption:expPricelaw}, \eqref{assumption:surface_grav} and the boundedness of $r$ provide a sufficiently large $v_1 > v_0$ such that
\begin{equation} \label{preliminarybound_kappa}
\left | 2K(0, v) - m^2 r(0, v) |\phi|^2(0, v) - 2K_+ \right| < \epsilon, \quad \forall \, v \ge v_1.
\end{equation}
This estimate will be improved during the next steps of the current proof, but for the moment we can use it in \eqref{lognu} and require $\epsilon$ to be sufficiently small to obtain
\begin{equation} \label{lognu_strictly_pos}
0 < K_+ < 2K_+ - \epsilon \le  \partial_v \log |\nu|(0, v)  \le 2K_+ + \epsilon, \quad \forall\, v \ge v_1,
\end{equation}
and so, by integrating and due to \eqref{assumption:nu}:
\begin{equation} \label{preliminarybound_nu}
 e^{(2K_+ -\epsilon)(v-v_1)} \le |\nu|(0, v) \le e^{(2K_+ + \epsilon)(v- v_1)}, \quad \forall \, v \ge v_1.
\end{equation}

Note that $\Omega^2(0, v) = -4 \nu(0, v)$ for every $v \ge v_0$, by \eqref{def_kappa}. So,  \eqref{assumption:expPricelaw}, \eqref{lognu} and the boundedness of $r$ entail
\[
\left | \partial_v \log \Omega^2 (0, v) - 2K(0, v) \right| \le C e^{-2sv}.
\]

To obtain \eqref{bound_lambda_eventhorizon}, we  first show that  $\frac{\lambda}{\nu}$ admits a finite limit as $v \to +\infty$, and that such a limit is zero. Indeed, by \hyperref[section:assumptions]{assumptions} \textbf{(A)} and \textbf{(F)}, we have that $\lim_{v \to +\infty} \lambda(0, v) = \lim_{v \to +\infty} (1-\mu)(0, v) = 0$.
Moreover, due to \eqref{preliminarybound_nu} and to the fact that $\lambda_{|\mathcal{H}^+} > 0$:
\begin{equation} \label{limit_lambdaonu}
\lim_{v \to +\infty}\frac{\lambda}{\nu}(0, v) = 0.
\end{equation}

\textbf{Exponential decays:}
We now focus on the proof of the decay of $\lambda$.  Using \eqref{raych_v} and \eqref{def_kappa} we cast one of the Raychaudhuri equations as
\[
\partial_v \br{\frac{\lambda}{\nu \kappa}} = -r \frac{|\partial_v \phi|^2}{\nu \kappa}.
\]
Using the above, assumptions \eqref{assumption:expPricelaw} and \eqref{assumption:lambda}, and exploiting  \eqref{preliminarybound_r_eventhorizon} and \eqref{limit_lambdaonu}, we have:
\begin{align*}
0 &< \lambda(0, v) = \nu(0, v) \int_{+\infty}^v \partial_v \br{\frac{\lambda}{\nu}}(0, v') dv' \\
& =
\nu(0, v) \int_v^{+\infty} \frac{r |\partial_v \phi|^2}{\nu} (0, v') dv'  
 \le C \nu(0, v) \int_v^{+\infty} \frac{e^{-2sv'}}{\nu(0, v')} dv' .
\end{align*}
We now multiply and divide by $K_+$, which is a positive quantity, and use that $K_+ < \partial_v \log |\nu|(0, v')$ (see \eqref{lognu_strictly_pos}), the fact that $v \mapsto \nu^{-1}(0, v)$ is an increasing function (indeed $\partial_v \nu^{-1} = |\nu|^{-1} \partial_v \log |\nu| > 0$) and \eqref{preliminarybound_nu} to write
\begin{align}
\nu(0, v) \int_v^{+\infty} \frac{e^{-2sv'}}{\nu(0, v')} dv' &\le \frac{|\nu|(0, v)}{K_+} \int_v^{+\infty} e^{-2sv'} \partial_v\br{\frac{1}{\nu}}(0, v')dv' \nonumber \\
 &\le C |\nu|(0, v)e^{-2sv} \int_v^{+\infty} \partial_v\br{\frac{1}{\nu}}(0, v')dv'  = C e^{-2sv}, \label{lambda_eventhorizon}
\end{align}
which therefore gives $0 < \lambda(0, v) \le C e^{-2sv}$ for every $v \ge v_0$.

After recalling \eqref{assumption:Q}, we integrate \eqref{maxwelleqnv}  along the event horizon and use \eqref{preliminarybound_r_eventhorizon} and the exponential Price law \eqref{assumption:expPricelaw} to write
\begin{equation} \label{bound_Q_eventhorizon_prel}
\left | Q(0, v) - Q_+ \right| = C \left | \int_v^{+\infty} r^2(0, v') \Im(\phi \overline{\partial_v \phi})(0, v') dv' \right| \le C e^{-2sv}, \quad \forall\, v \ge v_0.
\end{equation}
This gives \eqref{bound_Q_eventhorizon}.

We can use this after integrating \eqref{dvvarpi} (recall that $\theta = r \partial_v \phi$), together with \eqref{assumption:expPricelaw}, \eqref{assumption:varpi},   \eqref{bound_Q_eventhorizon_prel}, \eqref{preliminarybound_r_eventhorizon} and the decay of $\lambda$ to see that
\begin{equation} \label{bound_varpi_eventhorizon_prel}
|\varpi(0, v) - \varpi_+| \le Ce^{-2sv}, \quad \forall \, v \ge v_0,
\end{equation}
therefore giving \eqref{bound_varpi_eventhorizon}.

By integrating \eqref{lambda_eventhorizon} and using \eqref{assumption:r}, it also follows that
\[
0 < r_+ - r(0, v) \le C e^{-2sv}.
\]
Due to definition \eqref{def_surfacegrav} and to the decays proved for $Q$ and $\varpi$, this is also sufficient to prove 
\begin{equation} \label{bound_K_eventhorizon_prel}
|K(0, v)-K_+| \le C  e^{-2 s v}, \quad \forall\, v \ge v_0.
\end{equation} 
This result also allows us to improve \eqref{preliminarybound_kappa} and, following the same steps leading to \eqref{lognu_strictly_pos} and \eqref{preliminarybound_nu}, to obtain:
\begin{equation} \label{finalbound_nu}
|\nu|(0, v) \sim e^{2K_+(v - v_0)}, \quad \forall \, v \ge v_0.
\end{equation}

\textbf{Bounds on $\bm{|D_u \phi|}$}:
we use \eqref{waveeqn_phi3_r}, \eqref{assumption:expPricelaw}, \eqref{bound_Q_eventhorizon_prel} and  \eqref{preliminarybound_r_eventhorizon}  to write
\begin{equation} \label{duphi_preliminary_bound}
\left |\partial_v \br{r D_u \phi} \right|(0, v) = \left| -\nu \partial_v \phi + m^2 \nu r \phi + i \frac{q Q \nu}{r} \phi \right|(0, v) \le C|\nu|(0, v) e^{-sv}.
\end{equation}
Now, let us first assume that $s < 2K_+$ and define the constant $a \coloneqq K_+ - \frac{s}{2} > 0$. There exists $V \ge v_0$, depending uniquely on the initial data, such that,  using \eqref{dulambda_kappa}, \eqref{assumption:expPricelaw} and \eqref{bound_K_eventhorizon_prel}, we obtain
\[
\partial_v \br{|\nu| e^{-sv}}_{|\mathcal{H}^+} = |\nu|(0, v) e^{-sv} \br{2K(0, v) - m^2 r(0, v) |\phi|^2(0, v) - s} > a |\nu|(0, v) e^{-sv},
\]
for every $v \ge V$.
We can use this inequality when we integrate \eqref{duphi_preliminary_bound}:
\begin{align}
|r D_u \phi|(0, v) &\le C\br{1 + \int_{v_0}^v |\nu|(0, v') e^{-s v'} dv'} < C\br{1 + \frac{1}{a}\int_{V}^v \partial_v \br{|\nu|(0, v') e^{-s v'}} dv'} \nonumber \\
& \le C \br{1 + \frac{1}{a}|\nu|(0, v)e^{-sv}}, \label{duphi_preliminary_bound2}
\end{align}
where we emphasize that $C$ depends uniquely on the initial data.
Using \eqref{finalbound_nu} and the boundedness of $r$, we conclude that 
\[
|D_u \phi|(0, v) \le C |\nu|(0, v) e^{-sv} \lesssim  e^{(2K_+ - s)v}, \quad \forall\, v \ge v_0,
\]
when $s < 2K_+$.

On the other hand, when $s > 2K_+$, \eqref{finalbound_nu} gives:
\[
|r D_u \phi|(0, v) \le  C\br{1 + \int_{v_0}^v  e^{(2K_+ -s) v'} dv'} \le \tilde{C},
\]
for some $\tilde{C}$ determined by the initial data and for every $v \ge v_0$. The final estimate follows from \eqref{preliminarybound_r_eventhorizon}.

When $s= 2K_+$,  we can integrate \eqref{duphi_preliminary_bound} and use \eqref{finalbound_nu} to get
\[
|r D_u \phi|(0, v) \le C \br{1 + \int_{v_0}^v |\nu|(0, v') e^{-2K_+ v'} dv'} \le \tilde{C}(v - v_0), \quad \forall \, v \ge v_0,
\]
for some $\tilde{C} > 0$ depending on the initial data. 
The final estimate follows again from the boundedness of $r$. 

Finally, \eqref{final_Pricelaw} follows from  \eqref{assumption:expPricelaw}. We then choose $C_{\mathcal{H}}$ as the largest of the previous positive constants depending on the initial data.
\end{proof}

\begin{remark}[Redshift effect] \label{rmk:interplay}
From the point of view of a family of observers crossing the event horizon $\mathcal{H}^+$, the energy associated to the null geodesics ruling $\mathcal{H}^+$ decays as $e^{-2K_+ v}$ due to the redshift effect.

This is the same rate  at which we found the \textbf{geometric quantity} $\left| \nu^{-1} D_u \phi \right|$ to decay for $s > 2K_+$ along the event horizon (see \eqref{bootstrap_firstbound}). Here, the constant $s$ is the same as in \eqref{assumption:expPricelaw}. For $s < 2K_+$, the quantity $\left| \nu^{-1} D_u \phi \right|$ decays as $e^{-sv}$ or faster, which is the same rate dictated by the Price law upper bound \eqref{assumption:expPricelaw}. From these two different rates, we deduce the physical relevance of two competing phenomena along the event horizon: the redshift effect and the decay of the scalar field $\phi$.
\end{remark}

\subsection{Level sets of the area-radius function} \label{section:curves}
The strategy of the next proofs is based on a partition of the past set $\mathcal{P}$. In the two-dimensional quotient space that we are considering, such subsets are mainly separated by curves where the area-radius function is constant.  We review the main properties of the sets 
\begin{equation} \label{def_constantR}
\Gamma_{\varrho} \coloneqq \{(u, v) \in \mathcal{P}\, \colon\, r(u, v) = \varrho\}
\end{equation}
for some $\varrho \in (0, r_+)$ (see also \cite{CGNS4}).

\begin{lemma}[properties of $\Gamma_{\varrho}$] \label{lemma:properties_constantr}
Assume that $\Gamma_{\varrho} \ne \varnothing$. Then, the following set equality holds:
\[
\Gamma_{\varrho} = \{ (u_{\varrho}(v), v) \, \colon \, v \in [v_1, +\infty) \},
\] 
for some $v_1 \ge v_0$ and for some $C^1$ function $u_{\varrho} \colon [v_1, +\infty) \to \bbR$ such that $r(u_{\varrho}(\cdot), \cdot) \equiv \varrho$. Moreover, let 
\begin{equation} \label{def_underv}
\underline{v} \coloneqq \inf \{v \ge v_1 \colon \lambda(u_{\varrho}(v), v) \le 0 \}.
\end{equation}
Then, the following properties are satisfied:
\begin{description}
\item[$\bm{\mathrm{P}\Gamma 1.}$] $\Gamma_{\varrho}$ is a $C^1$ curve,
\item[$\bm{\mathrm{P}\Gamma 2.}$]  We have
\begin{equation} \label{du_R}
u'_{\varrho}(v) = -  \frac{\lambda(u_{\varrho}(v), v)}{\nu(u_{\varrho}(v), v)},
\end{equation}
\item[$\bm{\mathrm{P}\Gamma 3.}$]  $\lambda(u_{\varrho}(v), v) \le 0$ for every $v \ge \underline{v}$.
\end{description}
\end{lemma}

\begin{lemma}[conditions to bound $\lambda$ away from zero] \label{lemma:lambda_along_curve}
Let $\varrho \in (r_{-}, r_+)$. Assume that $\Gamma_{\varrho} \ne \emptyset$ and that for every $\epsilon > 0$ there exists $v_1 \ge v_0$ and $c > 0$ such that
\[
|Q(u_{\varrho}(v), v) - Q_+| + |\varpi(u_{\varrho}(v), v) - \varpi_+|  < \epsilon, \quad \forall \, v \ge v_1,
\]
and 
\[
0 < c \le \kappa(u_{\varrho}(v), v) \le 1, \quad \forall\, v \ge v_1.
\]
Then, there exist positive constants $C_1$ and $C_2$ such that:
\[
-C_2 \le \lambda(u_{\varrho}(v), v) \le - C_1 < 0, \quad \forall\, v \ge v_1.
\]
\end{lemma}
\begin{proof}
Assume that $\epsilon$ is sufficiently small. We then use \eqref{def_mu} to write
\begin{align*}
(1 - \mu)(u_\varrho(v), v) &= \br{1 - \frac{2}{\varrho}\br{\varpi-\varpi_+}-\frac{2}{\varrho}\varpi_+ + \frac{Q^2-Q_+^2}{\varrho^2}+\frac{Q_+^2}{\varrho^2} - \frac{\Lambda}{3}\varrho^2}(u_\varrho(v), v) \\
&= (1-\mu)(\varrho, \varpi_+, Q_+) + O(\epsilon),
\end{align*}
for every $v \ge v_1$.
After inspecting the plot of $(1-\mu)(\cdot, \varpi_+, Q_+)$ (see, e.g. \cite[section 3]{CGNS2}), using the smallness of $\epsilon$, \eqref{algconstr} and  the bounds on $\kappa(u_{\varrho}(\cdot), \cdot)$, we have:
\[
-C_2 \le \lambda(u_\varrho(v), v)=[\kappa(1-\mu)](u_\varrho(v), v) \le - C_1 < 0,
\]
for some positive constants $C_1$ and $C_2$.
\end{proof}

We use the definitions \eqref{def_trapped} and  \eqref{def_underv}  of $\mathcal{T}$ and  $\underline{v}$, respectively, for the next result.

\begin{corollary}[entering the trapped region] \label{corollary_trappedregion}
If the assumptions of lemma \ref{lemma:lambda_along_curve} are satisfied, the curve $\Gamma_{\varrho}$ is spacelike for large values of the $v$ coordinate. Moreover, $\underline{v} < +\infty$ and $\Gamma_{\varrho} \cap \{v \ge v_1\} \subset \mathcal{T}$.
\end{corollary}
\begin{proof}
The conclusion follows from the negativity of $\nu$ in $\mathcal{P}$ and of $\lambda$ in $\Gamma_{\varrho} \cap \{v \ge v_1\}$, for some $v_1 \ge v_0$ (see lemma \ref{lemma:lambda_along_curve}).
\end{proof}

\subsection{Redshift region} \label{sec:RR}
For some fixed $R$, chosen close enough to $r_+$ (see \hyperref[section:assumptions]{assumption} \textbf{(F)}), we now study the region $J^{-}(\Gamma_R)$, after proving that it is non-empty.\footnote{By definition \eqref{def_constantR}, $\Gamma_R = \{r = R\}$ is a (possibly empty) subset of the maximal past set $\mathcal{P}$. Since $\nu < 0$ in $\mathcal{P}$ (see remark \ref{rmk:signs}) and due to the extension criterion (remark \ref{rmk:extension_principle}), the curve being empty means that $R$ is smaller than any value of the area-radius function in $[0, U] \times [v_0, +\infty)$. We will see that this is never the case, i.e.\ for every choice of $v_1 \ge v_0$ and $U$, we can find a curve in $[0, U] \times [v_1, +\infty)$ such that $r = R$ along it.} We denote this set as the \textbf{redshift region} (see also \cite{Dafermos_2003, LukOh1, VdM1}).
Most of the estimates showed in proposition \ref{prop: boundsEH} propagate throughout the redshift region, and the interplay between  the decay of the scalar field and the redshift effect that we observed along the event horizon $\mathcal{H}^+$ (see remark \ref{rmk:interplay}) is still present.

To obtain quantitative bounds in this region, it is sufficient to bootstrap the estimates on $|\phi|$, $|\partial_v \phi|$, $|D_u \phi|$, $\kappa$ and $|\nu|$ from $\mathcal{H}^+$, and exploit these to achieve the remaining bounds. In particular, we employ the results of proposition \ref{prop: boundsEH} and use the smallness of $r_+ - r$ to close the bootstrap inequalities.  The proof is an adaptation of the one in \cite[proposition 4.5]{VdM1} to the case of exponential estimates. In particular, this requires a careful treatment of the Gronwall argument used to close the bootstrap for $|D_u \phi|$. The estimates that we obtain depend on the integrable quantity $u \Omega^2(0, v) \sim u e^{2K_+ v}$ (see proposition \ref{prop: boundsEH}).
\begin{proposition}[propagation of estimates by redshift] \label{prop: redshift}
Given a fixed $\eta \in (0, K_+)$, let $\delta \in (0, \eta)$ be small compared to the initial data and to $\eta$. Moreover, let $R$ be a fixed constant such that 
\[
0 < r_+ - R < \delta
\] 
and consider the function
\begin{equation} \label{def_cs}
c(s) \coloneqq \begin{cases}
s, & \text{if } 0 < s \le 2K_+ - \eta, \\
2K_+ - \eta, &\text{if } s \ge 2K_+ - \eta.
\end{cases}
\end{equation}
Then, given $v_1 \ge v_0$ large, whose size depends on the initial data, we have $J^{-}(\Gamma_R) \cap \{v \ge v_1\} \ne \varnothing$. Moreover, for every $(u, v) \in J^{-}(\Gamma_R) \cap \{v \ge v_1\}$, the following inequalities hold:
\begin{align}
2\br{r(0, v) - r(u, v)}  \le u \Omega^2(0, v) &\le C \delta, \label{redshift_uomega} \\
0 < r_+ -  r(u, v) &\le Ce^{-2sv} + u \Omega^2(0, v),  \label{redshift_r_final}\\
 |\lambda|(u, v) &\le Ce^{-2sv} + u \Omega^2(0, v), \label{redshift_lambda_final}\\
|\nu|(u, v) &\sim \Omega^2(0, v), \label{redshift_nu_final} \\
|\varpi(u, v) - \varpi_+| &\le C e^{-2c(s) v}, \label{decay_redshift_varpi}\\
|Q(u, v) - Q_+| &\le C e^{-2c(s)v}, \label{decay_redshift_Q} \\
|K(u, v) - K_+| &\le C\br{e^{-2c(s)v} + u\Omega^2(0, v)},\\
| \partial_v \log \Omega^2(u, v) - 2K(u, v) | &\le C e^{-2 c(s) v}, \label{redshift_dvlog_final} \\
|D_u \phi|(u, v) &\le C |\nu|(u, v) e^{-c(s)v}, \label{decay_redshift_Duphi} \\
|\phi|(u, v) + |\partial_v \phi|(u, v) &\le C e^{-c(s)v},
\end{align}
for a positive constant $C$ that  depends only on the initial data and on $\eta$.
Additionally, we have:
\begin{equation} \label{redshift_lambdasign}
\partial_u \lambda(u, v) < 0, \quad \forall \, (u, v) \in J^{-}(\Gamma_R) \cap \{v \ge v_1 \}.
\end{equation}
\end{proposition}
\begin{proof}
 In the following, $C_{\mathcal{H}}$ denotes the constant appearing in the statement of proposition \ref{prop: boundsEH}. 
We will use the notation 
\[
C = C(N_{\text{i.d.}}, \eta)
\]
to denote a positive constant depending on (a suitable norm of) the initial data of our IVP and on $\eta$. We will use the same letter $C$ to denote possibly different constants, when the exact expression of such constants is not important for the bootstrap. To close the bootstrap argument, we will require  $\delta$ to be sufficiently small with respect to $N_{\text{i.d.}}$ and $\eta$.  Moreover, we use that
 $0 < \kappa \le 1$ and $\nu < 0$
 in $\mathcal{P}$ (see \eqref{firstorder_kappa} and remark \ref{rmk:signs}) and that $\lambda_{|\mathcal{H}^+} > 0$ (see \hyperref[section:assumptions]{assumption} \textbf{(G)}) multiple times.

\textbf{Setting up the bootstrap procedure}:
we recall that $\mathcal{P}$ is the maximal past set where we can define a solution to the characteristic IVP of section \ref{section:ivp}, with initial data prescribed on $[0, U] \times \{v_0 \} \cup \{0\} \times [v_0, +\infty)$. In the following, we consider $v_1 \ge v_0$ to be sufficiently large with respect to the initial data and to $\eta$, and define the set
\begin{equation} \label{def_tildeP}
\mathcal{P}_{\delta} \coloneqq \mathcal{P} \cap \{(u, v) \in [0, U] \times [v_1, +\infty) \colon\,  0 < r_+ - r(u, v) \le \delta \}.
\end{equation}
The latter is non-empty, since the area-radius function provides an increasing parametrization on $\mathcal{H}^+$, and  $r(0, v) \to r_+$ as $v \to +\infty$ (see \eqref{assumption:r} and \eqref{assumption:lambda}).  

We now \textbf{define $\bm{E}$} as the set of points $q$ in $\mathcal{P}_{\delta}$
such that the following four inequalities hold for every $(u, v)$ in $J^{-}(q) \cap  \mathcal{P}_{\delta}$:
\begin{align}
|\phi|(u, v) + |\partial_v \phi|(u, v) &\le 4C_{\mathcal{H}} e^{-c(s) v}, \label{phi_bootstrap_constantr}\\
|D_u \phi|(u, v) &\le M|\nu|(u, v) e^{-c(s) v}, \label{Duphi_bootstrap_constantr} \\
\kappa(u, v) &\ge \frac12, \label{kappa_bootstrap_constantr} \\
\frac18 \Omega^2(0, v) \le |\nu|(u, v) &\le  \Omega^2(0, v), \label{bootstrap_nu_redshift}
\end{align}
where $M > 0$ is a constant depending uniquely on the initial data.

We notice that, since $E$ is a past set by construction,  inequalities \eqref{phi_bootstrap_constantr}--\eqref{bootstrap_nu_redshift} can be integrated along causal curves ending in $(u, v)$ and starting from the event horizon or from the null segment $[0, U] \times \{v_1\}$.  

\textbf{Outline of the proof}: in the following, we  show that estimates \eqref{redshift_uomega}--\eqref{redshift_lambdasign} are valid in $E$ and, at the same time, use this result to \emph{close the bootstrap} and show that $E = \mathcal{P}_{\delta}$. The conclusion then follows after proving that $\varnothing \ne J^{-}(\Gamma_R) \cap \{v \ge v_1\} \subset \mathcal{P}_{\delta}$. 

\textbf{Closing the bootstrap}: 
let us now fix $(u, v) \in E$. 
First, we stress that we have the following bounds on the area-radius function:
\begin{equation} \label{bounds_r_constantr}
\begin{cases} 
0 < r(0, v)-r(u, v) \le r_+ -r(u, v) \le \delta, \\
r(u, v) > \frac{r_+}{2}, \\
r(u, v) < r_+,
\end{cases}
\end{equation}
due to \eqref{def_tildeP}, \hyperref[section:assumptions]{assumption} \textbf{(G)} and by taking $\delta < \frac{r_+}{2}$.

We notice that, by \eqref{bounds_r_constantr} and \eqref{bootstrap_nu_redshift}:
\begin{equation} \label{deltauomega}
\delta \ge r(0, v) - r(u, v)  = \int_0^u |\nu|(u', v) du' \ge C u \Omega^2(0, v).
\end{equation}
In particular:
\begin{equation}  \label{uOmegaDelta}
u\Omega^2(0, v) \le C \delta.
\end{equation}
We stress that in \cite{VdM1} this relation was taken as the definition of the redshift region and, on the other hand, the bounds for $r$ were derived. A similar computation and \eqref{bound_r_eventhorizon} yield
\begin{equation} \label{bound_r_withOmega}
r_+ - r(u, v) =r_+ -r(0, v) +r(0, v) -r(u, v) \le C e^{-2sv} + u \Omega^2(0, v).
\end{equation}

Now, we use \eqref{maxwelleqnu} and bootstrap inequalities \eqref{phi_bootstrap_constantr} and \eqref{Duphi_bootstrap_constantr} to get
\[
|\partial_u Q|(u, v) \le  C M C_{\mathcal{H}}  r^2(u, v) |\nu|(u, v)e^{-2c(s)v}.
\]
After integrating in $u$ and using \eqref{bounds_r_constantr}:
\begin{equation} \label{Q_constantr}
|Q(u, v) - Q(0, v)|  \le C M C_{\mathcal{H}} \delta e^{-2 c(s)v}.
\end{equation}
Then, by proposition \ref{prop: boundsEH} and by  \eqref{def_cs}:
\begin{equation} \label{Q_constantr2}
|Q(u, v) - Q_+| \le |Q(u, v) - Q(0, v)| + |Q(0, v) - Q_+| \le C(N_{\text{i.d.}}, M) e^{-2c(s) v}.
\end{equation}

Furthermore, by integrating \eqref{rwaveeqn} and using \eqref{phi_bootstrap_constantr}, \eqref{bounds_r_constantr}  and the above bound on $Q$, we obtain
\[
|r \lambda|(u, v) \le r_+ |\lambda|(0, v) + \br{C(N_{\text{i.d.}}, \eta) + O\br{e^{-2 c(s) v}} }\int_0^u |\nu|(u', v)du'.
\]
The decay of $\lambda_{|\mathcal{H}^+}$ proved in proposition \ref{prop: boundsEH} and \eqref{bootstrap_nu_redshift}, yield
\begin{equation} \label{boundlambda_redshift}
|\lambda|(u, v) \le C e^{-2sv} + u\Omega^2(0, v),
\end{equation}
for a positive constant $C=C(N_{\text{i.d.}}, \eta)$. 

We can apply the latter bound, together with bootstrap inequalities \eqref{phi_bootstrap_constantr}--\eqref{kappa_bootstrap_constantr} and with \eqref{def_kappa}, to estimate $\varpi$. Indeed, \eqref{duvarpi2} gives
\[
|\partial_u \varpi|(u, v) \le C|\nu|(u, v)e^{-2c(s) v}.
\]
Thus, \eqref{bounds_r_constantr} and the bounds along the event horizon let us write
\begin{equation} \label{varpi_constantr}
|\varpi(u, v) - \varpi_+| \le |\varpi(u, v) - \varpi(0, v)| + |\varpi(0, v) - \varpi_+| \le C e^{-2c(s) v}.
\end{equation}

The definition of $K$ (see \eqref{def_surfacegrav}), together with bounds \eqref{bounds_r_constantr}, \eqref{bound_r_withOmega}, \eqref{Q_constantr2} and \eqref{varpi_constantr}, can be used to obtain
\begin{align} 
|K(u, v) - K_+| 
&\le \frac{1}{r_+^2} \left | \varpi(u, v) - \varpi_+ \right| + \frac{1}{r_+^3} \left | Q^2(u, v) - Q_+^2 \right| + \frac{|\Lambda|}{3}\left | r(u, v) - r_+ \right| \nonumber  \\
&+ |\varpi|(u, v) \left|\frac{1}{r^2(u, v)} - \frac{1}{r_+^2} \right| + Q^2(u, v) \left | \frac{1}{r^3(u, v)} - \frac{1}{r_+^3} \right| \nonumber \\
&\le C\br{ e^{-2sv}+u\Omega^2(0, v)} + O( e^{-2c(s)v}) \le C\br{ e^{-2c(s)v}+u\Omega^2(0, v)}, \label{preliminarybound_K}
\end{align}
for $v_1$ sufficiently large. 

In order to close the bootstrap for $\kappa$, we use  equation \eqref{dukappa}, the bootstrap inequalities \eqref{Duphi_bootstrap_constantr}, \eqref{bootstrap_nu_redshift} and recall that $\zeta = r D_u \phi$ to get
\[
|\partial_u \log \kappa|(u, v) = \frac{r |D_u \phi|^2}{|\nu|}(u, v) \le   C M^2 |\nu|(u, v) e^{-2c(s)v} \le C M^2 \Omega^2(0, v) e^{-2c(s)v}.
\]
So, after integrating from $\mathcal{H}^+$, we use \eqref{assumption:kappa} and \eqref{bounds_r_constantr} to get:
\begin{equation} \label{kappa_constantr}
\exp(-C M^2 u \Omega^2(0, v) e^{-2 c(s) v}) \le \kappa(u, v) \le  1.
\end{equation}
Since $u\Omega^2(0,v)$ is bounded (see \eqref{uOmegaDelta}), the above closes the bootstrap if we choose $v_1$ large enough. 

Notice that \eqref{dulambda_kappa}, the positivity of $K$ (see \eqref{preliminarybound_K}) and the smallness of $|\phi|$  imply that 
\begin{equation} \label{monotonicity_lambda}
\partial_v \nu = \partial_u \lambda  < 0 
\end{equation}
holds in $E$.

We can now improve estimate \eqref{Duphi_bootstrap_constantr} by following the procedure in \cite{VdM1} (see also the Gronwall inequalities exploited in \cite[section 5]{CGNS4}). 
 Let us consider a constant $a > 0$, which will be chosen in the next lines, and define
\begin{equation} \label{def_g}
f(u, v) \coloneqq  r(u, v)\frac{D_u \phi}{\nu}(u, v).
\end{equation}
Then, \eqref{dulambda_kappa} and \eqref{waveeqn_phi3_r}  imply
\begin{equation} \label{preliminary_Duphi}
\partial_v \br{e^{av} f(u, v)} 
= e^{av}\br{-\partial_v \phi + m^2 r \kappa \phi + \frac{i q Q \kappa}{r} \phi} + e^{a v} f(u, v) (a - \kappa(2K - m^2 r |\phi|^2)).
\end{equation}
Before continuing, we fix $a$ in such a way that the following two conditions are satisfied:
we require
\[
a - \kappa\br{2K - m^2 r |\phi|^2}(u, v') < 0, \quad \forall\, v' \ge v_1
\]
and
\[
a-c(s) > \frac{\eta}{2} > 0.
\]
Notice that such a choice of $a$ is admissible due to \eqref{kappa_constantr}, \eqref{def_cs}, \eqref{preliminarybound_K} and \eqref{phi_bootstrap_constantr}, provided that $\delta$ is chosen sufficiently small.

This choice of $a$ is then exploited when we integrate \eqref{preliminary_Duphi} and use  \eqref{phi_bootstrap_constantr}, \eqref{bounds_r_constantr} and \eqref{Q_constantr2}:
\begin{align*}
 e^{a v}|f|(u, v) &\le e^{a v_0}|f|(u, v_0)e^{\int_{v_0}^v \br{a - \kappa(2K - m^2 r |\phi|^2)}(u, v') dv'} \\
 &\quad + \int_{v_0}^v e^{a v'} \left | -\partial_v \phi + m^2 r \kappa \phi + i \frac{q Q \kappa}{r} \phi \right| (u, v') e^{\int_{v'}^v \br{a- \kappa \br{2K - m^2 r |\phi|^2}}(u, s)ds} dv'
 \\
&\le C(N_{\text{i.d.}}, v_1)\br{e^{a v_0} |f|(u, v_0) + 1} + C(N_{\text{i.d.}}, \eta)\int_{v_1}^{v}e^{(a-c(s))v'}dv',
\end{align*}
where we used
\[
\int_{v_0}^{v_1} e^{av'} \left | -\partial_v \phi + m^2 r \kappa \phi + i \frac{q Q \kappa}{r}\phi \right|(u, v')e^{\int_{v'}^{v_1}\br{a - \kappa(2K - m^2 r |\phi|^2)}(u, s) ds}dv' \le C(N_{\text{i.d.}}, v_1),
\]
with $C(N_{\text{i.d.}}, v_1)$ depending on the initial data, due to theorem \ref{thm:localexistence}, and on $v_1$. Notice that, so far, we increased the size of $v_1$ based on the $L^{\infty}$ norm of the initial data and on $\eta$. Thus, we have $C(N_{\text{i.d.}}, v_1) = C(N_{\text{i.d.}}, \eta)$.

Finally, we obtain:
\[
\left |\frac{D_u \phi}{\nu} \right | (u, v) \le C(N_{\text{i.d.}}, \eta) \br{e^{-av} +  e^{-c(s)v}},
\]
where we emphasize that the above constant does not depend on $M$. 
The bootstrap inequality \eqref{Duphi_bootstrap_constantr} then closes after choosing $M > 0$ sufficiently large with respect to the initial data. From now on, the constant $M$ will be absorbed in $C(N_{\text{i.d.}}, \eta)$.

Now, if we integrate \eqref{waveeqn_log} and use \eqref{phi_bootstrap_constantr}, \eqref{Duphi_bootstrap_constantr}, \eqref{Q_constantr2}, \eqref{boundlambda_redshift}, the fact that $\Omega^2 = 4|\nu| \kappa$, the boundedness of $r$ and $\kappa$ and, finally, \eqref{bootstrap_nu_redshift}:
\begin{equation} \label{partialvOmega}
\left| \partial_v \log \frac{\Omega^2(u, v)}{\Omega^2(0, v)} \right| \le C \int_0^u |\nu|(u', v) du' \le C u \Omega^2(0, v).
\end{equation}

Now, notice that, by \hyperref[section:assumptions]{assumption} \textbf{(A)}:
\[
\frac{\Omega^2(u, v_0)}{\Omega^2(0, v_0)} = \kappa(u, v_0) = 1+o(1),
\]
as $u$ goes to zero.\footnote{In the following, we assume that $v_1$ is sufficiently large. By \eqref{deltauomega} and \eqref{bound_nu_eventhorizon}, this implies that $u$ is small.} So, if we integrate \eqref{partialvOmega} along the $v$ direction and use \eqref{bound_nu_eventhorizon} and \eqref{bootstrap_nu_redshift}, we obtain:
\[
\left | \log \frac{\Omega^2(u, v)}{\Omega^2(0, v)} + o(1) \right| \le C u \int_{v_0}^v \Omega^2(0, v') dv' \le C u e^{2K_+ v} \le C u \Omega^2(0, v),
\]
for some $C=C(N_{\text{i.d.}}, \eta)$ that may differ from term to term. Therefore, using \eqref{uOmegaDelta}:
\begin{equation} \label{OmegaRatio}
 e^{-C \delta + o(1)} \le e^{-C u \Omega^2(0, v) + o(1)} \le \frac{\Omega^2(u, v)}{\Omega^2(0, v)} \le e^{C u \Omega^2(0, v) + o(1)} \le e^{C \delta + o(1)},
\end{equation}
where the $o(1)$ notation refers to the limit $v_1 \to +\infty$. Since $\Omega^2 = 4 |\nu| \kappa$, the above and \eqref{kappa_bootstrap_constantr} are enough to close the bootstrap inequality \eqref{bootstrap_nu_redshift}, provided that $\delta$ is chosen small and $v_1$ is sufficiently large with respect to the initial data and $\eta$.

Equations \eqref{dukappa}, \eqref{waveeqn_log_new} and the fact that $\zeta = r D_u \phi$ give
\begin{align}
\partial_u\br{ \partial_v \log \Omega^2 - 2K} &= (\kappa - 1) \partial_u (2K) - 2 \Re \br{D_u \phi \overline{\partial_v \phi}} - \frac{2\kappa}{r^2} \br{\partial_u \varpi - \frac{2 Q \partial_u Q}{r}}, \label{waveeqn_log_div} \\
&= \partial_u \left[2K (\kappa - 1) \right ] - 2K  \frac{  \kappa r |D_u \phi|^2}{\nu} - 2 \Re \br{D_u \phi \overline{\partial_v \phi}} - \frac{2\kappa}{r^2} \br{\partial_u \varpi - \frac{2 Q \partial_u Q}{r}}. \nonumber
\end{align}
By integrating the above from the event horizon and using \eqref{assumption:kappa}, \eqref{bound_dvlog_eventhorizon}, \eqref{phi_bootstrap_constantr}, \eqref{Duphi_bootstrap_constantr}, \eqref{bounds_r_constantr}, \eqref{preliminarybound_K}, \eqref{kappa_constantr}, the previous bounds on $Q$, $\partial_u Q$, $\varpi$ and $\partial_u \varpi$:
\begin{align*}
|\partial_v \log \Omega^2(u, v) - 2K(u, v)| &\lesssim e^{-2sv} + 2K(u, v)(1-\kappa(u, v)) + e^{-2c(s)v}\int_0^u |\nu|(u', v)du' \\
&\lesssim e^{-2sv} + 1-e^{-C e^{-2c(s)v}} + \delta e^{-2c(s)v}.
\end{align*}
Now, notice that, given a fixed $C>0$, the function $h(x) \coloneqq 1-e^{-Ce^{-x}}-Ce^{-x}$ is non-positive. Indeed, $h(0)<0$, $h'(x)>0$ for every $x$ in $\bbR$ and $\lim_{x \to +\infty}h(x) = 0$. Thus, recalling \eqref{def_cs}:
\begin{equation} \label{dvlog_constantr}
|\partial_v \log \Omega^2(u, v) - 2K(u, v)| \le C e^{-2 c(s) v}.
\end{equation}

In the next steps, we will bound $\phi$ and $\partial_v \phi$. 
By integrating $\partial_u \phi = D_u \phi - iq A_u \phi$ and using \eqref{final_Pricelaw}, \eqref{Duphi_bootstrap_constantr},   \eqref{bounds_r_constantr}, we have:
\begin{align}
|\phi|(u, v) & = \left | \phi(0 ,v) e^{-iq \int_0^u A_u(u', v)du'} + \int_0^u (D_u \phi)(u', v) e^{-iq \int_{u'}^u A_u(s, v) ds} du' \right| \label{integrate_phi} \\
& \le C_{\mathcal{H}}e^{-s v} + C(N_{\text{i.d.}}, \eta) e^{-c(s) v} \int_0^u |\nu|(u', v) du' \le  \br{C_{\mathcal{H}} + C(N_{\text{i.d.}}, \eta)\delta} e^{-c(s)v}. \nonumber
\end{align}
By taking $\delta$ small:
\begin{equation} \label{final_phi_redshift}
|\phi|(u, v) < \frac32 C_{\mathcal{H}}  e^{-c(s)v}.
\end{equation}

Now, the wave equation \eqref{waveeqn_phi2_r} for $\phi$ can be integrated to give
\begin{align}
(r \partial_v \phi) (u, v) &=  \br{r \partial_v \phi}(0, v) e^{-iq \int_{0}^u A_u(u', v)du'} +  \label{integrate_dvphi}\\
&+ \int_0^u e^{-iq \int_{u'}^u A_u(s, v)ds} \br{-\lambda D_u \phi + \nu \kappa  m^2 r \phi - i \frac{qQ\nu \kappa \phi}{r}}(u', v)du'. \nonumber
\end{align}
Then, expressions \eqref{final_Pricelaw}, \eqref{phi_bootstrap_constantr}, \eqref{Duphi_bootstrap_constantr}, \eqref{bounds_r_constantr},  \eqref{Q_constantr2} and \eqref{boundlambda_redshift} yield
\[
|r \partial_v \phi|(u, v) \le r_+ C_{\mathcal{H}} e^{-sv} +    C(N_{\text{i.d.}}, \eta) e^{-c(s)v}  \int_0^u |\nu|(u', v)du' \le r_+ C_{\mathcal{H}} e^{-sv} +   C(N_{\text{i.d.}}, \eta) \delta e^{-c(s)v},
\]
where $C(N_{\text{i.d.}}, \eta)$ may possibly differ from term to term and the quantity $O(e^{-2c(s)v})$ stemming from the estimate in $Q$ was reabsorbed in the costants, assuming that $v_1$ is sufficiently large.
The fact that $r(u, v) \ge \frac{r_+}{2}$ (see \eqref{bounds_r_constantr}), together with a suitably small choice of $\delta$, implies
\[
|\partial_v \phi|(u, v) < \frac52 C_{\mathcal{H}}e^{-c(s)v}.
\]
Combined with \eqref{final_phi_redshift}, this result closes the bootstrap inequality \eqref{phi_bootstrap_constantr}.
Due to the connectedness of $\mathcal{P}_{\delta}$, the previous steps reveal the set equality $E = \mathcal{P}_{\delta}$. 

\textbf{Localization of $\bm{J^{-}(\Gamma_R)}$}:
 To conclude the proof, it is enough to show that $J^{-}(\Gamma_R) \cap \{v \ge v_1\}\ne \varnothing$ and that $J^{-}(\Gamma_R) \cap \{v \ge v_1\} \subset \mathcal{P}_{\delta}$. 
First, we choose $v_1$ sufficiently large so that $u_R(v_1) \le U$: this choice of $v_1$ is possible by the exponential decay of the $u$ coordinate in this region (see \eqref{uOmegaDelta} and recall that $\Omega^2(0, v) \sim e^{2K_+ v}$) and entails that $\Gamma_R = \{r = R\} \ne \varnothing$, due to the extension criterion (see remark  \ref{rmk:extension_principle}). In particular, $J^{-}(\Gamma_R) \cap \{v \ge v_1\} \ne \varnothing$.

 Moreover,  $\Gamma_R$ is a spacelike curve for $v_1$ large (see corollary \ref{corollary_trappedregion}), and thus $u \le u_R(v) \le U$ for every $(u, v)$ in $J^{-}(\Gamma_R) \cap \{v \ge v_1 \}$.
 So, using again that $\nu < 0$ in $\mathcal{P}$, we have, for every $(u, v)$ in $ J^{-}(\Gamma_R) \cap \{v \ge v_1\}$: 
\[
0 < r_+ - r(u, v) \le r_+ - r(u_R(v), v)= r_+ - R < \delta,
\]
and thus $J^{-}(\Gamma_R) \cap \{v \ge v_1\} \subset \mathcal{P}_{\delta}$. 
\end{proof}

\begin{remark} \label{rmk:diffparam}
By corollary \ref{corollary_trappedregion} and by \eqref{monotonicity_lambda}, there exists a $C^1$ function $u \mapsto v_R(u)$ defined for $0 \le u \le u_R(v_1)$ such that (see also lemma \ref{lemma:properties_constantr}):
\[
\Gamma_R \cap \{v \ge v_1\} = \{(u_R(v), v) \colon \, v \ge v_1\} = \{(u, v_R(u)) \colon \, 0 \le u \le u_R(v_1)\}.
\]
We considered the case $v \ge v_1$ because  \eqref{monotonicity_lambda} holds for large values of the $v$ coordinate. However, notice that we might have 
\[
\Gamma_R \cap \{v \ge v_1\} \subsetneq \Gamma_R,
\]
see definition \eqref{def_constantR}. Since we are going to study solutions to the characteristic IVP restricted to large values of the $v$ coordinate, we will make no distinction between $\Gamma_{\varrho}$ and $\Gamma_{\varrho} \cap \{v \ge v_1\}$, for $\varrho \in (0, r_+)$, whenever the restriction to large values of $v$ is clear from the context.
 
Furthermore, by differentiating the relation $r(u, v_R(u)) = R$, we can see that
\begin{equation} \label{dv_R}
v_R'(u) = -\frac{\nu(u, v_R(u))}{\lambda(u, v_R(u))} < 0.
\end{equation}
\end{remark}

\subsection{Apparent horizon}
Assume that the constants $\eta$, $R$ and $\delta$ of proposition \ref{prop: redshift} are fixed. 
This allows us to give a parametrization of the apparent horizon 
\[
\mathcal{A} = \{(u, v) \in [0, U] \times [v_0, +\infty) \colon\, \lambda(u, v) = 0\}
\]
and to show that, for large values of the $v$ coordinate, $\mathcal{A}$ is a $C^1$ curve that lies in the causal past of $\Gamma_R$. 
We report  here the main results on $\mathcal{A}$, that follow from methods analogous to those used in \cite[section 4]{CGNS4}.

\begin{lemma}[on the sign of $\lambda$ along null cones] \label{lemma:constantu}
Let $\mathcal{P}$ be the maximal past set where a solution to our characteristic IVP is defined. Take $(u, v) \in \mathcal{P}$. If $\lambda(u, v) = 0$ (respectively, $\lambda(u, v) < 0$), then $\lambda(u, v + V) \le 0$ (respectively, $\lambda(u, v + V) < 0$) for every $V > 0$ for which $(u, v+V) \in \mathcal{P}$.
\end{lemma}

\begin{remark}
At the end of section \ref{section:stability} we will prove that $[0, U] \times [v_1, +\infty] \subset \mathcal{P}$ for a sufficiently small $U$ and a sufficiently large $v_1$. Therefore, a posteriori, the statement in lemma \ref{lemma:constantu} holds for every $V > 0$.
\end{remark}

\begin{corollary}[characterization of the apparent horizon] \label{corollary_apphorizon}
Under the assumptions of proposition \ref{prop: redshift}, we have that:
\begin{description}
\item[$\bm{\mathrm{P}\mathcal{A} 1.}$] For every $\tilde u \in [0, U]$, the set $\Gamma_R \cap \mathcal{A} \cap \{u= \tilde u\}$, if non-empty, consists only of a single point or only of a null segment of finite length. 
\end{description}
Moreover, for $v \ge v_1$, the set $\mathcal{A}$ is a non-empty $C^1$ curve and the following properties are satisfied:
\begin{description}
\item[$\bm{\mathrm{P}\mathcal{A} 2.}$] The inclusion
\begin{equation} \label{apphorizon_position}
\mathcal{A} \cap \{v \ge v_1 \} \subset I^{-}(\Gamma_R)
\end{equation}
holds;
\item[$\bm{\mathrm{P}\mathcal{A} 3.}$] The apparent horizon can be parametrized as 
\[
\mathcal{A} \cap \{v \ge v_1\} = \{(\uapp(v), v) \, \colon \, v \ge v_1  \}
\]
for some $C^1$ function $\uapp$ defined on $[v_1, +\infty)$ such that $\lambda(\uapp(v), v) = 0$  on such a domain;
\item[$\bm{\mathrm{P}\mathcal{A} 4.}$] We have
\begin{equation} \label{du_app}
\uapp'(v) = -\frac{\partial_v \lambda (\uapp(v), v)}{\partial_u \lambda(\uapp(v), v)} \le 0, \quad \forall\, v \ge v_1.
\end{equation}
\end{description}
\end{corollary}

\subsection{No-shift region} \label{sec:NR}
In regions corresponding to smaller values of the area-radius, the redshift effect is counterbalanced by a blueshift effect which is expected to originate in the Cauchy horizon (whose existence and stability will be proved as a consequence of the results of section \ref{section:stability}). A by-product of this phenomenon is the propagation, throughout this new region, of most of the estimates already obtained in proposition \ref{prop: redshift}. We recall that the constants $\eta$, $\delta$ and $R$ can be fixed according to the statement of proposition \ref{prop: redshift}, and that  $r_{-} \in (0, r_+)$ is one of the roots of the function $(1-\mu)(\cdot, \varpi_+, Q_+)$. A more precise formulation of the above heuristics is then the following: let us  choose\footnote{ The letter $Y$, which denotes the constant radius in the region where redshift and blueshift effects balance each other, stands for the yellow color, which lies between the red and blue colors in the visible electromagnetic spectrum.} $Y \in (r_{-}, R)$, close enough to $r_{-}$, and work in the set
\begin{equation} \label{noshift_region}
J^{+}(\Gamma_R) \cap J^{-}(\Gamma_Y),
\end{equation}
which we denote by the name \textbf{no-shift region}, 
after proving that it is non-empty. The following properties hold in this set:
\begin{itemize}
\item the functions $r_+-r(\cdot, \cdot)$ and $r(\cdot, \cdot) - r_{-}$ are bounded, but not necessarily $\delta$--small,
\item the surface gravity $K$ of the dynamical black hole changes sign: it remains positive near $\Gamma_R$ and becomes negative near $\Gamma_Y$. In this context, we introduce the quantity
\begin{equation} \label{def_surfacegrav_minus}
K_{-} \coloneqq \frac12 \abs{\partial_r (1-\mu)(r_{-}, \varpi_+, Q_+)},
\end{equation}
namely the absolute value of the surface gravity along the Cauchy horizon\footnote{We adopt the convention of \cite{CGNS4, LukOh1}, where $K_{-}$ is a positive quantity. In \cite{VdM1}, this quantity is defined with the opposite sign. } of the final sub-extremal black hole. 
\end{itemize}
In particular, the proof of proposition \ref{prop: redshift}, in which the smallness of $r_+-r$ is used to close the bootstrap inequalities, cannot be directly applied to this region.
 Nonetheless, we can subdivide the region \eqref{noshift_region} in finitely many narrow subsets on which we have a finer control on the area-radius function. The validity of the respective estimates is then proved in each subset by induction, where the basis step follows from the estimates on the redshift region, while the inductive step is proved with the aid of bootstrap (this technique was used in \cite{VdM1}, see also \cite{Dafermos_2005_uniqueness}).  A key role will be played by the bound 
on the geometric quantity $1-\mu$, which lets us relate the quantity $\nu$ defined in this region to $\nu_{|\Gamma_R}$. We stress that, in the uncharged case \cite{CGNS4}, the bound on $1-\mu$ could be obtained directly from the monotonicity of the renormalized Hawking mass. Such a property is not present in the current case.


\begin{proposition}[propagation of estimates due to coexisting redshift and blueshift effects] \label{prop: noshift}
Let $\Delta > 0$ be small, compared to the initial data, and let $Y>0$ be such that 
\[
0 < Y - r_{-} < \Delta.
\]
Then, given $v_1 \ge v_0$ large, we have $ J^+(\Gamma_R) \cap J^{-}(\Gamma_Y) \cap \{v \ge v_1 \} \ne \varnothing $ and, for every $(u, v)$ in $J^{+}(\Gamma_R) \cap J^{-}(\Gamma_Y) \cap \{v \ge v_1 \}$,  the following relations hold:
\begin{align}
 - C_{\mathcal{N}} \le \lambda(u, v) &\le -c_{\mathcal{N}} <0, \label{noshift_lambda_final}\\
 |\nu|(u, v) &\sim |\nu|(u, v_R(u))  \label{noshift_nu_final}\\
|\varpi(u, v) - \varpi_+| &\le C_{\mathcal{N}} e^{-2c(s) v}, \label{noshift_varpi_final} \\
|Q(u, v) - Q_+| &\le C_{\mathcal{N}} e^{-2c(s)v}, \label{noshift_Q_final} \\
|K(u, v) - K_+| + |K(u, v) + K_{-}| &\le C_{\mathcal{N}},\\
|\partial_v \log \Omega^2 (u, v) - 2K(u, v)| &\le C_{\mathcal{N}} e^{-2c(s) v}, \label{noshift_dvlog_final} \\
|D_u \phi|(u, v) &\le C_{\mathcal{N}} |\nu|(u, v) e^{-c(s)v}, \label{noshift_Duphi_final} \\
|\phi|(u, v) + |\partial_v \phi|(u, v) &\le C_{\mathcal{N}} e^{-c(s)v}, \label{noshift_phi_final}\\
 u &\sim e^{-2K_+ v}, \label{noshift_u_final}\\
 v-v_R(u) &\le C_{\mathcal{N}} \label{noshift_v_final},
\end{align}
for some positive constants $C_{\mathcal{N}}, c_{\mathcal{N}}$ depending on the initial data, and possibly on $\eta$, $R$ and $Y$. Here, $v_R$ is the function that gives a parametrization of $\Gamma_R$ as 
\[
\Gamma_R= \{r = R\} = \{(u, v_R(u)) \, \colon \, u \in [0, u_R(v_1)]\},
\]
and $c(s)$ was defined in \eqref{def_cs}. 
\end{proposition}
\begin{proof}
We recall that there exists a value $r_0 \in (r_{-}, r_+)$  such that\footnote{See also section 3 in \cite{CGNS2} for a more detailed analysis of the function $r \mapsto (1-\mu)(r, \varpi_+, Q_+)$.} $\partial_r (1-\mu)(r_0, \varpi_+, Q_+) = 0$.

\textbf{Definitions of the constants}: 
For every positive integer $l$, we define the following constants inductively:
\begin{equation} \label{def_Cl}
\begin{cases}
C_l \,\coloneqq 4 C_{l-1}, \\
C_0 \coloneqq 4C_{\mathcal{H}},
\end{cases}
\end{equation}
where $C_{\mathcal{H}}$ is the constant given by proposition \ref{prop: boundsEH}, and we consider\footnote{In particular, our choice will imply that $N \to +\infty$ if either $R \to r_+$ or $Y \to r_{-}$.}  $N \in \bbN$, which will be chosen sufficiently large during the proof, based on the initial data, on $R$ and on $Y$.
Moreover, let 
\begin{equation} \label{def_eps_noshift}
\epsilon = \epsilon(N) \coloneqq \frac{R-Y}{N} > 0.
\end{equation}
For $l=0, \ldots, N$, we define the sets
\[
\mathcal{N}_l \coloneqq J^{+}(\Gamma_R) \cap \{v \ge v_1\}  \cap \{(u, v) \in \mathcal{P} \, \colon \, R-l\epsilon \le r(u, v) \le R-(l-1)\epsilon \}
\]
and
\[
\mathcal{N} \coloneqq \bigcup_{l=0}^N \mathcal{N}_l.
\]
 We observe that, for every $l=0, \ldots, N$, we have $\mathcal{N}_l \ne \varnothing$, by the extension criterion (see remark  \ref{rmk:extension_principle}). It is straightforward to see that $\mathcal{N}_0 = \Gamma_R \cap \{v \ge v_1\}$.  
 In the following, we will repeatedly use that $0 < \kappa \le 1$ in $\mathcal{P}$ (see \eqref{firstorder_kappa}).
\begin{figure}[H]
\centering
\includegraphics[width=0.9\textwidth]{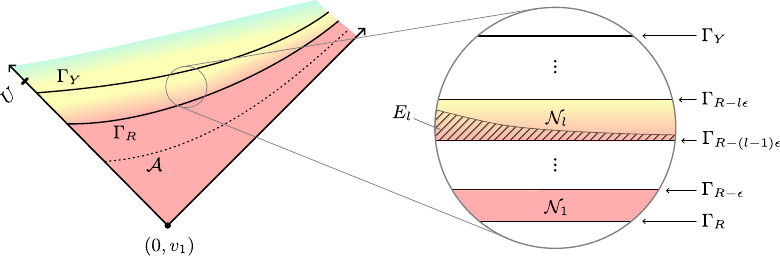}
\end{figure} 
\textbf{1. Induction}:
We want to prove, by induction, that for every $l \in \{0, \ldots, N\}$ the following holds:
\begin{align}
|\phi|( u,  v) + |\partial_v \phi|( u,  v) &\le C_l e^{-c(s)  v}, \label{phi_induction_noshift}\\
|D_u \phi|( u,  v) &\le M C_l|\nu|( u,  v)e^{-c(s)  v}, \label{Duphi_induction_noshift} \\
\kappa( u,  v) &\ge \frac12, \label{kappa_induction_noshift} \\
|\partial_v \log \Omega^2( u,  v) - 2K( u,  v)| &\le C e^{-2c(s)  v}, \label{dvlog_induction_noshift}
\end{align}
for every $(u,  v) \in \mathcal{N}_l$, for some constants $C=C(N_{\text{i.d.}}, \eta, R, Y)$ and $M = M(N_{\text{i.d.}}, \eta, R, Y)$. As we will see, this result implies that  inequalities \eqref{noshift_lambda_final}--\eqref{noshift_v_final} are satisfied. The above inequalities hold for $l=0$ by the results obtained in the redshift region (see, in particular, \eqref{redshift_dvlog_final}, \eqref{phi_bootstrap_constantr}--\eqref{kappa_bootstrap_constantr}) and by choosing $M$ large compared to the initial data. 

Let us now consider the inductive step: let $l \in \{1, \ldots, N\}$ be fixed and assume that, for every  $i=0, \ldots, l-1$, inequalities \eqref{phi_induction_noshift}--\eqref{dvlog_induction_noshift} are satisfied in $\mathcal{N}_i$.
To prove that these inequalities also hold for $i=l$ in $\mathcal{N}_l$, we adapt the bootstrap techniques of \cite{VdM1}.

 \textbf{2. Bootstrap}:  
  we define $E_l$ as the set of points $q$ in $\mathcal{N}_l$ such that the following inequalities hold for every $(u, v)$ in $J^{-}(q) \cap \mathcal{N}_l$:
\begin{align}
|\phi|(u, v) + |\partial_v \phi|(u, v) &\le  C_l e^{-c(s) v}, \label{phi_bootstrap_noshift}\\
|D_u \phi|(u, v) &\le M C_l |\nu|(u, v) e^{-c(s)v}, \label{Duphi_bootstrap_noshift} \\
\kappa(u, v) &\ge \frac12. \label{kappa_bootstrap_noshift}
\end{align}
The set $E_l$ is non-empty since \eqref{phi_bootstrap_noshift}--\eqref{kappa_bootstrap_noshift} hold in $\{r = R-(l-1)\epsilon\} \subset \partial \mathcal{N}_{l-1} \cap \partial \mathcal{N}_l$, by the induction hypothesis. Now, let $(u, v) \in E_l$ be fixed. 

We will repeatedly use the following estimates for $r$: 
\begin{equation} \label{bounds_r_noshift}
\begin{cases}
r_{-} < Y \le r \le R < r_+, &\text{in } \mathcal{N},\\
0 \le R-r(u, v) \le R-Y, \\
0 \le R-(l-1)\epsilon - r(u, v) \le \epsilon, \\
0 \le r(u, v) - (R- l \epsilon) \le \epsilon.
\end{cases}
\end{equation}

Now, equation \eqref{maxwelleqnu}, together with \eqref{phi_bootstrap_noshift}, \eqref{Duphi_bootstrap_noshift} and \eqref{bounds_r_noshift}, gives
\begin{equation} \label{noshift_duQ}
|\partial_u Q|(u, v) \le C M C_l^2 |\nu|(u, v) e^{-2 c(s)v}.
\end{equation}
If we integrate the latter from the curve $\Gamma_R$ and use the results of proposition \ref{prop: redshift} and \eqref{def_Cl}:
\begin{equation} \label{noshift_Q}
|Q(u, v) - Q_+| \le |Q(u, v)-Q(u_R(v), v)| + |Q(u_R(v), v) - Q_+| \le C(N_{\text{i.d.}}, \eta, M)  e^{-2 c(s) v}.
\end{equation}
Notice that the above constant depending on $M$ is coupled with the term $e^{-2c(s)}$, which decays faster than the leading-order term of the next inequalities that we take into account.

We can proceed similarly for $\varpi$ by using \eqref{def_kappa} and integrating equation \eqref{duvarpi2} from $\Gamma_R$. Indeed, we can use the fact that $|\lambda|$ is bounded from above by a positive constant (this can be proved by following the same steps that led to \eqref{boundlambda_redshift} in the redshift region), together with  \eqref{phi_bootstrap_noshift}--\eqref{kappa_bootstrap_noshift}, \eqref{bounds_r_noshift}, \eqref{noshift_Q} and the bounds obtained in the redshift region, to write
\begin{equation} \label{noshift_duvarpi}
|\partial_u \varpi|(u, v) \le C |\nu|(u, v) e^{-2c(s) v}
\end{equation}
and thus
\begin{equation} \label{noshift_varpi}
|\varpi(u, v) - \varpi_+| \le |\varpi(u, v) - \varpi(u_R(v), v)| + |\varpi(u_R(v), v) - \varpi_+| \le C e^{-2c(s)v}.
\end{equation}

The inequality 
\begin{equation} \label{noshift_mu}
-C(N_{\text{i.d.}}) <  (1-\mu)(u, v) <  - C(N_{\text{i.d.}}, R, Y) < 0,
\end{equation}
 follows from \eqref{noshift_Q}, \eqref{noshift_varpi}, \eqref{kappa_bootstrap_noshift} and lemma \ref{lemma:lambda_along_curve} (its proof does not require $r$ to be constant: it is sufficient that the area-radius is bounded away from $r_{-}$ and $r_+$  as in the current case). In particular, $(1-\mu)(u, v)$ approaches zero when either $R \to r_+$ or $Y \to r_{-}$.
Whence \eqref{kappa_bootstrap_noshift} entails
\begin{equation} \label{noshift_lambda}
-C(N_{\text{i.d.}}) < \lambda(u, v)=[\kappa(1-\mu)](u, v) < - \frac{C(N_{\text{i.d.}}, R, Y)}{2} < 0.
\end{equation}

We now show the validity of \eqref{noshift_u_final}. Due to \eqref{redshift_uomega}:
\begin{equation} \label{bounduR}
2(r(0, v_R(u)) - R) \Omega^{-2}(0, v_R(u)) \le u \le C \delta \Omega^{-2}(0, v_R(u)),
\end{equation}
where $v_R$ was introduced in remark \ref{rmk:diffparam}.
Furthermore, the difference between $v_R(u)$ and $v$ is bounded. In fact, if we integrate \eqref{noshift_lambda} from $v_R(u)$ to $v$:
\[
-C(N_{\text{i.d.}}) \br{v-v_R(u)} < r(u, v) - R < -\frac{C(N_{\text{i.d.}}, R, Y)}{2}\br{v - v_R(u)}.
\]
Hence, due to \eqref{bounds_r_noshift}:
\begin{equation} \label{bounddifferencev}
0 \le  v - v_R(u) \le 2 \frac{R-r(u, v)}{C(N_{\text{i.d.}}, R, Y)} \le 2 \frac{R-Y}{C(N_{\text{i.d.}}, R, Y)}.
\end{equation}
Since $\Omega^2(0, y) \sim e^{2K_+ y}$ for every $y$, if we plug \eqref{bounddifferencev} in \eqref{bounduR} then \eqref{noshift_u_final} is satisfied.

Now, using \eqref{def_kappa} and \eqref{algconstr}, the Raychaudhuri equation \eqref{raych_v} can be cast in the form 
\begin{equation} \label{raych_recast}
\partial_v \left ( \frac{\nu}{1-\mu} \right) = \frac{\nu}{1-\mu} \frac{r |\partial_v \phi|^2}{ \lambda},
\end{equation}
to get a bound on $\nu$. Indeed, by integrating the above from $v_R(u)$ to $v$, we have
\begin{equation} \label{nu_firstintegration}
|\nu|(u, v) = |\nu|(u, v_R(u)) \frac{|1-\mu|(u, v)}{|1-\mu|(u, v_R(u))}e^{\int_{v_R(u)}^v \frac{r}{\lambda}|\partial_v \phi|^2(u, v') dv' }.
\end{equation}
By  \eqref{phi_bootstrap_noshift}, \eqref{bounds_r_noshift} and \eqref{noshift_lambda}:
\begin{equation} \label{def_Cnu}
\left | \int_{v_R(u)}^v \frac{r}{\lambda} |\partial_v \phi|^2 (u, v')dv' \right| \le C(N_{\text{i.d.}}, R, Y) e^{-2c(s)v_R(u)}.
\end{equation}
Thus, \eqref{nu_firstintegration} and \eqref{noshift_mu} give
\begin{equation} \label{noshift_nu}
 |\nu|(u, v) \sim |\nu|(u, v_R(u)).
\end{equation}

We can now close bootstrap inequality \eqref{Duphi_bootstrap_noshift}. By inspecting \eqref{waveeqn_phi3_r} and using \eqref{phi_bootstrap_noshift}, \eqref{bounds_r_noshift}, \eqref{noshift_Q} and \eqref{noshift_nu}, we have
\[
|\partial_v(r D_u \phi)|(u, v) \le C(N_{\text{i.d.}})  C_l |\nu|(u, v)  e^{-c(s)v} \le  C(N_{\text{i.d.}}, R, Y) C_l  |\nu|(u, v_R(u)) e^{-c(s)v},
\]
for $v_1$ large. 
Notice that the above constant does not depend on $M$.
Before integrating the above inequality from $\Gamma_R$, we notice that \eqref{decay_redshift_Duphi} gives
\[
|r D_u \phi|(u, v_R(u)) \le C(N_{\text{i.d.}}, \eta) |\nu|(u, v_R(u)) e^{-c(s) v_R(u)} 
\]
Therefore, the relation $C_l = 4^{l+1}C_{\mathcal{H}}$
and estimates \eqref{bounddifferencev} and \eqref{noshift_nu} allow us to write
\begin{align*}
|rD_u \phi|(u, v) &\le \br{C(N_{\text{i.d.}}, \eta) + C(N_{\text{i.d.}}, R, Y)} C_l |\nu|(u, v_R(u))e^{c(s) \br{v - v_R(u)}} e^{-c(s) v} \\
&\le C(N_{\text{i.d.}}, \eta, R, Y)  C_l |\nu|(u, v) e^{-c(s)v}.
\end{align*}
Inequality \eqref{Duphi_bootstrap_noshift} is then closed due to \eqref{bounds_r_noshift} and by choosing $M$ sufficiently large with respect to the initial data, $\eta$, $R$ and $Y$ (this might possibly require a larger value for $v_1$, as well).

The bootstrap on $\kappa$ can then be closed by following the same steps that precede \eqref{kappa_constantr} in the proof of proposition \ref{prop: redshift}. 

Now,  $\Gamma_{R-(l-1)\epsilon}$ is non-empty (see  the end of the proof of proposition \ref{prop: redshift}). Lemma \ref{lemma:properties_constantr} then shows that we have the parametrization $\Gamma_{R-(l-1)\epsilon} = \{(u_{l-1}(v), v) \, \colon \, v \ge v_1\}$ for some function $u_{l-1}$.

In order to close the bootstrap for $\phi$ and $\partial_v \phi$, we first notice that, by the inductive assumption (i.e. \eqref{phi_induction_noshift} for the $(l-1)$-th step) and by \eqref{def_Cl}, we have:
\[
|\phi|(u_{l-1}(v), v) \le \frac{C_l}{4}e^{-c(s)v}.
\]
Then, by integrating $\partial_u \phi = D_u \phi - i q A_u \phi$ and using \eqref{Duphi_bootstrap_noshift} and \eqref{bounds_r_noshift}, we obtain
\[
|\phi|(u, v)  \le \frac{C_l}{4} e^{-c(s) v} + C_l C(N_{\text{i.d.}}, R, Y) e^{-c(s) v} \int_{u_{l-1}(v)}^u |\nu|(u', v) du' \le  \br{\frac14 + C(N_{\text{i.d.}}, R, Y) \epsilon} C_l e^{-c(s)v}, 
\]
By choosing $N$ sufficiently large (see \eqref{def_eps_noshift}),  the last constant multiplying $C_l e^{-c(s)v}$ is strictly smaller than $\frac{3}{8}$.

Similarly, we integrate \eqref{waveeqn_phi2_r} from $u_{l-1}(v)$ to $u$ (see also \eqref{integrate_dvphi}) and use \eqref{phi_bootstrap_noshift},  \eqref{Duphi_bootstrap_noshift},  \eqref{bounds_r_noshift}, \eqref{noshift_lambda} and the inductive assumption (see \eqref{phi_induction_noshift} for the $(l-1)$-th step) to obtain
\begin{align*}
|r \partial_v \phi|(u, v) &\le \frac{C_l}{4}\br{R - (l-1)\epsilon} e^{-c(s)v} +  C_{l} C(N_{\text{i.d.}}, R, Y) e^{-c(s)v}  \int_{u_{l-1}(v)}^u |\nu|(u', v)du' \\
&\le C_l e^{-c(s)v} \br{ \frac14 \br{R - (l-1)\epsilon} + C(N_{\text{i.d.}}, R, Y) \epsilon},
\end{align*}
assuming that $v_1$ is sufficiently large.
If we divide both sides of the above by $r(u, v) \ge R-l\epsilon$ (see \eqref{bounds_r_noshift}), the first term enclosed in brackets can be estimated using that $R-l\epsilon \ge Y$ and that,  taking $N$ large:
\[
\frac{R-l\epsilon + \epsilon}{R- l \epsilon} < 1+\frac{\epsilon}{Y} < 2.
\]
Thus, for $N$ large:
\[
|\partial_v \phi|(u, v) < \frac58 C_l e^{-c(s)v}.
\]
This result and the previous bound on $\phi$ close the bootstrap inequality \eqref{phi_bootstrap_noshift}.
Hence, we have $E_l = \mathcal{N}_l$. 

\textbf{3. Remaining estimates}: 
We are now left to prove \eqref{dvlog_induction_noshift} in $\mathcal{N}_l$. Similarly to the case of the redshift region, this follows after integrating \eqref{waveeqn_log_div} from $\Gamma_{R-(l-1)\epsilon}$, and using \eqref{dvlog_induction_noshift} (in $\mathcal{N}_{l-1}$),   \eqref{phi_bootstrap_noshift}, \eqref{Duphi_bootstrap_noshift},  \eqref{bounds_r_noshift}, \eqref{noshift_duQ}, \eqref{noshift_Q},  \eqref{noshift_duvarpi}, \eqref{noshift_varpi} and the fact that $K$ is bounded.

It then follows that  inequalities \eqref{noshift_lambda_final}--\eqref{noshift_v_final} are satisfied in $\mathcal{N}$. Moreover, we showed in the previous steps that $\Gamma_{R-(l-1)\epsilon} \ne \varnothing$, for $l \in \{1, \ldots, N\}$. The last step of the induction yields $\varnothing \ne \Gamma_Y \cap \{v \ge v_1\} \subset \mathcal{N}$, and we further know that $\Gamma_Y$ is a spacelike curve, for $v \ge v_1$, by corollary \ref{corollary_trappedregion}. The latter result, together with the fact that $\nu < 0$ in $\mathcal{P}$, entails that $\varnothing \ne J^{+}(\Gamma_R) \cap J^{-}(\Gamma_Y) \cap \{v \ge v_1 \} \subset \mathcal{N}$ (see last paragraph of the proof of proposition \ref{prop: redshift}).
\end{proof}

\subsection{Early blueshift region: general behaviour}  \label{sec:EBR1}
In the next steps, we inspect the consequences of the blueshift effect in the region $J^{+}(\Gamma_Y)$.
 Given a real number 
\begin{equation} \label{def_beta}
\beta > 0,
\end{equation}
sufficiently small compared to the initial data,  we define the curve
\begin{equation} \label{def_gamma}
\gamma \coloneqq \Big \{(u, (1+\beta)v_{Y}(u)) \colon \, u \in [0, u_Y(v_1)] \Big \}.
\end{equation}
For a suitably large choice of $v_1$, the curve $\gamma$ is spacelike, since that is the case for $\Gamma_Y = \{ (u, v_Y(u))\colon u \in [0, u_Y(v_1)] \}$. We also use the notation
\[
v_{\gamma}(u) \coloneqq (1+\beta)v_Y(u),
\]
and we denote by $u_{\gamma}(v)$ the $u$-coordinate of the unique point along the curve $\gamma$ corresponding to the outgoing null coordinate $v$.

The aim of the curve $\gamma$ is to probe a spacetime region where the blueshift effect starts to be relevant, but whose influence is weak enough so that the quantities describing the dynamical black hole can be controlled as done in the previously inspected regions.
This is why we denote the region 
\[
J^+(\Gamma_Y) \cap J^{-}(\gamma)
\]
by the name \textbf{early blueshift region}. 

Here the dynamical solution is still quantitatively close to the final sub-extremal Reissner-Nordstr{\"o}m-de Sitter black hole, in the sense that $\varpi$ converges to $\varpi_+$ and $Q$ converges to  $Q_+$ at an exponential rate. However, due to the non-zero value of $\beta$ (which distinguishes $\gamma$ from a curve of constant radius), we lose decay with respect to the no-shift region.  Also, although this region is delimited by a curve of constant area-radius only in the past, the function $r$ is still bounded away from zero for $v_1$ large.

Some technical difficulties, which were absent in the uncharged $\Lambda>0$ case and in the charged $\Lambda=0$ case, arise in this region. Differently from \cite{CGNS2} and \cite{CGNS4}, we cannot directly apply BV estimates  (originally presented in \cite{Dafermos_2005_uniqueness}) on the quantities $D_u \phi$ and $\partial_v \phi$, due to new terms in equations \eqref{dutheta} and \eqref{dvzeta} that were absent in the uncharged case. We then prove the results by bootstrap, similarly to the procedure followed in the redshift region, where now we close the bootstrap inequalities by exploiting  the smallness of $r - r_{-}$. 
The procedure differs from the bootstrap technique used in \cite{VdM1}, since the approximations of the $v$ coordinate used there give rise to exponential errors that cannot be removed\footnote{Notice that that the same problem occurred in the no-shift region of section \ref{sec:NR}. However, these error terms could be controlled due to \eqref{noshift_v_final}.} using the smallness of $r- r_{-}$, see also section \ref{subsec:review_work}.

 A prominent role in the proof is played by $\mathcal{C}_s$, a function in both the $u$ and $v$ coordinates that we define and that determines the exponential decay of the main quantities. This function is positive in the early blueshift region\footnote{It is not necessarily positive to the future of this region, but this will not be a problem, see already section \ref{section:lateblueshift}.} and replaces the linear function $v \mapsto c(s)v$ of  the previous regions. The aim of $\mathcal{C}_s$ is to interpolate between the exponential decay known along the curve $\Gamma_Y$ (see proposition \ref{prop: noshift}) and the one expected along $\gamma$ (see \cite{CGNS4}), so that the function $|\nu| \exp (-\mathcal{C}_s)$, appearing in many integral estimates, is monotonically increasing in $v$. In particular, $\mathcal{C}_s$ is chosen in such a way that it is monotonic in both coordinates, so that the main integral estimates of this region can be achieved similarly to the ones of the redshift and no-shift regions.

\begin{proposition}[estimates in the early blueshift region] \label{prop: blueshift}
Let $\beta > 0$.
Then, for every $\varepsilon > 0$ small enough, depending only on the initial data, there exists  $\Delta \in (0, \varepsilon)$ and $Y > 0$ such that
\begin{equation} \label{Y_constraint}
0 < Y - r_{-} < \Delta,
\end{equation}
 and such that the following inequalities hold for every $(u, v) \in J^+(\Gamma_Y) \cap J^{-}(\gamma) \cap \{v \ge v_1\}$:
\begin{align}
r(u, v) &\ge r_{-} - \varepsilon, \label{blueshift_r_final}\\
|\varpi(u, v) - \varpi_+| &\le C e^{-2(c(s) - \tau)v }, \\
 |Q(u, v) - Q_+| &\le C e^{-2 (c(s) - \tau)v }, \label{blueshift_final_Q}\\
|K(u, v) + K_{-}| &\le C \varepsilon, \label{blueshift_final_K} \\
|\partial_v \log \Omega^2(u, v) - 2K(u, v)| &\le C e^{-(c(s) - \tau)v}, \label{blueshift_final_dvlog}\\
|D_u \phi|(u, v) &\le C |\nu|(u, v) e^{-(c(s) - \tau)v }, \label{Duphi_blueshift_final}\\
 |\phi|(u, v) + |\partial_v \phi|(u, v) &\le C e^{-(c(s) - \tau)v }, \label{phi_blueshift_final}
\end{align}
for some $v_1 \ge v_0$ large, where $\tau > 0$, $\tau = O(\beta)$ and 
where $C > 0$ depends on the initial data and possibly on $\eta$, $R$ and $Y$. The curve $\gamma$ is defined in \eqref{def_gamma}. The positive constant $K_{-}$ is given by \eqref{def_surfacegrav_minus} and $c(s)$ is defined in \eqref{def_cs}. 
\end{proposition}
\begin{proof}
In the following, we will repeatedly use that $0 < \kappa \le 1$, that $\nu < 0$ in $\mathcal{P}$ and that $\lambda < 0$ in $J^{+}(\Gamma_Y)$ (this follows from \eqref{noshift_lambda_final} and lemma \ref{lemma:constantu}, together with the fact that $\Gamma_Y$ is a spacelike curve).
In the next lines the constant $C_{\mathcal{N}} = C_{\mathcal{N}}(N_{\text{i.d.}}, \eta, R, Y)$ (defined in the statement of proposition \ref{prop: noshift}) will be used. 

Given $C>0$ depending uniquely on the initial data, we define
\begin{equation} \label{def_curlyC}
\mathcal{C}_{s}(u, v) \coloneqq c(s)v_Y(u) - 2(2K_{-} + C \varepsilon)(v-v_Y(u)),
\end{equation}
which is positive in $J^+(\Gamma_Y) \cap J^{-}(\gamma)$.
Recall that, in the latter region, $v_Y(u) \le v \le (1+\beta)v_Y(u)$. So,  if $\beta$ is small:
\[
\mathcal{C}_s(u, v) = c(s) v_Y(u) - 2(2K_{-}+ C \varepsilon)(v-v_Y(u)) \ge [ c(s) - 2(2K_{-} + C \varepsilon)\beta ] v_Y(u) > 0,
\]
for every $(u, v) \in J^+(\Gamma_Y) \cap J^{-}(\gamma) \cap \{v \ge v_1\}$. In particular, $\mathcal{C}_s$ is a \textbf{positive} function in the early blueshift region. Using the restrictions on the $v$ coordinate again, the above also implies
\begin{equation} \label{translate_decay}
\mathcal{C}_s(u, v) > (c(s) - \tau)v,
\end{equation}
for $\tau > 0$, $\tau = O(\beta)$.
Furthermore, the fact that $v'_Y(u) < 0$ for every $u \in [0, u_R(v_1)]$ (see also \eqref{dv_R}) gives
\begin{equation} \label{monotonicity_curlyC}
\partial_u \mathcal{C}_s < 0 \quad \mathrm{in }\, J^+(\Gamma_Y) \cap J^{-}(\gamma) \cap \{v \ge v_1\}.
\end{equation}

\textbf{Setting up the bootstrap procedure}: let $E$ be the set of points $q$ in $\mathcal{D} \coloneqq J^+(\Gamma_Y) \cap J^{-}(\gamma) \cap \{v \ge v_1\}$ such that the following inequalities hold for every $(u, v)$ in $J^{-}(q) \cap \mathcal{D}$:
\begin{align}
r(u, v) &\ge r_{-} - 2\varepsilon, \label{r_bootstrap_blueshift}\\
|\phi|(u, v) + |\partial_v \phi|(u, v) &\le 4 C_{\mathcal{N}} e^{-\mathcal{C}_{s}(u, v)}, \label{phi_bootstrap_blueshift}\\
|D_u \phi|(u, v) &\le M C_{\mathcal{N}}  |\nu|(u, v) e^{-\mathcal{C}_{s}(u, v)}, \label{Duphi_bootstrap_blueshift} \\
\kappa(u, v) &\ge\ \frac12, \label{kappa_bootstrap_blueshift}
\end{align}
for some suitably large $M = M(N_{\text{i.d.}}) > 0$.

The validity of the bootstrap inequalities along $\Gamma_Y$ is promptly ensured for $v_1$ large, since $Y > r_{-}$, due to \eqref{noshift_phi_final}, \eqref{noshift_Duphi_final},  \eqref{kappa_induction_noshift} and to the fact that (see \eqref{def_curlyC}):
\[
\mathcal{C}_s(u, v_Y(u)) = c(s) v_Y(u), \quad \forall\, u \in [0, u_Y(v_1)].
\]

\textbf{Closing the bootstrap}: 
let $(u, v) \in E$.
First, we notice that, by construction of $E$:
\begin{equation} \label{bounds_r_blueshift}
0 < \frac{r_{-}}{2} < r_{-} - 2 \varepsilon \le r(u, v) \le Y < r_+,
\end{equation} 
provided that $\varepsilon$ is sufficiently small.

Similarly to the previous regions (and now using \eqref{monotonicity_curlyC}), the inequalities  \eqref{phi_bootstrap_blueshift}, \eqref{Duphi_bootstrap_blueshift}, \eqref{bounds_r_blueshift} applied to \eqref{maxwelleqnu} and the results of proposition \ref{prop: noshift} yield
\begin{equation} \label{blueshift_Q}
\abs{Q(u, v) - Q_+} \le \abs{Q(u, v) - Q(u_Y(v), v)} + \abs{Q(u_Y(v), v) -  Q_+} \le C(N_{\text{i.d.}}, \eta, R, Y) e^{-2 \mathcal{C}_{s}(u, v)},
\end{equation}
where, in particular, we used $e^{-c(s)v_Y(u)} \le e^{-\mathcal{C}_s(u, v)}$. 

Furthermore, integrating \eqref{rwaveeqn} from $\Gamma_Y$ and using \eqref{phi_bootstrap_blueshift}, \eqref{bounds_r_blueshift} and  \eqref{blueshift_Q}, together with the bound \eqref{noshift_lambda} obtained in the no-shift region:
\begin{equation} \label{boundlambda_blueshift}
|\lambda|(u, v) \le C=C(N_{\text{i.d.}}),
\end{equation}

In order to bound $\varpi$, we consider \eqref{duvarpi2} and apply \eqref{def_kappa}, the bootstrap inequalities \eqref{phi_bootstrap_blueshift}, \eqref{Duphi_bootstrap_blueshift},  \eqref{kappa_bootstrap_blueshift}, the previous bounds on $|\lambda|$ and $Q$ and finally integrate using \eqref{bounds_r_blueshift} and \eqref{monotonicity_curlyC} to get
\begin{equation} \label{blueshift_varpi}
\abs{\varpi(u, v) - \varpi_+} \le \abs{\varpi(u, v) - \varpi(u, v_Y(u))} + \abs{\varpi(u, v_Y(u)) - \varpi_+} \lesssim e^{-2 \mathcal{C}_{s}(u, v)}.
\end{equation}

Now, since $J^{+}(\Gamma_Y) \cap J^{-}(\gamma)$ is stricly contained in the trapped region $\mathcal{T}$ for $v$ large (see \eqref{apphorizon_position}),  relation \eqref{algconstr} entails that $(1-\mu)(u, v) < 0$. In particular, by \eqref{def_mu}:
\[
(1-\mu)(u, v) = 1 - \frac2r \varpi_+ + \frac{Q_+^2}{r^2} - \frac{\Lambda}{3}r^2 - \frac2r(\varpi(u, v)-\varpi_+) + \frac{Q(u, v)^2 - Q_+^2}{r^2} < 0.
\]
Due to \eqref{blueshift_Q} and \eqref{blueshift_varpi}, the latter implies that
\[
(1-\mu)(r(u, v), \varpi_+, Q_+) < \frac2r (\varpi(u, v) -\varpi_+) - \frac{Q(u, v)^2-Q_+^2}{r^2} \lesssim e^{-2 \mathcal{C}_{s}(u, v)}.
\]
Thus, by studying the plot of $r \mapsto (1-\mu)(r, \varpi_+, Q_+)$ (see e.g. \cite[section 3]{CGNS2}), we have 
\begin{equation} \label{bounds_r_improved}
r(u, v) > r_{-} - \varepsilon,
\end{equation}
provided that $v_1$ is sufficiently large. This closes the bootstrap inequality \eqref{r_bootstrap_blueshift} for $r$.

We now use \eqref{def_surfacegrav} and  \eqref{def_surfacegrav_minus} to write
\begin{align*}
|K(u, v) + K_{-}| 
&\le |\varpi|(u, v) \left | \frac{1}{r^2(u, v)} - \frac{1}{r_{-}^2} \right| + \frac{|\varpi(u, v) - \varpi_+|}{r_{-}^2} \\
&\quad + Q^2(u, v) \left | \frac{1}{r^3(u, v)}-\frac{1}{r_{-}^3} \right| + \frac{|Q^2(u, v) - Q_+^2|}{r_{-}^3} + \frac{|\Lambda|}{3}|r(u, v) - r_{-}|.
\end{align*}
The terms proportional to $|\varpi - \varpi_+|$ and $|Q - Q_+|$ give $O(e^{-2\mathcal{C}_s(u, v)})$,  due to the previously obtained estimates.
Moreover, we note that the inequality $r(u, v) \le Y$, together with \eqref{Y_constraint}, \eqref{bounds_r_improved} and the fact that $\Delta < \varepsilon$ gives:
\[
|r(u, v) - r_{-}| < \max\{\Delta, \varepsilon\}  =\varepsilon.
\]
On the other hand, from \eqref{bounds_r_blueshift} and the above we obtain:
\[
\left | \frac{1}{r^2(u, v)} - \frac{1}{r_{-}^2} \right|  < C(N_{\text{i.d.}}) \varepsilon,
\]
and similarly for the cubic case.
Therefore, we conclude that 
\begin{equation} \label{bound_2K_blueshift}
|2K(u, v) +2 K_{-}| \le C(N_{\text{i.d.}}) \varepsilon,
\end{equation}
for $v_1$ sufficiently large.
In particular, $K$ is a negative quantity in the early blueshift region.

The bootstrap condition on $\kappa$ can be closed as done in the redshift region (see e.g. \eqref{kappa_constantr}), now using \eqref{dukappa} (recall that $\zeta = r D_u \phi$), \eqref{Duphi_bootstrap_blueshift}, \eqref{bounds_r_blueshift} and the fact that $\kappa$ is close to 1 in the no-shift region. We emphasize that $\kappa$ can be taken close to 1 also in the current region, provided that $v_1$ is sufficiently large. Hence, for a suitable choice of $v_1$:
\begin{equation} \label{kappaK_lowerbound_blueshift}
\abs{\kappa(2K - m^2 r |\phi|^2)} > (1 - o(1))(2K_{-} - C\varepsilon) > 2K_{-} - 2 C \varepsilon > 0,
\end{equation}
due to the smallness of $\varepsilon$, where $C > 0$ depends uniquely on the initial data.

The estimate for $\partial_v \log \Omega^2$ is obtained analogously to the procedure of the previous regions, i.e. by integrating \eqref{waveeqn_log_div} from $\Gamma_Y$ and using \eqref{noshift_dvlog_final}, \eqref{monotonicity_curlyC}, \eqref{phi_bootstrap_blueshift}, \eqref{Duphi_bootstrap_blueshift}, \eqref{bounds_r_blueshift}, \eqref{bound_2K_blueshift},  the estimate for $\partial_v \log \Omega^2$ in the no-shift region and the previous bounds on $\partial_u Q$, $Q$, $\partial_u \varpi$ and $\varpi$.

To close the bootstrap for $\phi$ we integrate $\partial_u \phi = -iq A_u \phi + D_u \phi$ from the curve $\Gamma_Y$ (see also \eqref{integrate_phi} for an analogous expression). Thus, we  use \eqref{noshift_phi_final},  \eqref{Y_constraint},  \eqref{monotonicity_curlyC}, \eqref{Duphi_bootstrap_blueshift},  \eqref{bounds_r_improved} and the fact that $e^{-c(s)v_Y(u)} \le e^{-\mathcal{C}_s(u, v)}$ to get, for  a suitably large $v_1$:
\begin{align}
|\phi|(u, v) &\le C_{\mathcal{N}} e^{-c(s)v_Y(u)} + M C_{\mathcal{N}} e^{- \mathcal{C}_{s}(u, v)}\int_{u_Y(v)}^u |\nu|(u', v)du' \nonumber \\ 
&\le C_{\mathcal{N}}\br{1 + \br{Y- r_{-}+\varepsilon} M}   e^{-\mathcal{C}_{s}(u, v)} \nonumber
 \\
 &\le C_{\mathcal{N}} \br{ 1+ 2  \varepsilon M } e^{- \mathcal{C}_{s}(u, v)}. \label{closebootstrapphi}
\end{align}
The quantity $1+2\varepsilon M(N_{\text{i.d.}})$ is strictly smaller than $\frac32$ for $\varepsilon$ small.

To bound $|\partial_v \phi|$, we integrate \eqref{waveeqn_phi2_r} from $\Gamma_Y$ and obtain an expression for $r \partial_v \phi$ analogous to \eqref{integrate_dvphi}.
Then, by using \eqref{noshift_phi_final},  \eqref{monotonicity_curlyC}, \eqref{bounds_r_blueshift}, \eqref{blueshift_Q},   \eqref{boundlambda_blueshift},  \eqref{bounds_r_improved}, the bootstrap inequalities for $\phi$ and $D_u \phi$, and the fact that $e^{-c(s)v_Y(u)} \le e^{-\mathcal{C}_s(u, v)}$:
\begin{align*}
|r \partial_v \phi|(u, v) &\le C_{\mathcal{N}}\br{ Y + C(N_{\text{i.d.}}) \int_{u_Y(v)}^u |\nu|(u', v)du' } e^{- \mathcal{C}_{s}(u, v)} \\
&\le C_{\mathcal{N}} \br{ Y + (Y - r_{-} + \varepsilon) C(N_{\text{i.d.}}) } e^{- \mathcal{C}_{s}(u, v)},
\end{align*}
for  $v_1$ large. After dividing both sides of the above by $r \ge r_{-} - \varepsilon$ (see \eqref{bounds_r_improved}) and using \eqref{Y_constraint} and \eqref{bounds_r_blueshift} to notice that
\[
\frac{Y}{r_{-} - \varepsilon}  < 2 \qquad \text{ and } \qquad \frac{Y-r_{-} + \varepsilon}{r_{-} - \varepsilon} < \frac{4 \varepsilon}{r_{-}},
\]
for $\varepsilon$ small, we conclude that
\[
|\partial_v \phi|(u, v) \le C_{\mathcal{N}} \br{ 2 + \frac{4 \varepsilon}{r_{-}} C(N_{\text{i.d.}})} e^{- \mathcal{C}_{s}(u, v)} < \frac52 C_{\mathcal{N}} e^{-\mathcal{C}_s(u, v)}.
\]
This result, together with \eqref{closebootstrapphi}, is sufficient to close bootstrap inequality \eqref{phi_bootstrap_blueshift}.

We now improve the bound \eqref{Duphi_bootstrap_blueshift} for $D_u \phi$.
First, we notice that, by integrating \eqref{dulambda_kappa} from $\Gamma_Y$:
\begin{equation} \label{blueshift_nu_explicit}
|\nu|(u, v) = |\nu|(u, v_Y(u))\exp\br{\int_{v_Y(u)}^v \kappa(2K - m^2 r |\phi|^2)(u, v') dv'}.
\end{equation}
 Moreover, by \eqref{waveeqn_phi3_r}:
\[
|r D_u \phi|(u, v) \le |r D_u \phi|(u, v_Y(u)) + C(N_{\text{i.d.}}) C_{\mathcal{N}} \int_{v_Y(u)}^v |\nu|(u, v') e^{-\mathcal{C}_{s}(u, v')} dv'.
\]
Since \eqref{dulambda_kappa}, \eqref{def_curlyC}  and \eqref{bound_2K_blueshift} give\footnote{Here and in the following steps the expression of $\mathcal{C}_s$ plays a prominent role. The same reasoning would not work if we replaced $\mathcal{C}_s(u, v)$ with $(c(s)-\tau)v$.}
\begin{align*}
\partial_v \br{|\nu|(u, v)e^{-\mathcal{C}_{s}(u, v)}} &= |\nu|(u, v) e^{-\mathcal{C}_{s}(u, v)} \br{\kappa (2K - m^2 r |\phi|^2)(u, v) +4K_{-} + 2C \varepsilon} \\
& \ge |\nu|(u, v) e^{-\mathcal{C}_{s}(u, v)}(2K_{-} + C \varepsilon) > 0
\end{align*}
in the bootstrap set $E$, then:
\[
\int_{v_Y(u)}^v |\nu|(u, v') e^{-\mathcal{C}_{s}(u, v')}dv' \le \frac{1}{2K_{-} + C \varepsilon} \int_{v_Y(u)}^v \partial_v \br{|\nu| e^{-\mathcal{C}_s}} dv' \le \frac{|\nu|(u, v) e^{-\mathcal{C}_{s}(u, v)}}{2K_{-}}.
\]
Expressions \eqref{noshift_Duphi_final}, \eqref{blueshift_nu_explicit}, \eqref{bound_2K_blueshift}, \eqref{phi_bootstrap_blueshift} and the definition of $\mathcal{C}_s$ then yield
\begin{align*}
|r D_u \phi|(u, v) &\le C_{\mathcal{N}} |\nu|(u, v) \br{r_+ e^{-c(s)v_Y(u) + (2K_{-} + C  \varepsilon)(v - v_Y(u))} + C(N_{\text{i.d.}}) e^{-\mathcal{C}_s(u, v)}} \\
& \le C_{\mathcal{N}} C(N_{\text{i.d.}}) |\nu|(u, v) e^{-\mathcal{C}_{s}(u, v)},
\end{align*}
for a suitably large $v_1$. 
Hence, the bootstrap closes by choosing $M$ large compared to the initial data (this might possibly reduce the size $\varepsilon$ even further).

Finally, we can express the above estimates in terms $c(s)-\tau$, rather than $\mathcal{C}_s(u, v)$, using \eqref{translate_decay}.
\end{proof}

\begin{remark} \label{rmk:not_necessary}
We observe that:
\begin{itemize}
\item The inequality $\Delta < \varepsilon$ that we assumed in the statement of proposition \ref{prop: blueshift} is not needed for the proof. Such a condition can be replaced by an assumption on the smallness of $\Delta + \varepsilon$ (with respect to the initial data).
\item Proposition \ref{prop: blueshift} can be proved in an analogous fashion if we replace \eqref{def_curlyC} with the following definition of $\mathcal{C}_s$:
\[
\mathcal{C}_s(u, v) = c(s) v_Y(u) - \alpha (2K_{-} + C \varepsilon)(v - v_Y(u)),
\] 
for $(u, v) \in J^+(\Gamma_Y) \cap J^{-}(\gamma) \cap \{v \ge v_1 \}$ and for any $\alpha > 1$.
\end{itemize}
\end{remark}

\subsection{\texorpdfstring{Early blueshift region: bounds along $\gamma$}{Early blueshift region: bounds along gamma}} \label{sec:EBR2}

In this section, we improve the estimates of $\nu$ and $\lambda$ along the future boundary of the early blueshift region.

\begin{lemma}[integral bounds on $\Omega^2$] \label{lemma:IB_Omega}
The following results hold:
\begin{enumerate}
\item If there exists $D >0$ such that
\[
\partial_v \log \Omega^2(w, z) \le -D, \quad \forall\, (w, z) \in J^+(\Gamma_Y) \cap J^{-}(\gamma),
\]
and if $v$ is sufficiently large, then:
\begin{equation} \label{IB_Omega1}
\int_{u_{Y}(v)}^u \Omega^2(\tilde u, v)d\tilde u \le C(N_{\text{i.d.}}, R, Y) \frac{\beta}{1+\beta}v \, e^{-D\frac{\beta}{1+\beta}v}, \quad \forall\, (u, v) \in \gamma.
\end{equation}
\item Let $E \subseteq J^+(\gamma)$ be a past set such that $\gamma \subset E$. If there exists $ D >0$ such that
\begin{equation} \label{assumption_IB}
\partial_v \log \Omega^2(w, z) \le -D, \quad \forall\, (w, z) \in J^+(\Gamma_Y) \cap J^{-}(E),
\end{equation}
and if $v$ is sufficiently large, then:
\[
\int_{u_{\gamma}(v)}^u \Omega^2(\tilde u, v)d\tilde u \le C(N_{\text{i.d.}}, R, Y) (v- v_{\gamma}(u)) e^{-D\frac{\beta}{1+\beta}v}, \quad \forall\, (u, v) \in E.
\]
\end{enumerate}
\end{lemma}
\begin{proof}
We will prove the second statement. Inequality \eqref{IB_Omega1}  follows similarly.

First, we notice that estimate \eqref{noshift_mu} can be improved when $1-\mu$ is restricted to the curve $\Gamma_Y$. Indeed, we have
\[
(1-\mu)(u_Y(z), z) = (1-\mu)(Y, \varpi_+, Q_+) + 2 \frac{\varpi_+ - \varpi(u_Y(z), z)}{Y} + \frac{Q^2(u_Y(z), z) - Q_+^2}{Y^2},
\]
for every $z \ge v_1$. Using \eqref{noshift_Q_final} and \eqref{noshift_varpi_final}:
\begin{equation} \label{bound_gammaY}
\Big |  |1-\mu|(u_Y(z), z) - |1- \mu|(Y, \varpi_+, Q_+) \Big | \lesssim e^{-2c(s)z},
\end{equation}
for every $z$ suitably large. An analogous estimate can be proved for $|1-\mu|(w, v_Y(w))$, for every $w$.

Let us now integrate \eqref{assumption_IB} from $\Gamma_Y$ to $(\tilde u, \tilde v) \in E$ and use the fact that $\tilde v \ge (1+\beta) v_Y(\tilde u)$ to get
\[
\log \frac{\Omega^2(\tilde u, \tilde v)}{\Omega^2(\tilde u, v_Y(\tilde u))} \le -D \frac{\beta}{1+\beta}\tilde v,
\]
and thus
\begin{equation} \label{proof_boundOmega1}
\Omega^2(\tilde u, \tilde v) \le \Omega^2(\tilde u, v_Y(\tilde u)) e^{-D \frac{\beta}{1+\beta} \tilde v}, \quad \forall\, (\tilde u, \tilde v) \in E.
\end{equation}
Now, if we differentiate the relation $r(\cdot, v_Y(\cdot))=Y$, we have
\[
\nu(\cdot, v_Y(\cdot)) + \lambda(\cdot, v_Y(\cdot)) v'_Y(\cdot) \equiv 0 \quad \text{ in } [0, u_Y(v_1)].
\]
Given $(u, v) \in \text{int}E$, we can integrate the above in $[u_{\gamma}(v), u]$, use that $u_{\gamma} = v_{\gamma}^{-1}$ and perform a change of variables (see \eqref{dv_R}): 
\begin{equation} \label{mini_BV_gamma}
\int_{u_{\gamma}(v)}^u \nu(\tilde u, v_Y( \tilde u))d \tilde u = \int_{v_Y(u)}^{v_Y(u_{\gamma}(v))} \lambda(u_{Y}(\tilde v), \tilde v) d\tilde v.
\end{equation}
On the other hand, if we integrate \eqref{proof_boundOmega1} from the curve $\gamma \subset E$ to $(u, v) \in \text{int}E$ and use: \eqref{proof_boundOmega1}, the facts that $\Omega^2 = 4|\nu|\kappa$ and $0 < \kappa \le 1$ in $\mathcal{P}$ and the relations \eqref{mini_BV_gamma}, \eqref{algconstr} and \eqref{bound_gammaY}, in this order:
\begin{align*}
\int_{u_{\gamma}(v)}^u \Omega^2(\tilde u, v)d \tilde u &\le e^{-D\frac{\beta}{1+\beta}v} \int_{u_{\gamma}(v)}^u \Omega^2(\tilde u, v_Y(\tilde u)) d \tilde u \le 4 e^{-D \frac{\beta}{1+\beta}v} \int_{u_{\gamma}(v)}^u |\nu|(\tilde u, v_Y(\tilde u)) d \tilde u \\
&= 4 e^{-D\frac{\beta}{1+\beta}v} \int_{v_Y(u)}^{v_Y(u_{\gamma}(v))} |\lambda|(u_{Y}(\tilde v), \tilde v) d \tilde v \\
&\le C(N_{\text{i.d.}}, R, Y)(v_Y(u_{\gamma}(v)) -v_Y(u)) e^{-D\frac{\beta}{1+\beta}v} \\
& \le C(N_{\text{i.d.}}, R, Y)(v - v_{\gamma}(u))  e^{-D\frac{\beta}{1+\beta}v},
\end{align*}
for $v_1$ suitably large.
\end{proof}

We recall that, in the following, the notation $\lesssim$ refers to constants depending on the initial data, $\eta$, $R$ and $Y$.

\begin{lemma}[bounds along $\gamma$] \label{lemma:lambdanu_blueshift}
Under the assumptions of proposition \ref{prop: blueshift}, if we further suppose that $U$ is sufficiently small, the following inequalities hold for every $(u, v) \in \gamma \cap \{v \ge v_1\}$:
\begin{align}
  e^{-(2K_{-} + c_1 \varepsilon)  \frac{\beta}{1+\beta}v} 
 \lesssim -\lambda(u, v) &\lesssim  e^{-(2K_{-}-c_2  \varepsilon )  \frac{\beta}{1+\beta}v},
  \label{bound_lambda_gamma}\\
 \, u^{-1  +\frac{K_{-}}{K_+}\beta + \tilde{\alpha}} \lesssim -\nu(u, v) &\lesssim  u^{-1  +\frac{K_{-}}{K_+}\beta - \alpha}, \label{bound_nu_gamma}
\end{align}
for some positive constants $\alpha$, $\tilde{\alpha}$, $c_1$ and $c_2$.  The constants $\alpha$ and $\tilde{\alpha}$ depend on the initial data, on $\beta$ and on $\varepsilon$, whereas $c_1$ and $c_2$ depend on the initial data and on $\eta$, $R$ and $Y$.   Moreover, for every
\[
\beta >0
\]
small compared to the initial data, there exists $\varepsilon > 0$ such that
\[
\frac{K_{-}}{K_+} \beta - \alpha(\beta, \varepsilon) > 0.
\] 
\end{lemma}

\begin{remark}[on the estimates along the curve $\gamma$] \label{rmk:lambda_decay}
As $\beta$ tends to zero (see \eqref{def_beta}), the curve $\gamma$ approaches the curve of constant radius $\Gamma_Y$. The importance of a non-zero value for $\beta$ is apparent in the results of lemma \ref{lemma:lambdanu_blueshift}, showing that, along  $\gamma$:
\begin{itemize}
\item Not only the quantity $|\lambda|$ is bounded (as occurs in the redshift, no-shift and early blueshift regions, see \eqref{redshift_lambda_final}, \eqref{noshift_lambda_final} and \eqref{boundlambda_blueshift}) but it also decays at a rate which is integrable in $[v_1, +\infty)$; this result will be vital to prove that $r$ is bounded and bounded away from zero at the Cauchy horizon;
\item The ``blow-up'' behaviour of $|\nu|$ in the $u$ coordinate (property already present in the previous regions, see \eqref{bound_nu_eventhorizon}, \eqref{redshift_nu_final}, \eqref{noshift_nu_final}) is now integrable in $u$: this can be used in alternative to the previous point to show the stability of the Cauchy horizon, once we propagate this result to $J^+(\gamma)$.
\end{itemize}
We also notice that, in $J^+(\Gamma_Y) \cap J^{-}(\gamma)$, requiring $U$ to be small is equivalent to requiring suitably large values of $v_1$, due to the spacelike nature of $\gamma$ and its location in the black hole interior. Indeed, we showed that every $u \in [0, u_Y(v_1)]$ satisfies \eqref{noshift_u_final}:
\[
 u \sim e^{-2K_+ v_Y(u)},
\]
and a similar relation holds for points along $\gamma$, since $v_Y(u) = (1+\beta)^{-1}v_{\gamma}(u)$. On the other hand, in the region $J^+(\gamma)$ the restrictions on the $v$ coordinate do not constrain the ingoing null coordinate.
\end{remark}

\begin{proof}[proof of lemma \ref{lemma:lambdanu_blueshift}] The proof follows the arguments of \cite[lemma 6.6]{CGNS2}, where however Reissner-Nordstr{\"o}m-de Sitter initial data were prescribed along the event horizon. For convenience of the reader, we present an explicit proof in appendix \ref{appendixProofGamma}.
\end{proof}

\subsection{Late blueshift region} \label{section:lateblueshift}
In the causal future of the curve $\gamma$, also dubbed the \textbf{late blueshift region}, the dynamical black hole departs considerably\footnote{By this, we mean that we have less control on the main dynamical quantities, e.g. $K$ and $\varpi$, with respect to the previous regions that we analysed. A precise notion of  `distance' between the dynamical black-hole and the fixed sub-extremal Reissner-Nordstr{\"o}m-de Sitter solution depends on the norm used to control the PDE variables.} from the fixed sub-extremal solution. After showing that  the upper bounds in  \eqref{bound_lambda_gamma} and \eqref{bound_nu_gamma}  propagate to $J^+(\gamma)$, it will follow that $r$ remains bounded away from zero in this entire region. This will be the key ingredient to show the existence, for the dynamical model, of a Cauchy horizon $\mathcal{CH}^+$ beyond which the metric $g$ can be continuously extended. We stress that the smallness of the region $J^+(\gamma)$ is what allows us to propagate decaying estimates  despite the adverse contribution of the blueshift effect (see already lemma \ref{lemma:r_bounded_future}). 

The core idea of the proof is again a bootstrap argument involving the quantity $\partial_v \phi$. The proof, however, differs significantly from the bootstrap procedures of the previous regions, due to the lack of an upper bound\footnote{The main obstructions to this result stemming from the quadratic term $|D_u \phi|^2$ in \eqref{duvarpi2} and the term $\kappa^{-1}$ in \eqref{dvvarpi}. The renormalized Hawking mass was shown to blow up  in the real and massless case \cite{CGNS4} for a class of initial data and by further assuming a lower bound on $|\partial_v \phi|$ along the event horizon.} on the quantities $|\varpi|$ and $|K|$ and due to the absence of a lower bound on $\kappa$. A natural way to close the argument is to control $\phi$ and $D_u \phi$ by integrating the estimate on $\partial_v \phi$ that we bootstrap from $\gamma$.  This is also possible due to a bound on $\partial_v \log \Omega^2$, which is related to $K$ via \eqref{waveeqn_log_new}.

The technical difficulties of the proof can be summarised as follows. Due to the presence of the non-constant charge function $Q$, we cannot apply the soft analysis used in \cite{CGNS2, CGNS4} to propagate the bounds on $|\nu|$ and $|\lambda|$  from the curve $\gamma$. A bootstrap argument lets us circumvent this difficulty and obtain additional estimates on $|\phi|$, $|\partial_v \phi|$, $|D_u \phi|$ and $|Q|$.
Two core ingredients of the proof are given by a bound on $\partial_v \log \Omega^2$, which also plays a crucial role in \cite{VdM1}, and by the sub-extremality condition $\frac{K_{-}}{K_+} > 1$.


\begin{proposition}[estimates in the late blueshift region] \label{prop:lateblueshift}
Under the assumptions of proposition \ref{prop: blueshift}, if $U$ and $\beta$ are sufficiently small and if  $\varepsilon$ is chosen small so that (see lemma \ref{lemma:lambdanu_blueshift})
\[
\frac{K_{-}}{K_+} \beta - \alpha( \beta, \varepsilon) > 0,
\]
then the following inequalities hold for every $(u, v) \in J^+(\gamma) \cap \{r \ge r_{-} - \varepsilon, \, v \ge v_1 \}$:
\begin{align}
0 <-\lambda(u, v) &\lesssim \frac{\Omega^2(u, v)}{\Omega^2(u, v_{\gamma}(u)) }|\lambda|(u, v_{\gamma}(u))+ e^{-\min\{2c(s)-4\tau, 2K_{-}-\tau\}v}, \label{lateblueshift_final_lambda} \\
0 < -\nu(u, v) &\lesssim u^{-1 + \frac{K_{-}}{K_+}\beta - \alpha},  \label{lateblueshift_final_nu}\\
|\partial_v \phi|(u, v) &\lesssim e^{-\br{c(s) -\tau}v}, \label{lateblueshift_final_dvphi} \\
|\phi|(u, v) &\lesssim e^{-\br{c(s) -\tau}v_{\gamma}(u)}, \label{lateblueshift_final_phi} \\
|D_u \phi|(u, v) &\lesssim u^{-1+\frac{K_{-}}{K_+}\beta -\alpha} e^{-\br{c(s) -\tau}v_{\gamma}(u)}, \label{lateblueshift_final_Duphi} \\
\partial_v \log \Omega^2(u, v) &=  -2K_{-} + O(\varepsilon), \label{lateblueshift_final_dvlog} \\
|Q|(u, v) & \lesssim |Q|(u, v_{\gamma}(u)) +e^{-2\br{c(s) -\tau}v_{\gamma}(u)},
\end{align}
where $\tau > 0$, $\tau = O(\beta)$.
\end{proposition}
\begin{proof}
As for the previous proofs, we will repeatedly use that $0 < \kappa \le 1$ and $\nu  < 0$ in $\mathcal{P}$, and that $\lambda < 0$ in $J^+(\gamma)$ (see lemma \ref{lemma:constantu} and lemma \ref{lemma:lambdanu_blueshift}). The value of $\varepsilon$ is taken small so that, in particular:
\begin{equation} \label{bounds_r_lateblueshift}
\frac34 r_{-} < r_{-} - \varepsilon \le r \le Y \quad \text{ in } J^{+}(\Gamma_Y) \cap \{r \ge r_{-} - \varepsilon, \, v \ge v_1 \}.
\end{equation} 
Moreover, we will use that
\begin{equation} \label{upperBoundcs}
c(s) < 2K_{-}.
\end{equation}
The above follows from  \eqref{def_cs} and from the fact that $\frac{K_{-}}{K_+} > 1$ (see appendix A in \cite{CGNS3}).

\textbf{Setting up the bootstrap procedure}: Define $E$ as the set of points $q$ in $\mathcal{D} \coloneqq J^+(\gamma) \cap \{r \ge r_{-} - \varepsilon, \, v \ge v_1\}$ such that the following inequalities hold for every $(u, v)$ in $J^{-}(q) \cap \mathcal{D}$:
\begin{align}
|r \partial_v \phi|(u, v) &\le M e^{-\br{c(s) -\tau}v}, \label{phi_bootstrap_lateblueshift}\\
|\lambda|(u, v) &\le  C(N_{\text{i.d.}}),\label{lambda_bootstrap_lateblueshift}  \\
\partial_v \log \Omega^2(u, v) &\le -K_{-}, \label{dvlog_bootstrap_lateblueshift}
\end{align}
 where $M=M(N_{\text{i.d.}}) > 0$ is a large constant and $\tau > 0$, $\tau = O(\beta)$. Moreover, notice that, whenever $(u, v)$ belongs to $\gamma$, we have $(u, v) = (u_{\gamma}(v), v) = (u, v_{\gamma}(u))$.  
 
\textbf{Closing the bootstrap}: let us fix $(u, v) \in E$. 
By integrating the bootstrap inequality \eqref{phi_bootstrap_lateblueshift} and using \eqref{phi_blueshift_final} and the above bounds on $r$, we can write:
\begin{equation} 
|\phi|(u, v) \le |\phi|(u, v_{\gamma}(u)) + C(N_{\text{i.d.}}, M) \int_{v_{\gamma}(u)}^v e^{-\br{c(s) -\tau}v'} dv'  \le C(N_{\text{i.d.}}, \eta, R, Y, M) e^{-\br{c(s) -\tau}v_{\gamma}(u)}. \label{bound_phi_lateblueshift}
\end{equation}
Since the constants $R$, $Y$ and $\eta$ will be left unchanged in the next steps, we can absorb those in $C = C(N_{\text{i.d.}}, M) > 0$.

In the causal future of $\gamma$, the charge function $Q$ remains bounded. Indeed, by integrating \eqref{maxwelleqnv} and using \eqref{bounds_r_lateblueshift},  \eqref{phi_bootstrap_lateblueshift} and  \eqref{bound_phi_lateblueshift}:
\begin{equation} \label{bound_Q_lateblueshift}
|Q|(u, v) \le |Q|(u, v_{\gamma}(u)) + C(N_{\text{i.d.}}, M) e^{-2(c(s)-\tau)v_{\gamma}(u)}.
\end{equation}

Now, \eqref{rwaveeqn}, \eqref{def_kappa} and the bounds \eqref{bounds_r_lateblueshift}, \eqref{dvlog_bootstrap_lateblueshift},  \eqref{bound_phi_lateblueshift} and  \eqref{bound_Q_lateblueshift}  give
\[
|\partial_v (r\nu)|(u, v) \lesssim \Omega^2(u, v)\le  - \frac{\Omega^2 (u, v)}{K_{-}} \partial_v \log \Omega^2 (u, v) = - \frac{\partial_v [ \Omega^2](u, v)}{K_{-}}.
\]
After integrating the latter from $\gamma$ and exploiting  \eqref{bound_nu_gamma},  \eqref{def_kappa} and the boundedness of $r$:
\[
|\nu|(u, v) \lesssim |\nu|(u, v_{\gamma}(u)) + \Omega^2(u, v_{\gamma}(u)) \le C(N_{\text{i.d.}}) u^{-1 + \frac{K_{-}}{K_+} \beta - \alpha}.
\]

To estimate $|D_u \phi|$, we first notice that bound \eqref{Duphi_blueshift_final}, obtained for the early blueshift region, and \eqref{bound_nu_gamma} entail
\[
|D_u \phi|(u, v_{\gamma}(u)) \lesssim u^{-1 + \frac{K_{-}}{K_+}\beta - \alpha} e^{-(c(s)- \tau)v_{\gamma}(u)}.
\]
Therefore, the relations \eqref{bounds_r_lateblueshift}, \eqref{phi_bootstrap_lateblueshift}, \eqref{dvlog_bootstrap_lateblueshift}, \eqref{bound_phi_lateblueshift},  the fact that $\Omega^2=-4\nu \kappa$ and the above bounds on $\nu$ and $Q$ can be used to integrate \eqref{waveeqn_phi3_r} as follows: 
\begin{align*}
|r D_u \phi|(u, v) &\lesssim |r D_u \phi|(u, v_{\gamma}(u)) + \int_{v_{\gamma(u)}}^v |\nu \partial_v \phi|(u, v') dv' + C(N_{\text{i.d.}}) \int_{v_{\gamma(u)}}^v   \Omega^2 |\phi| (u, v') dv'   \\
&\lesssim u^{-1 + \frac{K_{-}}{K_+}\beta - \alpha} e^{-(c(s)- \tau)v_{\gamma}(u)} + C(N_{\text{i.d.}}, M)  u^{-1 + \frac{K_{-}}{K_+}\beta - \alpha} e^{-(c(s) - \tau)v_{\gamma}(u)} \\
&- C(N_{\text{i.d.}}, M) e^{-(c(s) - \tau)v_{\gamma}(u)} \int_{v_{\gamma}(u)}^v \frac{\partial_v [\Omega^2](u, v')}{K_{-}} dv',
\end{align*}
where all above constants are positive. So, since $\Omega^2(u, v_{\gamma}(u)) \lesssim u^{-1+\frac{K_{-}}{K_+}\beta-\alpha}$:
\begin{equation} \label{lateblueshift_temp_duphi}
|D_u \phi|(u, v) \le C(N_{\text{i.d.}}, M) u^{-1 + \frac{K_{-}}{K_+}\beta - \alpha} e^{-(c(s)-\tau)v_{\gamma}(u)}.
\end{equation}

To close bootstrap inequality \eqref{dvlog_bootstrap_lateblueshift}, we first note that the results \eqref{blueshift_final_dvlog} and \eqref{bound_2K_blueshift} obtained in the early blueshift region give
\begin{align}
|\partial_v \log \Omega^2(u_{\gamma}(v), v) + 2K_{-}| &\le |\partial_v \log \Omega^2(u_{\gamma}(v), v) - 2K(u_{\gamma}(v), v)| + |2K(u_{\gamma}(v), v) + 2K_{-} | \nonumber \\
&= o(1) + O(\varepsilon), \label{preliminary_dvlog_gamma}
\end{align}
where $o(1)$ refers to the limit $v_1 \to +\infty$. On the other hand, it is useful to notice that \eqref{phi_bootstrap_lateblueshift} and the previous bound on $|D_u \phi|$ yield\footnote{We recall that $\frac{K_{-}}{K_+}\beta - \alpha > 0$ by assumptions.} 
\begin{align}
\int_{u_{\gamma}(v)}^u |D_u \phi \, \partial_v \phi|(u', v) du' = o(1) \label{preliminary_lateblueshift1},
\end{align}
as $v_1 \to +\infty$.

Furthermore, \eqref{dvlog_bootstrap_lateblueshift} also holds in $J^+(\Gamma_Y) \cap J^{-}(E)$ (see \eqref{blueshift_final_K} and \eqref{blueshift_final_dvlog}). By applying lemma \ref{lemma:IB_Omega} with $D=K_{-}$:
\begin{equation}
\int_{u_{\gamma}(v)}^u \Omega^2(u', v) du' \le C(N_{\text{i.d.}}) (v-v_{\gamma}(u)) e^{-K_{-} \frac{\beta}{1+\beta}v} \le  e^{-K_{-} \frac{\beta}{2+2\beta}v}, \label{preliminary_lateblueshift2}
\end{equation}
for $v_1$ sufficiently large.
Now, if we integrate \eqref{waveeqn_log}  from the curve $\gamma$ and use \eqref{bounds_r_lateblueshift}, \eqref{lambda_bootstrap_lateblueshift}, \eqref{preliminary_dvlog_gamma}, \eqref{preliminary_lateblueshift1}, \eqref{preliminary_lateblueshift2}  and the previous bound on $|\nu|$:
\begin{align*}
 \partial_v \log \Omega^2(u, v) +  2K_{-} +O(\varepsilon) + o_{v_1}(1)&\le   \partial_v \log \Omega^2(u, v) - \partial_v \log \Omega^2(u_{\gamma}(v), v)  \\
&= o_{v_1}(1) + o_U(1),
\end{align*}
where $o_U(1)$ refers to the limit $U \to 0$.
This closes bootstrap inequality  \eqref{dvlog_bootstrap_lateblueshift}. Notice that a similar reasoning gives
\[
\partial_v \log \Omega^2(u, v) = -2K_{-}+O(\varepsilon)
\]
for $v_1 $ large and $U$ small. Hence, after integration:
\begin{equation} \label{estimateOmega2Blueshift}
e^{-(2K_{-} +\tau)(v- v_{\gamma}(u))} \le \frac{\Omega^2(u, v)}{\Omega^2(u, v_{\gamma}(u))} \le e^{-(2K_{-} -\tau)(v- v_{\gamma}(u))},
\end{equation}
where $\tau > 0$, $\tau = O(\beta)$ and we assumed that $\varepsilon$ is small compared to $\beta$.\footnote{Recall that for every small choice of $\beta$, the value of $\varepsilon$ may possibly be taken smaller so that $\frac{K_{-}}{K_+}\beta-\alpha > 0$. See also the proof of lemma \ref{lemma:lambdanu_blueshift}.}

To close the bootstrap on $|\lambda|$: integrate \eqref{raych_v} and use \eqref{bounds_r_lateblueshift} to get
\[
\frac{|\lambda|}{\Omega^2}(u, v) \le \frac{|\lambda|}{\Omega^2}(u, v_{\gamma}(u)) + C(N_{\text{i.d.}}) \int_{v_{\gamma}(u)}^v \frac{|\partial_v \phi|^2}{\Omega^2}(u, v')dv'. 
\]
Then, after using \eqref{phi_bootstrap_lateblueshift} and \eqref{estimateOmega2Blueshift}:
\begin{align}
|\lambda|(u, v) \le \frac{\Omega^2(u, v)}{\Omega^2(u, v_{\gamma}(u))} \br{|\lambda|(u, v_{\gamma}(u)) + C(N_{\text{i.d.}}, M)e^{-(2K_{-} + \tau)v_{\gamma}(u)} \int_{v_{\gamma}(u)}^v e^{-(2c(s) -2K_{-} + \tau)v'}dv'}.
\end{align}
We finally obtain \eqref{lateblueshift_final_lambda} after distinguishing between the cases $c(s) \le K_{-}$ and $c(s) > K_{-}$, and exploiting \eqref{estimateOmega2Blueshift} again. In particular, \eqref{lambda_bootstrap_lateblueshift} is closed.

For the next step, it is useful to notice that \eqref{estimateOmega2Blueshift}, \eqref{bound_phi_lateblueshift}, \eqref{upperBoundcs} and \eqref{bound_nu_gamma} (recall that $\Omega^2 \sim |\nu|$ along $\gamma$) give:
\begin{align}
\int_{u_{\gamma}(v)}^u \br{\Omega^2 |\phi|}(u', v)du' &\le C(N_{\text{i.d.}}, M) \int_{u_{\gamma}(v)}^u \Omega^2(x, v_{\gamma}(x)) e^{-(2K_{-} - \tau)(v - v_{\gamma}(x))} e^{-(c(s) -\tau)v_{\gamma}(x)} dx \nonumber \\
&\le C(N_{\text{i.d.}}, M) u^{\frac{K_{-}}{K_+}\beta - \alpha} e^{-(c(s) -\tau)v} \label{intOmegaphi}.
\end{align}

To close the bootstrap inequality on $\partial_v \phi$, we integrate \eqref{waveeqn_phi2_r} (see  \eqref{integrate_dvphi} for an explicit expression) and use \eqref{def_kappa}, \eqref{phi_blueshift_final}, the boundedness of $r$, of the charge $Q$, \eqref{lateblueshift_final_lambda},  \eqref{intOmegaphi}, \eqref{lateblueshift_temp_duphi}, \eqref{bound_lambda_gamma} and \eqref{upperBoundcs}:
\begin{align*}
|r \partial_v \phi|(u, v) &\lesssim |r \partial_v \phi|(u_{\gamma}(v), v) + \int_{u_{\gamma}(v)}^u \br{|\lambda D_u \phi| + \Omega^2 |\phi|}(u', v) du' \\
&\le C(N_{\text{i.d.}}) e^{ - (c(s) -\tau)v}+ C(N_{\text{i.d.}}, M) e^{-(2K_{-}- \tau)v} \int_{u_{\gamma}(v)}^u e^{(2K_{-}-c(s))v_{\gamma}(x)}x^{-1+\frac{K_{-}}{K_+}\beta -\alpha} dx \\
&+ C(N_{\text{i.d.}}, M) u^{\frac{K_{-}}{K_+}\beta-\alpha} e^{-\min \{2c(s) -2\tau, 2K_{-}-\tau\}v}  +C(N_{\text{i.d.}}, M) u^{\frac{K_{-}}{K_+} \beta - \alpha}  e^{-(c(s)-\tau)v}   \\
&\le \br{C(N_{\text{i.d.}}) +  C(N_{\text{i.d.}}, M)u^{\frac{K_{-}}{K_+}\beta -\alpha}  } e^{-(c(s) -\tau)v},
\end{align*}
assuming that $\beta$ and $\varepsilon$ are small compared to the initial data. Inequality \eqref{phi_bootstrap_lateblueshift} is then improved, provided that $U$ is sufficiently small.
\end{proof}

\begin{lemma} \label{lemma:r_bounded_future}
For every $\tilde{\varepsilon} > 0$ small, there exists $U_{\tilde{\varepsilon}} > 0$ such that
\[
r(u, v) > r_{-} - 2\tilde{\varepsilon}
\]
for every $(u, v)$ in $J^+(\gamma) \cap \{u < U_{\tilde{\varepsilon}}\}$.
\end{lemma}
\begin{proof}
The proof is analogous to the one of lemma 7.2 in \cite{CGNS2}, where the results on $\nu$ and $\lambda$ exploited in such a proof can be replaced with those of proposition \ref{prop:lateblueshift}.
\end{proof}

\begin{corollary} \label{coroll: cauchyhor}
Let $U_{\tilde{\varepsilon}} > 0$ be the value provided by lemma \ref{lemma:r_bounded_future} for some $\tilde{\varepsilon} < \frac{r_{-}}{2}$. Then:
\[
[0, U_{\tilde{\varepsilon}}] \times [v_0, +\infty) \subset \mathcal{P},
\]
and the estimates of proposition \ref{prop:lateblueshift} hold in $J^+(\gamma) \cap \{u < U_{\tilde{\varepsilon}}\}$. Moreover, the limit
\[
\lim_{v \to +\infty} r(u, v) 
\]
exists and is finite for every $u \in [0, U_{\tilde{\varepsilon}}]$, and
\[
\lim_{u \to 0} \lim_{v \to +\infty} r(u, v) = r_{-}.
\]
\end{corollary}
\begin{proof}
The inclusion follows from the extension criterion (see remark \ref{rmk:extension_principle}) and lemma \ref{lemma:r_bounded_future}. The limits are well-defined by the monotonicity of $r$ in $J^+(\gamma)$, namely the fact that $\lambda < 0$ (see \eqref{noshift_lambda_final} and lemma \ref{lemma:constantu}) and $\nu < 0$. The final limit follows from lemma \ref{lemma:r_bounded_future}.
\end{proof}

\subsection{A continuous extension} \label{section:C0extension}
The construction obtained so far allows us to represent the Cauchy horizon  $\mathcal{CH}^+$ of the dynamical black hole as the level set of a regular function. The problem of constructing a continuous extension is well understood (see e.g.  \cite{Dafermos_2005_uniqueness, CGNS4, LukOh1, VdM1}), so we only provide an overview.

In particular, let $(u, V)$ be a new coordinate system defined as follows. Given a sufficiently regular function $f \colon [v_0, +\infty) \to \bbR$ such that the limit of its first derivative as $v \to +\infty$ exists and such that
\begin{equation} \label{der_condition}
\frac{df(v)}{dv} \underset{v \to +\infty}{\longrightarrow} -2K_{-},
\end{equation}
we define
\begin{equation}\label{def_Vcoord}
V(v) \coloneqq 1 - \int_{v}^{+\infty} e^{f(y)}dy, \quad \forall\, v \in [v_0, +\infty).
\end{equation}
The choice $f(v) = -2K_{-} v$ was taken into account in \cite{LukOh1} to construct a continuous extension in the case $\Lambda = 0$, where the authors studied the evolution of two-ended asymptotically flat initial data. Here, we are going to consider any function $f$ such that
\begin{equation} \label{def_fprimed}
\frac{df(v)}{dv} = 2K(u_{\gamma}(v), v), \quad \forall\, v \in [v_0, +\infty),
\end{equation}
as done in \cite{KehleVdM1}. This gives, specifically,
\[
\mathcal{C H}^+ = V^{-1}(\{1\}).
\]
The solutions to our IVP expressed in the coordinate system $(u, V)$ will be denoted by a subscript, e.g. $r_V(u, V) = r(u, v)$ and $\lambda_V (u, V) = \partial_V r_V(u, V) = e^{-f(v)} \lambda(u, v)$. We designate the smallest value that $V$ can take by $V_0 = V(v_0)$.

In the following, we restrict the range of the ingoing null coordinate to the set $[0, U]$, where we use again the letter $U$ to denote a value chosen sufficiently small, so that corollary \ref{coroll: cauchyhor} holds for some fixed $\tilde{\varepsilon} < \frac{r_{-}}{2}$. 

In this context, we show that the metric $g$ and the scalar field $\phi$ admit a continuous extension beyond the Cauchy horizon, in the sense of theorem 11.1 in \cite{Dafermos_2005_uniqueness}. This is a strong signal for the potential failure of conjecture \ref{conjcontinuous}, whose definitive resolution however requires a complementary study of the exterior problem.

\begin{theorem}[A $C^0$ extension of $g$ and $\phi$] \label{thm:general_C0_extension}
For every $0 < \delta_U < U$, the functions $r_V$, $\Omega^2_V$ and $\phi_V$ admit a continuous extension to the set $[\delta_U, U] \times [V_0, 1]$. Moreover, on $[\delta_U, U]$, the functions $r_V(\cdot, 1)$ and $\Omega^2_V(\cdot, 1)$ are positive.
\end{theorem}
\begin{proof}
The proof is analogous to that of theorem 11.1 in \cite{Dafermos_2005_uniqueness} and of proposition 8.14 in \cite{LukOh1}. We summarize the main ideas for the reader's convenience.

The extension for $r$ follows after noticing that, by monotonicity, the following limit
\begin{equation} \label{limitr_exists}
 \lim_{v \to +\infty}r(u, v) =\lim_{V \to 1} r_V(u, V) \eqqcolon  r_V(u, 1)
\end{equation}
exists for every $u \in [\delta_U, U]$.
Moreover, the above convergence is uniform in $u$  due to the decay of $|\lambda|$ proved in the late blueshift region (see \eqref{lateblueshift_final_lambda} and \eqref{estimateOmega2Blueshift}). Therefore, $r_V(\cdot, 1)$ is continuous. This is sufficient to show Cauchy-continuity for $r$.

The extensions of $\phi$ and $\log \Omega^2$ follow similarly.
\end{proof}

\section{\texorpdfstring{No-mass-inflation scenario and $H^1$ extensions}{No-mass-inflation scenario and H1 extensions}} \label{sec:nomassinflation}

Heuristics and numerical work (see section \ref{subsec:review_work}) suggest the validity of the exponential Price law upper bound \eqref{assumption:expPricelaw} along the event horizon of a solution asymptotically approaching (in the sense of section \ref{section:assumptions}) a Reissner-Nordstr{\"o}m-de Sitter black hole. Compared to the case $\Lambda = 0$, however, the exact asymptotics given by the constant value $s$ in \eqref{assumption:expPricelaw} play an even more important role, since the  redshift (resp. blueshift) effect is expected to give an exponential contribution to the decay (resp. growth) of the main dynamical quantities: the value of $s$ then determines the leading exponential contribution. This has a decisive role in determining the final fate of the dynamical black hole: a dominance of the blueshift effect is expected to cause geometric quantities such as the renormalized Hawking mass to blow up at the Cauchy horizon (this is the \textbf{mass inflation} scenario \cite{IsraelPoisson, Dafermos_2005_uniqueness, CGNS4}), whereas a sufficiently fast decay of the scalar field along $\mathcal{H}^+$ may allow to extend the dynamical quantities across the Cauchy horizon in such a way that they are still weak solutions to the Einstein-Maxwell-charged-Klein-Gordon system (we have, in this case, \textbf{lack of global uniqueness}).

In \cite{CGNS4}, where the massless, real scalar field model was studied, any positive value of $s$ was allowed  and a further lower bound on the scalar field was assumed along the event horizon.\footnote{With respect to our exponential Price law, the work \cite{CGNS4} uses a different convention. For every $\epsilon > 0$, they  choose initial data such that $e^{-(SK_+ + \epsilon)v} \lesssim |\partial_v \phi|(0, v) \lesssim e^{-(SK_+ - \epsilon)v}$ along the event horizon, for some $S > 0$.} There, different relations between $s$ and the quantity\footnote{See also appendix A in \cite{CGNS3} for a proof of the positivity of $\rho - 1$. The case $\rho = 1$ is associated to extremal black holes.}
\begin{equation} \label{def_rho}
\rho \coloneqq \frac{K_{-}}{K_+} > 1
\end{equation}
led to different outcomes. The authors found a set in the $\rho-s$ plane for which mass inflation can be excluded for generic initial data: in this context, an $H^1$ extension of the solution was built and the so-called Christodoulou-Chru{\'s}ciel version of strong cosmic censorship was called into question. On the other hand, they found another set in the $\rho-s$ plane, disjoint from the previous one, for which mass inflation was guaranteed. This suggests that the $H^1$ version of the strong cosmic censorship, in this spherically symmetric toy model with $\Lambda > 0$, does hold if such initial data can be attained when starting the evolution from a suitable spacelike hypersurface. 

In the following, by requiring an upper bound (but no lower bound) on the scalar field along the event horizon (see \eqref{assumption:expPricelaw})  we prove, despite the lack of many monotonicity results, that mass inflation can be avoided generically\footnote{Provided that $\rho$ and $s$ satisfy \eqref{assumption:no_inflation}. We emphasize that, to study the admissibility of such values of $\rho$ and $s$, a rigorous analysis of the exterior problem is required. Although it has been established \cite{DyatlovLinear, BarretoZworski}  that the value of $s$ corresponding to the decay of the (real or complex) scalar field along $\mathcal{H}^+$ is determined by the spectral gap of the Laplace-Beltrami operator, a rigorous determination of such $s$ is still generally lacking. } also in the massive and charged case. With respect to the case of a real-valued scalar field, we extend the set of possible values of $s$ for which $H^1$ extensions can be constructed (see Fig.\ \ref{fig:no_mass_inflation1}), in agreement with the results obtained for the linear problem \cite{CostaFranzen, HintzVasy1}. 

Also in this section, we set $U$ as the maximum of the ingoing null coordinate, where $U$ is chosen suitably small, so that corollary \ref{coroll: cauchyhor} holds for a fixed $\tilde{\varepsilon} < \frac{r_{-}}{2}$. 

\subsection{No mass inflation} \label{subsec:nomassinflation}
\begin{theorem}[A sufficient condition to prevent mass inflation] \label{thm:no_mass_inflation}
Given the characteristic IVP of section \ref{section:ivp}, choose initial data for which
\begin{equation} \label{assumption:no_inflation}
s > K_{-} \quad \text{ and } \quad \rho = \frac{K_{-}}{K_+} < 2
\end{equation}
are satisfied.
Then, there exists a positive constant $C$, depending only on the initial data, such that
\[
|\varpi|(u, v) \le C, \quad \forall\, (u, v) \in J^+(\gamma),
\]
provided that the value of $U$ is sufficiently small.
\end{theorem}
\begin{proof}
We first present the main idea of the proof. Here, we choose the constants $\eta$, $\delta$, $R$, $\Delta$, $\varepsilon$, $Y$ and $\beta$ introduced in section \ref{section:stability} in such a way that the assumptions of proposition \ref{prop:lateblueshift} are satisfied.

 Using \eqref{algconstr}, equation \eqref{dvvarpi} can be cast as
\[
\partial_v \varpi =- \frac{|\theta|^2}{r \lambda} \varpi + \frac{\lambda}{2} m^2 r^2 |\phi|^2 + Qq \Im \br{\phi \overline{\theta}} + \frac{|\theta|^2}{2 \lambda} \br{1 + \frac{Q^2}{r^2} - \frac{\Lambda}{3}r^2},
\]
and thus, for every $(u, v) \in J^+(\gamma)$:
\begin{align}
\varpi(u, v) &= \varpi(u, v_{\gamma}(u))e^{-\int_{v_{\gamma}(u)}^v \frac{|\theta|^2}{r \lambda}(u, v')dv'} +  \label{varpi_explicit}\\
&+ \int_{v_{\gamma}(u)}^v \left[ \frac{\lambda}{2} m^2 r^2 |\phi|^2 + Q q \Im \br{\phi \overline{\theta}} + \frac{|\theta|^2}{2 \lambda} \br{1 + \frac{Q^2}{r^2} - \frac{\Lambda}{3}r^2} \right](u, v') e^{-\int_{v'}^v \frac{|\theta|^2}{r \lambda}(u, y)dy} dv'. \nonumber
\end{align}
We recall that $\varpi(u, v_{\gamma}(u)) \to \varpi_+$ as $u \to 0$ (see proposition \ref{prop: blueshift}) and notice that the quantities inside the above square bracket, except for $\frac{|\theta|^2}{2 \lambda}$, either decay or are bounded in $L^{\infty}$ norm (see proposition \ref{prop:lateblueshift}). To control $\varpi$  and extend it continuously to the Cauchy horizon is then sufficient to have a uniform integral bound on $\frac{|\theta|^2}{|\lambda|}$ (see also theorem \ref{thm:general_C0_extension}, where uniform estimates were used to extend the metric). In the following, we bootstrap such an integral bound from the curve $\gamma$. This allows us to get a lower bound on $|\lambda|$. We are then able to close the bootstrap argument using \eqref{lateblueshift_final_dvphi} and  \eqref{assumption:no_inflation}. The idea of bootstrapping such an integral bound is taken from \cite{CGNS4}, where, however, the expression for $\varpi$ is considerably simpler and further restrictions on the values of $s$ and $\beta$ are required since less estimates are available in the region $J^+(\gamma)$.  We now proceed with the detailed proof.

In the following, we repeatedly use the boundedness of $r$. Let $n \in \bbN_0$ be fixed and large with respect to the initial data. We denote by $E$ the set of points $q \in J^+(\gamma)$ such that the following inequality holds for every $(u, v) \in J^{-}(q) \cap J^+(\gamma)$:
\begin{equation} \label{noinflation_bootstrap}
\int_{v_{Y}(u)}^v \frac{|\theta|^2}{|\lambda|}(u, v') dv' = \int_{v_{Y}(u)}^v \frac{r^2|\partial_v \phi|^2}{|\lambda|}(u, v') dv' \le \frac{1}{n}.
\end{equation}
The above holds for every $(u, v) \in \gamma$, due to \eqref{phi_blueshift_final} and \eqref{bound_lambda_gamma}, provided that $U$ and $\beta$ are sufficiently small.

Now, let $(u, v) \in E$. Notice that $|r - r_{-}|$ and $|Q-Q_+|$ are small in $E$, provided that $U$ is sufficiently small, due to the results of proposition \ref{prop:lateblueshift} and to corollary \ref{coroll: cauchyhor}. Moreover, the reasoning at the beginning of the current proof shows that $|\varpi - \varpi_+|$ is small in $E$, too.  Therefore, by \eqref{def_surfacegrav}:
\begin{equation} \label{supbound_2K}
2K(u, v) = -2K_{-} + o(1),
\end{equation}
where $o(1)$ refers to the limit $U \to 0$.

This gives a lower bound on $|\lambda|$ as follows. First, notice that \eqref{raych_recast} gives
\[
0 < \frac{\nu}{1-\mu}(u, v) \le \frac{\nu}{1-\mu}(u, v_Y(u))
\]
and, by construction, the same relation holds in $J^{-}(u, v) \cap J^+(\gamma)$.
Expressions  \eqref{bound_gammaY}, \eqref{mini_BV_gamma},  \eqref{algconstr} and the definition of the curve $\gamma$ yield 
\begin{align} 
\int_{u_{\gamma}(v)}^u \frac{\nu}{1-\mu}(\tilde u, v) d\tilde u &\le \int_{ u_{\gamma}(v)}^u \frac{\nu}{1-\mu}(\tilde u, v_{Y}(\tilde u)) d\tilde u  \nonumber \\
&\lesssim \frac{1}{C(N_{\text{i.d.}}) + o(1)} \int_{v_Y(u)}^{v_Y(u_{\gamma}(v))} \br{C(N_{\text{i.d.}}) + e^{-2c(s)v'}}dv' \nonumber \\
&\le  \frac{v- v_{\gamma}(u)}{(1+o(1))(1+\beta)},  \label{preliminary_upp_IB}
\end{align}
where the notation $o(1)$ refers to the limit $U \to 0$.
Thus, if we integrate \eqref{blueshift_lambda_eqn} from $\gamma$ and use \eqref{bound_lambda_gamma}, \eqref{phi_blueshift_final}, \eqref{supbound_2K}, \eqref{preliminary_upp_IB}:
\begin{align}
|\lambda|(u, v) &= |\lambda|(u_{\gamma}(v), v) e^{\int_{u_{\gamma}(v)}^u \frac{\nu}{1-\mu}(2K - m^2 r |\phi|^2)(u', v) du' } \nonumber \\
&\ge C(N_{\text{i.d.}}) e^{-(2K_{-}+c_1\varepsilon)\frac{\beta}{1+\beta}v} e^{-(2K_{-} + O(\beta) + o(1))(v- v_{\gamma}(u))} \nonumber \\
& \gtrsim C(u) e^{-(2K_{-} + c_1 \varepsilon)\br{1+O(\beta)+o(1)}v}, \label{noinflation_lowlambda}
\end{align}
where we defined
\[
C(u) \coloneqq e^{2K_{-}(1+O(\beta)+o(1))v_{\gamma}(u)}.
\]

We now close the bootstrap argument. Using \eqref{noinflation_lowlambda} and  \eqref{lateblueshift_final_dvphi}:
\[
\int_{v_{\gamma}(u)}^v \frac{|\theta|^2}{|\lambda|}(u, v')dv' \lesssim \frac{1}{C(u)} \int_{v_{\gamma}(u)}^v e^{(-2 c(s)+2K_{-} + C(\beta, \varepsilon) + o(1))v'}dv',
\]
for some $C(\beta, \varepsilon) > 0$ that goes to zero as $\beta \to 0$ or as $\varepsilon \to 0$. Under our assumptions, the last exponent is negative. Indeed, given $s$ and $\rho$ as in \eqref{assumption:no_inflation}, there exists  $\delta_{\varpi} > 0$ and we can choose $\eta$  sufficiently small in proposition \ref{prop: redshift} (see also remark \ref{remark:mainconstants}), so that
\begin{equation} \label{noinflation_aux}
s > K_{-} + \delta_{\varpi} \quad \text{and} \quad \rho < 2 - \frac{\delta_{\varpi} + \eta}{K_+}.
\end{equation}
Therefore, by \eqref{def_cs} we have:
\[
c(s) > K_{-} + \delta_{\varpi}.
\]

Based on the initial data, we can choose $\beta$, $\varepsilon$ and $U$ suitably small so that $C(\beta, \varepsilon) < \delta_{\varpi}$. Thus, if $U$ is sufficiently small, there exists $a > 0$ such that:
\[
\int_{v_{\gamma}(u)}^v \frac{|\theta|^2}{|\lambda|}(u, v')dv' \le C(N_{\text{i.d.}}) e^{-a v_{\gamma}(u)},
\]
and bootstrap inequality \eqref{noinflation_bootstrap} closes, possibly after choosing an even smaller value of $U$.
\end{proof}

\begin{remark}
The linear analysis performed in \cite{CostaFranzen} shows the analogue of a no-mass-inflation regime in the entire region $K_{-} < \min\{s, 2K_+\}$. This coincides with the region of the $\rho-s$ plane given in theorem \ref{thm:no_mass_inflation} (see also Fig.\ \ref{fig:no_mass_inflation1}). The linear analysis in \cite{HintzVasy1} presents, using different techniques, an analogous result where the regularity of the extension depends on the decay of the initial data. 
\end{remark}

\begin{remark}[On mass inflation]
In \cite{CGNS4}, sufficient conditions to show the presence of mass inflation were found. By assuming a lower bound on the scalar field along the event horizon, namely by requiring
\[
e^{-(SK_+ + \varepsilon)v} \lesssim |\partial_v \phi|(0, v) \lesssim e^{-(SK_+ - \varepsilon)v},
\]
for $0 < S < \min \{\rho, 2\}$ and $\varepsilon > 0$, it was possible to show that $\lim_{v \to +\infty}|\varpi|(u, v) = +\infty$, for every $u \in ]0, U]$. 

In the charged scalar field case, oscillations affect this scenario and integral bounds for $|\partial_v \phi|$ are expected to be more relevant. Furthermore, the monotonicity of the quantities $\varpi$, $\theta$, $\zeta$ and $\kappa$ play a prominent role in the proof of the above result in the real, massless case. In particular, in \cite{CGNS4}, it was possible to prove that the mass inflation scenario is strongly related with the behaviour of the integral 
in \eqref{noinflation_bootstrap}. In particular, in the real and massless case:
\begin{enumerate}
\item The convergence of the integral in \eqref{noinflation_bootstrap} implies 
\[
\int_{v_{\gamma}(u)}^v \frac{|\theta|^2}{|\lambda|}(u, v')dv' \to 0, \quad \text{ as } u \to 0,
\]
whereas in the current case it is not even clear if the limit is well-defined. Even working with subsequences does not allow to immediately control the quantity $\varpi$.
\item The lower bound on the scalar field can be propagated from the event horizon to $J^+(\gamma)$, due to the positivity of $\zeta$ and $\theta$ in $J^+(\mathcal{A})$. In the current case, on the other hand, the system does not include such explicit  monotonicity properties (for instance we recall that $\zeta$, $\theta$ and $\phi$ are complex-valued quantities). 
\end{enumerate}
\end{remark}

\subsection{\texorpdfstring{Construction of an $H^1$ extension}{Construction of an H1 extension}} \label{subsec:constructionH1extension}
By prescribing initial data such that conditions  \eqref{assumption:no_inflation} are satisfied, we proved that the solution to the Einstein-Maxwell-charged-Klein-Gordon system admits a non-trivial Cauchy horizon $\mathcal{CH}^+$, along which $|\varpi|$ remains bounded. In this context, we now prove that such a solution further admits an $H^1_{\text{loc}}$ extension up to its Cauchy horizon, namely we show that all PDE variables of the IVP (except possibly for the ones corresponding, in a different coordinate system, to $\theta$ and $\kappa$) admit a continuous extension,  that the Christoffel symbols can be extended in $L^2_{\text{loc}}$ and that $\phi$ can be extended in $H^1_{\text{loc}}$.

Extensions in such a class represent the minimal requirement to understand the extended spacetimes as weak solutions to the Einstein equations.
We stress that we required enough regularity so that the second-order PDE system under consideration is equivalent to the first-order system \eqref{dur}--\eqref{algconstr} (see also remark \ref{rmk:equivalence}). It is interesting to notice that, if less regularity is demanded, it is a priori possible to construct metric extensions beyond the Cauchy horizon that solve the first-order system but not the second-order system of PDEs.

In the current case, the fact that our extensions solve the second-order system of PDEs in a weak sense beyond $\mathcal{CH}^+$ is a corollary of previous work \cite{GajicLuk}, see already remark \ref{remark:beyond}.

To construct the $H^1_{\text{loc}}$ extension, we first introduce the coordinate
\[
\tilde{v}(v) \coloneqq r(U, v_{\mathcal{A}}(U)) - r(U, v),
\]
where $v_{\mathcal{A}}(U) \coloneqq \max \{ v \colon \lambda(U, v) = 0\}$ (see also corollary \ref{corollary_apphorizon} for a localization and the main properties of the apparent horizon). We  provide a further $C^0$ extension in the $(u, \tilde{v})$ coordinate system, which will be used to construct an $H^1_{\text{loc}}$ extension. We will denote the solutions to our IVP in the coordinate system $(u, \tilde{v})$ by a tilde, e.g. $\tilde{r}(u, \tilde{v}) = r(u, v)$ and $\tilde{\lambda}(u, \tilde{v})  = \partial_{\tilde v} \tilde{r}(u, \tilde{v}) = - \frac{\lambda(u, v)}{\lambda(U, v)}$.
In this coordinate system, the Cauchy horizon can be expressed as a level set: $\mathcal{CH}^+ = \tilde{v}^{-1}(\tilde V)$, for some positive constant $\tilde{V}$. Moreover, we denote $\tilde{v}_0 \coloneqq \tilde{v}(v_0)$.

In order to establish an $H^1_{\text{loc}}$ extension, we use a further coordinate $\mathring{v}$ defined by the equality
\begin{equation} \label{def_ring}
\mathring{\Omega}^2 (U, \mathring{v}) \equiv 1,
\end{equation}
where the circle superscript is used here and in the following to denote PDE variables expressed in this newly-defined coordinate.
Expression \eqref{def_ring}, together with \eqref{def_kappa} and \eqref{algconstr}, entails, in particular, that
\[
\frac{d \mathring{v}}{dv} = \Omega^2(U, v).
\]
Similarly to the previous cases, we denote $\mathring{v}_0 \coloneqq \mathring{v}(v_0)$.

The coordinate systems $(u, \tilde{v})$  and $(u, \mathring{v})$ were used in \cite[section 12]{CGNS4}, where their $C^1$-compatibility up to the Cauchy horizon and the finiteness of $\mathring{V} \coloneqq \mathring{v}(\tilde{V})$ were shown.\footnote{In particular, every $C^0$ extension in the $(u, \tilde{v})$ coordinate system gives rise to a $C^0$ extension in the $(u, \mathring{v})$ coordinate system, and vice versa.}

\begin{theorem} \label{thm:noinflation_C0_extension}
Suppose that the assumptions of theorem \ref{thm:no_mass_inflation} hold. Then, for every $0 < \delta_U < U$, the functions $\tilde{\lambda}$, $\tilde{\varpi}$,  $\tilde{\phi}$, $\tilde{Q}$, $\tilde{r}$, $\tilde{\nu}$, $\tilde{\zeta}$, $\tilde{A}_u$ (but not necessarily $\tilde{\theta}$ or $\tilde{\kappa}$) admit a continuous extension to $[\delta_U, U] \times [0, \tilde{V}]$. 
\end{theorem}
\begin{proof}
The idea of the proof is to prove that any dynamical quantity $\tilde{q} \in \{\tilde{\lambda}$, $\tilde{\varpi}$, $\tilde{\phi}$, $\tilde{Q}$, $\tilde{r}$, $\tilde{\nu}$, $\tilde{\zeta}$, $\tilde{A}_u\}$ can be extended continuously, using the fact that (see also the proof of theorem \ref{thm:general_C0_extension}):
\[
\tilde{q}(u, \tilde{V}) \coloneqq \lim_{\tilde v \to \tilde V} \tilde{q}(u, \tilde{v})
\]
exists, is finite and that
\[
\left | \tilde{q}(u, \tilde{v})- \tilde{q}(u, \tilde{V}) \right| \underset{\tilde{v} \to \tilde{V}}{\longrightarrow} 0 \quad \text{uniformly in } [\delta_U, U].
\]
Alternatively, it is sufficient to show that $\tilde{q}$ can be expressed in terms of quantities that can be extended.

To show the above, we can use the estimates of proposition \ref{prop:lateblueshift} and follow the proof of  \cite[theorem 12.1]{CGNS4}, keeping in mind that, in our case, we need to control additional quantities such as $Q$ and $A_u$ due to the non-zero charge of the scalar field.
We refer the reader to \cite{CGNS4, mythesis} for the details.
\end{proof}

\begin{theorem}[An $H^1_{\text{loc}}$ extension] \label{thm:H1}
Suppose that the assumptions of theorem \ref{thm:no_mass_inflation} hold. Then, for every $0 < \delta_U < U$, there exists a continuous extension of the metric $g$ to the set $\mathcal{D} = [\delta_U, U] \times [\tilde{v}_0, \tilde{V}]$, with Christoffel symbols in $L^2(\mathcal{D})$ and such that $\phi \in H^1(\mathcal{D})$.
\end{theorem}
\begin{proof}
Theorem \ref{thm:noinflation_C0_extension} gives a continuous extension in the $(u, \tilde{v})$ coordinate system. Since $(u, \tilde{v})$ and $(u, \mathring{v})$ correspond to  $C^1$--compatible local charts (see \cite[section 12]{CGNS4}), a continuous extension in the latter coordinate system promptly follows. We will work in the $(u, \mathring{v})$ system throughout the remaining part of the proof.

The Christoffel symbols are given explicitly in appendix \ref{appendixUsefulExpressions}. Due to the fact that $\tilde r$ and $\tilde{\Omega}^2$ are strictly positive on the Cauchy horizon (see the proof of theorem \ref{thm:noinflation_C0_extension}), we only need to check that $\Gamma^u_{uu}$ and $\Gamma^{\mathring{v}}_{\mathring{v} \mathring{v}}$ are in $L^2_{\text{loc}}$. The proof of this result is completely analogous to the proof of theorem 12.3 in \cite{CGNS4}.

To show that $\| \phi \|_{H^1(\mathcal{D})}$ is bounded, with $\mathcal{D}=[\delta_U, U] \times [\tilde{v}_0, \tilde V]$, we notice two key results:
\begin{itemize}
\item Due to the continuous extensions of $\mathring{\phi}$, $\mathring{\zeta}$ and $\mathring{A}_u$, there exists $C > 0$, depending uniquely on the initial data, such that
\[
\int_{\mathring{v}_0}^{\mathring{V}} \br{|\phi|^2 + |\partial_u \phi|^2}(u, z) dz \le C,
\]
for every $u \in [\delta_U, U]$.
\item In the proof of theorem \ref{thm:no_mass_inflation} we showed that there exists $a > 0$ such that
\[
 \frac{|\theta|^2}{|\lambda|} (u, v) \lesssim \frac{e^{-av}}{\tilde{C}(u)}, \quad \forall\, (u, v) \in  J^+(\gamma),
\]
for some constant $\tilde{C}(u) > 0$ depending on $u$.
We obtained such a result using the lower bound \eqref{noinflation_lowlambda}. The same bound also gives
\[
\frac{|\theta|^2(u, v)}{|\lambda|(U, v)}  \lesssim \frac{e^{-av}}{\tilde{C}(U)}, \quad \forall\, (u, v) \in  J^+(\gamma).
\]
So, using the latter, using the fact that:\footnote{We notice that $|\tilde{\lambda}|(U, \cdot) \equiv 1$ and that $\tilde{\Omega}^2$ is positive up to the Cauchy horizon, see the proof of theorem \ref{thm:noinflation_C0_extension}.}
\[
\frac{\Omega^2}{ |\lambda|} (U, v) = 4\frac{\nu}{1-\mu}(U, v) = 4 \frac{\tilde{\nu}}{1- \tilde{\mu}}(U, \tilde v(v)) = \tilde{\Omega}^2(U, \tilde{v}(v)) \ge C > 0, \quad \forall\, v \ge v_0,
\]
and, finally, exploiting the relation $\mathring{\theta}(u, \mathring{v}) = \Omega^{-2}(U, v(\mathring{v}))\theta(u, v(\mathring{v}))$, we have:
\[
\int_{\mathring{v}_0}^{\mathring{V}} |\mathring{\theta}|^2(u, z) dz = \int_{v_0}^{+\infty} \frac{|\theta|^2(u, z)}{\Omega^2(U, z)}dz \lesssim \int_{v_0}^{+\infty} \frac{|\theta|^2(u, z)}{|\lambda|(U, z)} dz \le C,
\]
for some $C>0$ depending uniquely on the initial data. \qedhere
\end{itemize} 
\end{proof}

\begin{corollary}
Let $M_0$ be the preimage of $(0, U] \times [\mathring{v}_0, \mathring{V}]$ via the coordinate functions $(u, \mathring{v})$. Then, the pull-backed scalar field and Christoffel symbols on $M_0$ are, respectively, in $H^1_{\text{loc}}(M_0)$ and $L^2_{\text{loc}}(M_0)$.
\end{corollary}
\begin{proof}
For every $0 < \delta_U < U$, we denote by $\mathcal{M}_{\delta}$ the preimage of $[\delta_U, U]\times [\mathring{v}_0, \mathring{V}]$ via the coordinate functions $(u, \mathring{v})$. The volume form on $\mathcal{M}_{\delta}$ is 
\[
dV_4 = \frac{\mathring{\Omega}^2}{2} \mathring{r}^2 du d \mathring{v}.
\]
Since $\mathring{r}$ and $\mathring{\Omega}^2$ are continuous in the compact set $[\delta_U, U]\times [\mathring{v}_0, \mathring{V}]$, the results of theorem \ref{thm:H1} are sufficient to conclude the proof. 
\end{proof}

\begin{remark}[Extensions beyond the Cauchy horizon as weak solutions] \label{remark:beyond}
In the asymptotically flat, extremal case studied in \cite[sec. 10]{GajicLuk}, it was explicitly shown how to construct $H^1$ extensions solving (in a weak sense) the second-order system of PDEs beyond the Cauchy horizon $\mathcal{CH}^+$. Indeed, notice that the blueshift instability described in Fig. \ref{fig:blueshiftinst} becomes weaker as the surface gravity of the Cauchy horizon approaches zero, i.e. as the extremal limit is reached.

The procedure exploits a well-posedness result, under spherical symmetry, that requires less regularity for the initial data  compared to our local existence theorem \ref{thm:localexistence}. The result then follows after considering a sequence of characteristic initial value problems approaching the Cauchy horizon and using the fact that the time of existence for the respective solutions does not depend on the location of such IVP in the sequence.

Due to the freedom in defining the main extendable PDE variables beyond the Cauchy horizon, the technique reveals the existence of infinitely many weak solutions beyond $\mathcal{CH}^+$. 

We stress that the same procedure can be applied in our case, since the presence of the cosmological constant $\Lambda$ in our second-order system \eqref{raych_u}--\eqref{klein_gordon}, \eqref{duQ}--\eqref{dvQ} does not represent an obstruction to the reasoning of \cite{GajicLuk}, and neither does the sub-extremality condition.
\end{remark}

\appendix

\section{Useful expressions} \label{appendixUsefulExpressions}
Let us consider the metric for a spherically symmetric spacetime. It is always possible to locally express such a metric using null coordinates $(u, v)$ as
\begin{equation}
g=-\Omega^2(u, v)du dv +r^2(u, v)\sigma_{\bbS^2}.
\end{equation}
The respective Christoffel symbols are
\begin{align*}
&\Gamma^{u}_{uu}=\partial_u \log \Omega^2, &&\Gamma^{u}_{\theta \theta}=\frac{\Gamma^{u}_{\varphi \varphi}}{\sin^2 \theta} = \frac{2r}{\Omega^2}\partial_v r,\\
&\Gamma^{v}_{vv} = \partial_v \log \Omega^2, &&\Gamma^v_{\theta \theta} = \frac{\Gamma^{v}_{\varphi \varphi}}{\sin^2 \theta} = \frac{2r}{\Omega^2}\partial_u r,\\
&\Gamma^{\theta}_{u \theta}=\Gamma^{\varphi}_{u \varphi} = \frac{1}{r} \partial_u r, && \Gamma^{\theta}_{v \theta} = \Gamma^{\varphi}_{v \varphi} = \frac{1}{r}\partial_v r\\
& \Gamma^{\theta}_{\varphi \varphi} = -\cos \theta \sin \theta, && \Gamma^{\varphi}_{\theta \varphi}= \cot \theta,
\end{align*}
where the remaining symbols are either zero or can be obtained by symmetries.  We used the xAct package suite \cite{xAct} in Mathematica for the following computations. The non-zero components of the Ricci tensor are:
\begin{align*}
R_{uu} &= \frac{2}{r\Omega}\left(2 \partial_u r \partial_u \Omega -\Omega \partial_u^2 r \right),\\
R_{uv} &= R_{vu}= -\frac{2 \partial_u \partial_v r}{r} + \frac{2}{\Omega^2} \left( \partial_v \Omega \partial_u \Omega - \Omega \partial_u \partial_v \Omega \right), \\
R_{vv} &= \frac{2}{r\Omega}\left( 2\partial_v r \partial_v \Omega -  \Omega \partial_v^2 r \right), \\
R_{\theta \theta} &= 1 + \frac{4}{\Omega^2} \left( \partial_v r \partial_u r + r\partial_u \partial_v r \right), \\
R_{\varphi \varphi} &= \sin^2 \theta R_{\theta \theta},
\end{align*}
and the Ricci scalar is given by
\[
R = \frac{2}{r^2} - \frac{8 \partial_v \Omega \partial_u \Omega}{\Omega^4} + \frac{8}{r^2 \Omega^2} \left( \partial_v r \partial_u r + 2 r \partial_u \partial_v r \right) + \frac{8}{\Omega^3} \partial_u \partial_v \Omega.
\]
Finally, the non-zero components of the Einstein tensor are
\begin{align*}
G_{uu} &= \frac{4 \partial_u r \partial_u \Omega}{r \Omega} - \frac{2}{r}\partial_u^2 r,\\
G_{uv} &= G_{vu} = \frac{\Omega^2}{2r^2} + \frac{2 \partial_v r \partial_u r}{r^2} + \frac{2 \partial_u \partial_v r}{r}, \\
G_{vv} &= \frac{4}{r\Omega} \partial_v r \partial_v \Omega - \frac{2}{r} \partial_v^2 r, \\
G_{\theta \theta} &= -\frac{4r}{\Omega^2} \partial_u \partial_v r + \frac{4r^2}{\Omega^4} \partial_v \Omega \partial_u \Omega - \frac{4r^2}{\Omega^3} \partial_u \partial_v \Omega, \\
G_{\varphi \varphi} &= \sin^2 \theta G_{\theta \theta}.
\end{align*}
To study the system of equations \eqref{main_system}, we also report the components of the energy-momentum tensor $\Tem$:
\begin{align*}
&\Tem_{uu}=\Tem_{vv}=0,\\
&\Tem_{uv} = \frac{\Omega^2 Q^2}{4r^4},\\
&\Tem_{\theta \theta} = \frac{\Tem_{\varphi \varphi}}{\sin^2 \theta}=\frac{Q^2}{2r^2},
\end{align*}
and of $\Tphi$:
\begin{align*}
&\Tphi_{uu}= |D_u \phi|^2, \quad &&\Tphi_{vv}=|D_v \phi|^2,\\
&\Tphi_{uv} = \frac{m^2 \Omega^2}{4}|\phi|^2, \quad &&\Tphi_{\theta \theta} = \frac{\Tphi_{\varphi \varphi}}{\sin^2 \theta}= \frac{2r^2}{\Omega^2}\Re{\br{D_u \phi \overline{D_v \phi}}} - \frac{1}{2}m^2 r^2 |\phi|^2.
\end{align*}
A short computation also gives (see \eqref{noethercurr} and \eqref{fieldtensor}):
\[
\star F = -Q(u, v) \sin \theta d\theta \wedge d \varphi
\]
and 
\[
\star J = q r^2 \sin \theta \br{\Im \br{\phi \overline{D_u \phi}} du \wedge d\theta \wedge d\varphi + \Im \br{\phi \overline{D_v \phi}} dv \wedge d\theta \wedge d\varphi }.
\]

\section{The Reissner-Nordstr{\"o}m-de Sitter solution and the choice of coordinates} \label{app:eddfinkel}
In this short section we review the main properties of spherically symmetric charged black holes. At the same time, we establish the relation between the null coordinates employed in this work and some Eddington-Finkelstein coordinates commonly adopted for the Reissner-Nordstr{\"o}m-de Sitter solution.

The (sub-extremal, or non-degenerate) Reissner-Nordstr{\"o}m-de Sitter metric is the unique static, spherically symmetric solution of the electro-vacuum Einstein equations
\[
\text{Ric}_{\mu \nu}(g) -\frac{1}{2} R(g) g_{\mu \nu} + \Lambda g_{\mu \nu} = 2g^{\alpha \beta}F_{\alpha \mu}F_{\beta \nu} - \frac12 g_{\mu \nu} F^{\alpha \beta}F_{\alpha \beta},
\]
for a positive cosmological constant $\Lambda$. Here, the strength field tensor $F$ satisfies the Maxwell equations $dF = 0$ and $d \star F = 0$. This solution is denoted by the three real constants $\varpi_+ > 0$ (mass parameter), $Q_+$ (charge parameter) and $\Lambda$.\footnote{This notation for the mass and charge is not standard, but we adopt it here to draw a comparison with our work on the Einstein-Maxwell-charged-Klein-Gordon system.} To such values, we associate the function
\[
1-f(r)\coloneqq 1-\frac{2\varpi_+}{r}+\frac{Q_+^2}{r^2}-\frac{\Lambda}{3}r^2,
\]
required to have three distinct real roots $r_{-}$, $r_+$ and $r_C$, with $r_{-} < r_+ < r_C$ (\textbf{sub-extremality assumption}). In general, $1-f$ may possess a fourth distinct (negative) root, which we denote by $r_n$.
 The Reissner-Nordstr{\"o}m-de Sitter metric, able to describe a black hole spacetime
 in an expanding cosmological scenario, can be written in coordinates $(t, r) \in \bbR \times (r_{-}, r_+)$ as
\begin{equation} \label{RNdS_metric}
g_{\text{RNdS}} = -\br{1- f(r)}dt^2 + \frac{1}{1-f(r)}dr^2 + r^2 \sigma_{S^2},
\end{equation}
where $\sigma_{S^2}$ is the standard metric on the unit round sphere.
We refer the reader to \cite{CostaFranzen, CostaGirao} for an overview on this spacetime, on the different coordinate systems to resolve the coordinate singularities at $r=r_+$ (corresponding to the event horizon $\mathcal{H}^+_{\text{RNdS}}$) and $r=r_{-}$ (corresponding to the Cauchy horizon $\mathcal{CH}^+_{\text{RNdS}}$), and on the failure of determinism inherent in this non-generic solution.

We now show that the outgoing null coordinate chosen in section \ref{section:assumptions}  for our Einstein-Maxwell-charged-Klein-Gordon system is strictly related to the outgoing Eddington-Finkelstein coordinate
\begin{equation} \label{eddfink}
\tilde{v} =  \frac{r^* + t}{2}, \quad \text{with } r^* \coloneqq \int \frac{dr}{1-f(r)},
\end{equation}
that can be used to describe regions of the Reissner-Nordstr{\"o}m-de Sitter spacetime close to the event horizon $\mathcal{H}_{\text{RNdS}}^+$. As we will see, this result implies that \textbf{the $\bm{v}$ coordinate  in the exponential Price law upper bound \eqref{assumption:expPricelaw} is half of the outgoing Eddington-Finkelstein coordinate adopted for the linear analyses \cite{CostaFranzen} and \cite{HintzVasy1}}. This is of particular interest in the cosmological scenario, since different coordinate choices could lead to different exponential behaviours of the scalar field along the event horizon $\mathcal{H}^+$ of the dynamical black hole.

In particular, using the Eddington-Finkelstein coordinates
\[
\tilde{u} = \frac{r^* -t}{2}, \qquad
\tilde{v} = \frac{r^* + t}{2},
\]
and then the coordinate
\[
u = \frac{1}{2K_+}e^{2K_+ \tilde u},
\]
where $K_+$ is the surface gravity of $\mathcal{H}^+_{\text{RNdS}}$, 
we can write \eqref{RNdS_metric} as
\begin{align*}
g_{\text{RNdS}} &= 4\big ( 1-f(r(\tilde u, \tilde v)) \big ) d\tilde u d \tilde v + r^2(\tilde u, \tilde v) \sigma_{S^2}\\
& = 4 e^{-2K_+ \tilde u(u)}\big (1- f(r(u, \tilde v)) \big) du d \tilde v + r^2( u, \tilde v) \sigma_{S^2} \\
&\eqqcolon  -\Omega^2_{\text{RNdS}}(u, \tilde v)^2 du d \tilde v + r^2( u, \tilde v) \sigma_{S^2}.
\end{align*}
Notice that $1-f(r(u, \tilde v))$ is non-positive for $r(u, \tilde v) \in [r_{-}, r_+]$ (see \cite[section 3]{CGNS2}), and  that  $\mathcal{H}_{\text{RNdS}}^+=$ $\{\tilde u = -\infty\} = \{u = 0\}$.
We observe that:
\[
\tilde v + \tilde u = r^* = \frac{1}{2K_+} \log |r - r_+| + O(1),
\]
for $r$ close to $r_+$.
So, by decomposing $1-f$ into the product of its roots, we have
\[
1-f(r) \sim -e^{2K_+(\tilde v+ \tilde{u})} , \quad \text{ as } r \nearrow r_+.
\]
Therefore:
\begin{equation} \label{asymp1}
\Omega^2_{\text{RNdS}} \sim e^{2K_+ \tilde v}, \quad \text{ as } r \nearrow r_+. 
\end{equation}
We now compare this expression with \eqref{metric_eqn}. Assumption \textbf{(A)} in section \ref{section:assumptions}, \eqref{def_kappa} and \eqref{bound_nu_eventhorizon} give
\begin{equation} \label{asymp2}
\Omega^2(0, v) = -4 \nu(0, v) \sim e^{2K_+ v}, \quad \forall\, v \ge v_0.
\end{equation}
This result holds for all initial data satisfying the assumptions of section \ref{section:assumptions} and also for Reissner-Nordstr{\"o}m-de Sitter initial data\footnote{Namely the case in which $\varpi_0 \equiv \varpi_+$, $r(0, v) \equiv r_+$, $Q_0 \equiv Q_+$, $\lambda_0 \equiv 0$, $\kappa_0 \equiv 1$ and $\Omega^2(0, v) = 4 e^{2K_+ v}$. See also \cite{CGNS1}.} for the characteristic IVP of section \ref{section:ivp}. Indeed, the proof of proposition \ref{prop: boundsEH} can be easily extended to the case $\lambda|_{\mathcal{H}} \equiv 0$. In particular, we have $v =\tilde v$.

\section{Proof of lemma \ref{lemma:lambdanu_blueshift}} \label{appendixProofGamma}
In the following, we give a proof of lemma \ref{lemma:lambdanu_blueshift} (see also \cite[lemma 6.6]{CGNS2}).
\begin{proof}
Let $(u, v) \in \gamma \cap \{v \ge v_1 \}$.
To prove the bounds for $|\lambda|$, some preliminary estimates will be useful. First, we obtain an upper bound on $\frac{\nu}{1-\mu}$ by exploiting \eqref{raych_recast}:
\begin{equation} \label{preliminary_lambda_blueshift1}
0 <  \frac{\nu}{1-\mu}(u, v) \le \frac{\nu}{1-\mu}(u, v_Y(u)).
\end{equation} 

Furthermore, by differentiating the relation $r(\cdot, v_Y(\cdot))=Y$, we note that
\[
\nu(\cdot, v_Y(\cdot)) + \lambda(\cdot, v_Y(\cdot)) v'_Y(\cdot) \equiv 0 \quad \text{ in } [0, u_Y(v_1)].
\]
If we integrate the latter in $[u_Y(v), u]$, use that $u_Y = v_Y^{-1}$ and perform a change of variables:
\begin{equation} \label{mini_BV}
\int_{u_Y(v)}^u \nu(\tilde u, v_Y( \tilde u))d \tilde u = \int_{v_Y(u)}^v \lambda(u_Y(\tilde v), \tilde v) d\tilde v.
\end{equation}

Now, we obtain integral bounds on $\frac{\nu}{1-\mu}$ as a preliminary step. An upper bound is achieved as follows. We use \eqref{preliminary_lambda_blueshift1}, 
together with \eqref{bound_gammaY},  \eqref{mini_BV}, \eqref{algconstr} and with the definition of the curve $\gamma$, in this precise order, to obtain
\begin{align}
0 &< \int_{u_Y(v)}^u \frac{\nu}{1-\mu}(u', v) du' \le \int_{u_Y(v)}^u \frac{\nu}{1-\mu}(\tilde u, v_Y(\tilde u)) d \tilde u \nonumber \\
&\le  \frac{1}{C(N_{\text{i.d.}}, R, Y) + o(1)} \int_{u_Y(v)}^u (-\nu)(\tilde u, v_Y(\tilde u)) d \tilde u = \frac{1}{C(N_{\text{i.d.}}, R, Y) + o(1)} \int_{v_Y(u)}^v(-\lambda)(u_Y(\tilde v), \tilde v) d \tilde v \nonumber \\
&\le \frac{C(N_{\text{i.d.}}, R, Y) + o(1)}{C(N_{\text{i.d.}}, R, Y) + o(1)} \int_{v_Y(u)}^v \kappa(u_Y(\tilde v), \tilde v) d \tilde v \le (1+o(1)) \frac{\beta}{1+\beta} v, \label{preliminary_lambda_blueshift2}
\end{align}
where the notation $o(1)$ refers to the limit $u \to 0$.
\begin{figure}[H]
\centering
\includegraphics[width=0.5\textwidth]{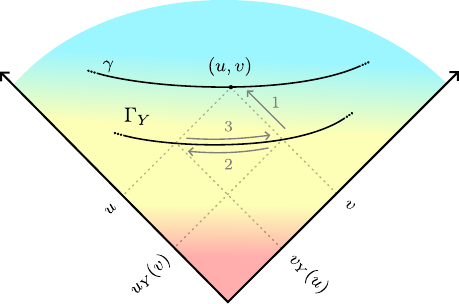}
\end{figure}
This yields a lower bound on $|\lambda|$ as follows. 
We use \eqref{algconstr} to cast \eqref{dulambda_kappa} as
\begin{equation} \label{blueshift_lambda_eqn}
\partial_u \lambda = \lambda \frac{\nu}{1-\mu}(2K - m^2 r |\phi|^2).
\end{equation}
If we integrate the latter from $\Gamma_Y$ and use that $|\lambda|(u_Y(v), v) = [\kappa(1-\mu)](u_Y(v), v)$, together with the boundedness of $r$, \eqref{phi_blueshift_final}, \eqref{bound_2K_blueshift}, \eqref{bound_gammaY},  \eqref{preliminary_lambda_blueshift2} and the fact that $\kappa$ is close to 1 in the early blueshift region:
\begin{align}
|\lambda|(u, v)  &= |\lambda|(u_Y(v), v)\exp\br{\int_{u_Y(v)}^u \frac{\nu}{1-\mu}\br{2K - m^2 r |\phi|^2}(u', v)du'} \nonumber \\
& \ge C(N_{\text{i.d.}}, R, Y) \exp \br{-(2K_{-} +  C(N_{\text{i.d.}}) \varepsilon ) (1+o(1))\frac{\beta}{1+\beta}v} \ \nonumber \\
&\ge C(N_{\text{i.d.}}, R, Y) e^{-(2K_{-} + c_1 \varepsilon)  \frac{\beta}{1+\beta}v}  \label{lowdecay_lambda_gamma} 
\end{align}
for $U$ sufficiently small,
where the constant $c_1$ depends uniquely on the initial data.

Now, we imitate the previous steps to get an upper bound on $|\lambda|$. 
 Since $\lambda < 0$ and $\partial_u |\lambda| < 0$ in this region (see \eqref{dulambda_kappa} and  \eqref{blueshift_final_K}), we use \eqref{phi_blueshift_final}, \eqref{lowdecay_lambda_gamma} and the fact that $v = (1+\beta)v_Y(u)$ to write:
\begin{align*}
\int_{v_Y(u)}^v \frac{|\partial_v \phi|^2}{|\lambda|} (u, v') dv' &\le \int_{v_Y(u)}^v \frac{|\partial_v \phi|^2(u, v')}{|\lambda| (u_{\gamma}(v'), v')} dv' \\
&\lesssim \int_{v_Y(u)}^v e^{-2(c(s)-\tau)v' + \br{2K_{-} + c_1 \varepsilon}\frac{\beta}{1+\beta}v'} dv' \\
&\lesssim e^{-a v_Y(u)} = e^{-\frac{a}{1+\beta}v},
\end{align*}
for some $a > 0$ and for $\beta$ small.
 In particular, provided that $v_1$ is chosen sufficiently large,  a reasoning similar to that of the steps  \eqref{raych_recast}--\eqref{noshift_nu} leads to
\begin{equation} \label{preliminary_lambda_blueshift1_low}
\frac{\nu}{1-\mu}(u, v) \ge \frac{1}{(1+ o(1))} \cdot \frac{\nu}{1-\mu} (u, v_Y(u)),
\end{equation}
where the notation $o(1)$ refers to the limit $u \to 0$.

Analogously to the steps in \eqref{preliminary_lambda_blueshift2}, but now using \eqref{preliminary_lambda_blueshift1_low} with \eqref{bound_gammaY}, \eqref{mini_BV}, \eqref{algconstr}, with the fact that $\kappa$ is close to 1 in the early blueshift region, and with the definition of the curve $\gamma$, in this precise order, we obtain:
\begin{equation}
 \int_{u_Y(v)}^u \frac{\nu}{1-\mu}(u', v) du' \ge(1+o(1)) \frac{\beta}{1+\beta}v, \label{preliminary_lambda_blueshift3}
\end{equation}
for $(u, v) \in \gamma \cap \{ v \ge v_1 \}$.
We now integrate \eqref{blueshift_lambda_eqn} from $\Gamma_Y$ and use that $|\lambda|(u_Y(v), v) = [\kappa(1-\mu)](u_Y(v), v)$, together with the boundedness of $r$, \eqref{phi_blueshift_final}, \eqref{bound_2K_blueshift} and   \eqref{preliminary_lambda_blueshift3}:
\begin{align}
|\lambda|(u, v)  &= |\lambda|(u_Y(v), v)\exp\br{\int_{u_Y(v)}^u \frac{\nu}{1-\mu}\br{2K - m^2 r |\phi|^2}(u', v)du'} \nonumber  \\
&\le C(N_{\text{i.d.}}, R, Y) \exp \br{-\br{2K_{-} - C(N_{\text{i.d.}}) \varepsilon}(1+o(1)) \frac{\beta}{1+\beta} v } \nonumber \\
&\le C(N_{\text{i.d.}}, R, Y) e^{-(2K_{-} - c_2 \varepsilon)\frac{\beta}{1+\beta}v} , \label{uppdecay_lambda_gamma}
\end{align}
where  $c_2 > 0$ depends uniquely on the initial data and can be chosen so that $2K_{-} - c_2 \varepsilon > 0$.

We now proceed to estimate $|\nu|$ along $\gamma$. First, by  relation  \eqref{noshift_u_final} proved in the no-shift region, we have
\begin{equation} \label{old_estimate_u}
 u \sim e^{-2K_+ v_Y(u)}, \quad \forall\,  u \in [0, u_Y(v_1)].
\end{equation}
Now, for every $0 < \sigma < 2 K_{-}$ and for every $u \le u_Y(v_1)$:
\begin{align} 
e^{-\br{2  K_{-} \pm \sigma}\beta v_Y(u)} &\sim u^{\frac{2 K_{-} \pm \sigma}{2K_+} \beta} = u^{ \frac{K_{-}}{K_+}\beta \pm \alpha}, \label{preliminary_estimate_u}
\end{align}
with $0 < \alpha= \frac{\sigma \beta}{2K_+}  \lesssim 1$. 

We now use, in this order, \eqref{blueshift_nu_explicit}, \eqref{noshift_nu} (from the no-shift region), \eqref{kappaK_lowerbound_blueshift},  the definition \eqref{def_gamma} of $\gamma$,  \eqref{redshift_nu_final}, \eqref{redshift_uomega} and \eqref{preliminary_estimate_u} to obtain:
\begin{align*}
|\nu|(u, v) &= |\nu|(u, v_Y(u))\exp\br{\int_{v_Y(u)}^v \kappa(2K - m^2 r |\phi|^2)(u, v') dv'} \\
&\le C(N_{\text{i.d.}}, R, Y) |\nu|(u, v_R(u)) e^{-(2K_{-}-2C\varepsilon)\beta v_Y(u)} \\
&\le C(N_{\text{i.d.}}, R, Y) \frac{1}{u} e^{-(2K_{-}-2C\varepsilon)\beta v_Y(u)} \\
&\lesssim u^{-1 + \frac{K_{-}}{K_+}\beta -\alpha},
\end{align*}
for every $(u, v) \in \gamma \cap \{ v \ge v_1 \}$,
where  $0 < \alpha \lesssim \varepsilon \beta$. In particular, for every $\beta$, we can assume without loss of generality that $\varepsilon$ was chosen sufficiently small so that
\begin{equation} \label{rhobeta}
\frac{K_{-}}{K_{+}}\beta - \alpha > 0.
\end{equation}

The lower bound in \eqref{bound_nu_gamma} then follows analogously to the previous steps, after recalling that $v_Y(u)-v_R(u)$ is bounded by a constant (see \eqref{bounddifferencev}).
\end{proof}

\printbibliography[heading=bibintoc]

\end{document}